\documentclass[12pt,reqno,a4paper]{article}
\usepackage[utf8]{inputenc}
\usepackage[T1]{fontenc}
\usepackage{enumerate}
\usepackage{amsmath,amsthm,amsfonts,amscd,amsbsy,amsthm,bm,cancel,mathtools,physics,comment,eucal,latexsym,amssymb,mathrsfs,siunitx}
\usepackage{enumerate,cases,slashed}
\usepackage{accents}
\usepackage{color}
\usepackage{float}
\usepackage{relsize}
\usepackage[normalem]{ulem}
\usepackage{pst-node}
\usepackage{epsfig,caption,float,tikz,subcaption,graphicx}  
\usepackage{bm,cite}
\usepackage{simplewick}
\usepackage{simpler-wick}
\usepackage{indentfirst}
\usepackage{tensor}

\usepackage{enumitem}
\setlist[itemize,1]{label={\tiny\textbullet}}

\definecolor{darkgreen}{rgb}{0,0.5,0}
\definecolor{darkred}{rgb}{0.7,0,0}
\definecolor{darkblue}{rgb}{0,.2,.7}
\usepackage[english]{babel}
\usepackage[margin=2cm]{geometry}

\usepackage{url}

\usepackage{cite}

\usepackage{titlesec}
\titleformat{\section}{\large\bfseries\filcenter}{\thesection}{1em}{}
\titleformat{\subsection}{\bfseries}{\thesubsection}{1em}{}

\newtheorem{theorem}{Theorem}[section]
\newtheorem{corollary}[theorem]{Corollary}
\newtheorem{lemma}[theorem]{Lemma}
\newtheorem{proposition}[theorem]{Proposition}

\theoremstyle{remark}

\theoremstyle{definition}
\newtheorem{remark}[theorem]{Remark}

\numberwithin{equation}{section}
\allowdisplaybreaks[1]

\catcode`@=12

\usepackage[colorlinks, 
citecolor=black, linkcolor=black, urlcolor=darkblue, pdfpagelabels=true, unicode=true,
]{hyperref}

\hypersetup{
 pdfauthor={Paolo Meda, Simone Murro and Nicola Pinamonti},
 pdftitle={Semiclassical Einstein equations},
 pdflang={en}
}

\newcommand{\doi}[1]{\href{https://www.doi.org/#1}{[#1]}}

\def\ab{{{ab}}}
\newcommand{\omegazd}{{\omega_0}_2}
\newcommand{\omegaz}{\omega_0}
\newcommand{\omegaud}{{\omega_1}_2}
\newcommand{\omegau}{\omega_1}

\def\cc{{\mathrm{cc}}}

\def\pc{\mathrm{pc}}

\def\g{{\mathsf g}}

\def\cA{{\mathcal A}}

\def\cD{{\mathcal D}}

\def\cF{{\mathcal F}}

\def\cM{{\mathcal M}}

\def\cS{{\mathcal S}}
\def\cT{{\mathcal T}}

\def\cK{{\mathcal K}}

\def\mS{\mathscr{S}}

\def\mS{\mathscr{S}}

\def\RR{{\mathbb R}} 
\def\CC{{\mathbb C}}

\def\beq{\begin{eqnarray}}
\def\eeq{\end{eqnarray}}
\def\pa{\partial}
               		   
\def\at{\left(}               		   
\def\aq{\left[}               		   
\def\ag{\left\{}

\def\ct{\right)}              		   
\def\cq{\right]}              		
\def\cg{\right\}}

\renewcommand{\vec}[1]{\mathbf{#1}}
\newcommand{\Wick}[1]{{:}#1{:}}

\newcommand{\Eq}{Eq.\@ }
\newcommand{\Eqs}{Eqs.\@ }

\begin{document}

\begin{flushright}

\baselineskip=4pt

\end{flushright}
\begin{center}
\vspace{5mm}

{\Large\bf The Semiclassical Einstein-Klein-Gordon System: \\[3mm] Asymptotic Analysis of Minkowski Spacetime\\[3mm] } 

\vspace{5mm}

{\bf by}

\vspace{5mm}
{  \bf Stefano Galanda$^1$, Paolo Meda$^2$, Simone Murro$^3$, \\[3mm]
Nicola Pinamonti$^3$, Gabriel Schmid$^4$}\\[5mm]
\noindent  {\it $^1$Department of Mathematics, University of York, 
York YO10 5GH, United Kingdom}\\[1mm]
\noindent  {\it $^2$Dipartimento di Matematica, Università di Trento \& GNFM-INdAM \& TIFPA-INFN, Italy}\\[1mm]
\noindent  {\it $^3$Dipartimento di Matematica, Università di Genova \& GNFM-INdAM \& INFN, Italy}\\[1mm]
\noindent  {\it $^4$Fakult\"at f\"ur Mathematik, Universit\"at Regensburg, 93053 Regensburg, Germany}\\[5mm]
email: \ {\tt  stefano.galanda@york.ac.uk, paolo.meda@unitn.it, simone.murro@unige.it, nicola.pinamonti@unige.it, gabriel.schmid@ur.de}\\[10mm]
\end{center}

\bigskip

\noindent 
\small 
\begin{abstract}
We establish the linear instability of the semiclassical Einstein–Klein–Gordon system linearised about the Minkowski vacuum spacetime. The proof relies on formulating a forcing problem for both metric and state perturbations within the space of past-compact sections. This geometric framework admits a unique tensor decomposition which, in conjunction with the quantum M\o ller operator, enables the decoupling of the linearised system into two distinct Cauchy problems. Consequently, the metric perturbations are shown to be governed by a higher-order, nonlocal hyperbolic partial differential equation. By relegating the nonlocal contributions to subleading order, we establish the well-posedness of this forcing problem. Furthermore, we provide a rigorous asymptotic analysis for physically admissible choices of the renormalisation constants. We prove that the system exhibits a late-time linear instability: the metric perturbations grow exponentially, bounded strictly by a universal scale $H$, thereby indicating a quantum backreaction-driven transition toward a de Sitter cosmological spacetime. Provided the parameters governing the system are restricted to a physically relevant regime,
this universal scale is compatible with the measured expansion of our universe.
\end{abstract}

\vspace{2em}

\begin{flushright}
\small
\textit{Dedicated to the memory of Manuele Filaci
}
\end{flushright}

\vspace{2em}

\bigskip

\tableofcontents

\section{Introduction}

While a rigorous mathematical quantisation of gravitational interaction remains an open problem, the study of quantum field theory on fixed Lorentzian manifolds has yielded profound structural insights, providing the mathematical foundation for phenomena such as the Hawking effect for spacetime containing black holes~\cite{Hawking} and 
 particle creation in the context of cosmological backgrounds
~\cite{Parker}. A fundamental extension of this paradigm consists in coupling the quantum matter to the background spacetime geometry via the {\bf semiclassical Einstein equations}, proposed independently by M\o ller and Rosenfeld in \cite{Moller,Rosenfeld} in the early 1960s. This framework yields a highly non-linear system of partial differential equations wherein the geometry, represented by the Einstein tensor $G_{ab}$ of a Lorentzian metric $g_{ab}$ on a smooth manifold $\mathcal{M}$, is driven by the renormalised expectation value of the stress-energy tensor associated to a quantum field propagating on $(\mathcal{M}, g_\ab)$. Formally, the system is given by
$$\begin{cases}
G_{ab}[g_\ab] = \kappa \langle \colon T_{ab} \colon \rangle_{\omega} \\
\mathcal{P}_{x} \, \omega(x,y)= 0 \\
\mathcal{P}_{y} \, \omega(x,y)= 0,
\end{cases} $$where $\kappa = 8\pi G$ denotes the gravitational coupling constant, $\mathcal{P}$ is a normally hyperbolic differential operator acting on $(\mathcal{M}, g_\ab)$ that encodes the dynamics of the underlying quantum field, and $\omega \in \mathcal{D}'(\mathcal{M} \times \mathcal{M})$ represents the two-point correlation distribution of a quasifree algebraic state. The quantity $\langle \colon T_{ab} \colon \rangle_{\omega}$ denotes the properly renormalised expected stress-energy tensor evaluated in the state $\omega$. For this coupled system to admit a well-posed initial value formulation and to carry physical significance, the pair $\Big((\mathcal{M},g_\ab), \omega\Big)$ is required to satisfy two stringent admissibility criteria:
\begin{enumerate}
\item[(I)] The spacetime $(\mathcal{M},g_\ab)$ must be {\bf globally hyperbolic}, i.e.~it does not possess closed causal curves and, 
for all points $p,q\in \cM$, the intersection between the future lightcone of $p$ with the past lightcone of $q$ is compact.
\item[(II)] The state $\omega$, realised as a normalised, positive linear functional on the canonical algebra of observables, must satisfy the {\bf Hadamard condition},  i.e.~the primed wavefront set of the Schwartz integral kernel of the two-point correlation distribution should be contained in the Cartesian product of the upper null cones.
\end{enumerate}

The rigorous analysis of the global and asymptotic behaviour of solutions to the semiclassical Einstein equations is conjectured to provide the mathematical scaffolding governing early-universe cosmology and the dynamical evaporation of trapped surfaces. In this paper, we aim to study the Cauchy problem for the semiclassical Einstein--Klein--Gordon system, linearised around the Minkowski background. In particular, we shall prove that the backreaction of the massive scalar field dynamically generates an effective de Sitter cosmological expansion. Contrary to previous results, the Hubble parameter, which measures the effective expansion, is compatible with physical observations.

\subsection{The quantisation of a Klein-Gordon field}

To rigorously define the expectation value of the renormalised stress-energy tensor for the Klein-Gordon field, we must briefly review the quantisation of scalar fields on globally hyperbolic spacetimes.
We adopt the {\bf algebraic approach to quantum field theory}, a mathematical framework for quantisation on Lorentzian manifolds. We refer the reader to textbooks
on the subject, e.g.~\cite{gerardbook,Rejzner,AdvancesAQFT,Hack,Waldbook}. This paradigm decomposes quantisation into a two-step algebraic procedure:
\begin{enumerate}
\item[(1)] The assignment of a unital $*$-algebra of observables $\mathcal{A}$ to the physical system, encoding the causal structure, dynamics, and canonical commutation relations.
\item[(2)] The specification of a physical state $\omega \colon \mathcal{A} \to \mathbb{C}$, defined as a normalised, positive linear functional satisfying the Hadamard condition.\end{enumerate}

Consider a globally hyperbolic manifold $(\cM,g_\ab)$. As shown by Bernal and Sanchez in~\cite{BernSanch1,BernSanch2}, global hyperbolicity is tantamount to the existence of an isometry $$\Xi: (\cM,g_{ab}) \to (\RR \times \Sigma,- \beta^2 d \tau\otimes d\tau +  \gamma_{ij}(\tau) ),
$$
 	with the time-orientation on the right-hand side induced from $(\cM,g_{ab})$ through $\Xi$.
 Here, $\tau$ is the canonical projection  $\RR\times \Sigma \to  \RR$ and the following holds true:
 \begin{itemize}
\item[(a)]  $\nabla \tau := \sharp d\tau$ is  past-directed timelike, 
\item[(b)]	 $\beta \in C^\infty(\RR \times \Sigma, (0,+\infty))$;
\item[(c)] $ \gamma_{ij}(\tau)$ is a smooth Riemannian metric on each leaf $\{\tau\} \times \Sigma$, $\tau\in \RR$, depending smoothly on $\tau$;
\item[(d)]  every embedded co-dimension-$1$ submanifold  $\{t_0\} \times \Sigma = \tau^{-1}(t_0)$  is a  spacelike (smooth) Cauchy hypersurface.
\end{itemize}
For a real, massive Klein-Gordon field propagating on a globally hyperbolic manifold, the construction of a quasifree Hadamard state reduces to finding a weak solution to the characteristic initial value problem for the two-point correlation distribution $\omega_2 \in \mathcal{D}'(\mathcal{M} \times \mathcal{M})$:
$$\begin{cases}
\mathcal{P}_x \, \omega_{2}(x,y)= 0\\
\mathcal{P}_y \, \omega_{2}(x,y)= 0 \\
\omega_{2}|_{\Sigma_0 \times \Sigma_0} \in \mathcal{D}'(\Sigma_0\times \Sigma_0) \\
\nabla_{n_x}\omega_{2}|_{\Sigma_0 \times \Sigma_0} \in \mathcal{D}'(\Sigma_0\times \Sigma_0)  \\
\nabla_{n_y} \omega_{2}|_{\Sigma_0 \times \Sigma_0} \in \mathcal{D}'(\Sigma_0\times \Sigma_0) \\
\nabla_{n_x}\nabla_{n_y} \omega_{2}|_{\Sigma_0 \times \Sigma_0} \in \mathcal{D}'(\Sigma_0\times \Sigma_0),
\end{cases}$$
where $\mathcal{P}:= \square_g - m^2 - \xi R$  denotes the Klein-Gordon operator, $R$ is the scalar curvature, $\square_g : = g^{ab}\nabla_a\nabla_b$ is the connection d'Alembertian operator defined on $(\cM,g_\ab)$,  $m$ is the mass of the field, $\xi$ the coupling to the scalar curvature and $n$ is the future-directed unit normal vector field to the Cauchy hypersurface $\Sigma_0$. The associated bidistribution $\omega_2$ must further satisfy the following structural constraints:
\begin{align*}
\mathrm{(i)} &\quad \omega_2(f,f') - \omega_2(f',f) = i \mathcal{G}_{\mathcal{P}}(f,f') \\
\mathrm{(ii)} &\quad |\mathcal{G}_{\mathcal{P}}(f,f')|^2 \leq 4 \omega_2(f,f) \omega_2(f',f') \\
\mathrm{(iii)} &\quad \omega_2(f,f) > 0 \quad \text{for all } f \notin \operatorname{ran}(\mathcal{P}) \\
\mathrm{(iv)} &\quad \mathrm{WF}(\omega_2) = \{ (x, k_x; y, -k_y) \in (T^*\mathcal{M}\times T^*\mathcal{M}) \setminus \{0\} \mid (x,k_x) \sim_{\parallel} (y,k_y),  k_x \triangleright 0 \}
\end{align*}
where, $\mathcal{G}_{\mathcal{P}}$ is the \emph{causal propagator}, i.e.~the difference between the retarded and the advanced Green operator, and where $\mathrm{WF}(\omega_{2})$ denotes the \emph{wavefront set} of the bidistribution $\omega_{2}$. Condition (i) implements the canonical commutation relations, while (ii) and (iii) enforce the strict positivity of the state. Condition (iv), known as the {\bf Hadamard condition} (or {\bf microlocal spectrum condition}), geometrically constrains the wavefront set $\mathrm{WF}(\omega_2)$: the relation $(x,k_x) \sim_{\parallel} (y,k_y)$ signifies that $x$ and $y$ are connected by a null geodesic along which the cotangent vector $k_x$ is parallel-transported to $k_y$, and $k_x \triangleright 0$ restricts $k_x$ to the open future lightcone. As proved by G\'erard and Wrochna in\cite{GW1,GW2}, when $(\cM,g_\ab)$ is a Lorentzian manifold of bounded geometry in a neighborhood of a Cauchy hypersurface, constructing a state satisfying the constraints (i)--(iv) boils down to construct a pseudodifferential operator $b \in \Psi^1(\Sigma_0)$ that factorises the Klein--Gordon operator modulo smoothing operators, i.1. 
$$\mathcal{P} - (\partial_t + i b)(\partial_t - i b) = r^{\pm}_{-\infty} \in C_b^{-\infty}(\mathbb{R}, \Psi^{-\infty}(\Sigma_0)).$$
In the specialised setting of an \emph{ultrastatic} globally hyperbolic spacetime, i.e.~a Lorentzian manifold that admits a temporal function with $\beta=1$ and for which the spatial metric $\gamma_{ij}$ is a complete Riemannian metric independent of $\tau$, the initial data characterising the vacuum state take the explicit form
$$\begin{cases}
\mathcal{P}_x \, \omega_{2}(x,y)= 0\\
\mathcal{P}_y \, \omega_{2}(x,y)= 0 \\
\omega_{2}|_{\Sigma_0 \times \Sigma_0} = \frac{1}{2} (-\Delta_{\gamma_0} + m^2)^{-\frac{1}{2}}(x,y) \\
\nabla_{n_x}\omega_{2}|_{\Sigma_0 \times \Sigma_0} = \frac{i}{2} \delta_{\gamma_0}(x,y) \\
\nabla_{n_y} \omega_{2}|_{\Sigma_0 \times \Sigma_0} = -\frac{i}{2} \delta_{\gamma_0}(x,y) \\
\nabla_{n_x}\nabla_{n_y} \omega_{2}|_{\Sigma_0 \times \Sigma_0} = \frac{1}{2} (-\Delta_{\gamma_0} + m^2)^{\frac{1}{2}}(x,y),
\end{cases}$$
where $\Delta_{\gamma_0}$ is the Laplace-Beltrami operator on $(\Sigma_0,\gamma_0)$ and $\delta_{\gamma_0}$ is the canonical Dirac distribution with respect to the induced Riemannian volume measure. Following the Gelfand-Naimark-Segal (GNS) construction, the elements of the abstract algebra $\mathcal{A}$ are represented as operator-valued distributions on a dense domain of a Hilbert space. However, the classical stress-energy tensor for a real scalar field $\phi \colon \mathcal{M} \to \mathbb{R}$, given by
\begin{equation}\label{eq:T_ab-omega}
    T_{ab} = \frac{1}{2} \nabla_a \nabla_b \phi^2 - \phi \nabla_a \nabla_b \phi - \frac{1}{2} g_{ab} \phi \square_g \phi - \frac{1}{2} g_{ab} \left(\frac{1}{2} \square_g \phi^2 + m^2 \phi^2 \right) + \xi\left( G_{ab} - \nabla_a \nabla_b + g_{ab} \square_g \right)\phi^2
\end{equation}
involves pointwise products of field distributions, which are ill-defined as operators. The Hadamard condition ensures that the bidistribution $\omega_2$ shares its singular structure with a geometric parametrix $\mathcal{H} \in \mathcal{D}'(\mathcal{M} \times \mathcal{M})$ known as {\bf Hadamard parametrix}. As explained by Brunetti, Fredenhagen and Kohler in \cite{BFK}, the difference $\omega_2(x,y) - \mathcal{H}(x,y)$ is smooth, therefore admitting a well-defined restriction to the diagonal of $\cM\times\cM$. Imposing the standard axiomatic requirements on the regularisation scheme --namely locality, covariance, conservation, and appropriate analytic scaling behaviour-- Hollands and Wald proved in a series of paper~\cite{HW01,HW02,HW05}  see also \cite{Moretti}, that the renormalised expectation value of the stress energy tensor in a quasifree Hadamard state is uniquely determined up to local curvature terms
\begin{equation}\label{eq:stress-tensor}
    \langle \colon T_{ab} \colon \rangle_\omega = \lim_{y \to x} \mathcal{D}_{ab} \big( \omega_2(x,y) - \mathcal{H}(x,y) \big) + \frac{A}{3} g_{ab} + \alpha_1 m^4 g_{ab} + \alpha_2 m^2 G_{ab} + \alpha_3 J_{ab} + \alpha_4 I_{ab}\,.
\end{equation}
Here, the terms are strictly defined as follows:
\begin{itemize}
\item $\omega_2(x,y)$ represents the Schwartz integral kernel of the state.
\item $\mathcal{D}_{ab}$ is a bidifferential operator derived from the classical expression for $T_{ab}$, symmetrised in $x$ and $y$, acting on the variables prior to the coinciding point limit.
\item The Hadamard parametrix $\mathcal{H} \in \mathcal{D}'(\mathcal{M} \times \mathcal{M})$ is locally expressible as
\begin{equation}\label{eq:Hadamard-parametrix}
    \mathcal{H}(x,y) = \lim_{\epsilon \to 0^+} \left( \frac{\mathsf{u}(x,y)}{\sigma_\epsilon(x,y)} + \mathsf{v}(x,y) \log(\mu^2 \sigma_\epsilon(x,y)) \right),
\end{equation}
where $\sigma_\epsilon(x,y) := \sigma(x,y) + i \epsilon (t(x) - t(y))$ is the $\epsilon$-regularised Synge world function (half the squared geodesic distance,  taken with a minus sign for causally separated points), $t$ is a global time function, $\mathsf{u}, \mathsf{v} \in C^\infty(\mathcal{M} \times \mathcal{M})$ are smooth geometric biscalars, and $\mu > 0$ is an arbitrary reference mass scale.
\item $A$ is the coefficient that dictates the breakdown of conformal invariance, explicitly given by
\begin{equation}\label{eq:Anomaly}
    \frac{A}{3} = \frac{1}{4\pi^2} \left( \frac{(6\xi - 1)m^2 R}{24} + \frac{(6\xi-1)^2 R^2}{288} + \frac{R_{abcd}R^{abcd} - R_{ab}R^{ab}}{720} + \frac{(5\xi -1)\square_g R}{120} \right),
\end{equation}
where $R_{abcd}$ and $R_{ab}$ denote the Riemann and Ricci tensors, respectively.
\item $\alpha_i \in \mathbb{R}$ are {\bf renormalisation constants}, fixing the residual finite ambiguities of the theory.
\item $I_{ab}$ and $J_{ab}$ are conserved local geometric tensors 
given by
        \begin{align*}
            J_{ab} &= 2\nabla_a \nabla_b R -2 g_{ab} \square_g R + \frac{1}{2} g_{ab} R^2 -2 R R_{ab} \\
        I_{ab} &= -2\square_g R_{ab} + \frac{2}{3} \nabla_a \nabla_b R + \frac{1}{3} \square_g R g_{ab} - \frac{1}{3} R^2 g_{ab} + \frac{4}{3} R R_{ab} + g_{ab}R_{cd}R^{cd} -4 R_{abcd}R^{cd} \,.
        \end{align*}    
\end{itemize}
We emphasise that the trace coefficient $A$ constitutes an inevitable {\bf anomalous contribution} to the trace of the stress-energy tensor \cite{Waldtrace, Moretti}. It corresponds geometrically to the coinciding point limit of the Hadamard coefficient $\mathsf{v}_1$, and it cannot be entirely absorbed by redefining the finite renormalisation parameters $\alpha_i$. In flat spacetimes, such as the Minkowski background, $A$ vanishes identically.

\subsection{The Cauchy problem for the semiclassical Einstein-Klein-Gordon system}

\noindent The semiclassical Einstein-Klein-Gordon equations are a  system of coupled PDEs for a Lorentzian metric $g_{ab}\in\Gamma(T^*\cM\otimes_s T^*\cM)$ and the integral kernel of a bidistribution $\omega \in \cD'(\cM\times\cM)$,
\begin{equation}
\label{eq:system-of-equation}
    \begin{cases}
         G_\ab  = \kappa \langle \Wick{ T_\ab }\rangle_{\omega} \\
          \mathcal{P}_{x} \;\omega(x,y)= 0\\
          \mathcal{P}_{y}\; \omega(x,y)= 0 
    \end{cases}
\end{equation}
subject to the following constraints: 
\begin{enumerate}
    \item[(I)] the resulting quasifree state $\omega$ is Hadamard;
\item[(II)] $(\cM, g_{ab})$ is globally hyperbolic.
\end{enumerate}

Due to the presence of constraints, formulating a Cauchy problem is highly non-trivial, as constructing a Hadamard state would {\em a priori} require specifying an infinite number of initial conditions on $\Sigma_0$, see e.g. \cite{GS}. To avoid prescribing an infinite set of Cauchy data, we have to change our perspective. Instead of setting a forward initial value problem (IVP), we shall set a forcing problem (FP) in a spacetime so that the solutions of the (FP) will solve the constraint (I) and it will also be a solution of~\eqref {eq:system-of-equation} in the future of $\Sigma_0$. Then, the (IVP) can be recovered {\em a posteriori}, by restricting the obtained solutions to $\Sigma_0$. 
More precisely, since in this paper we are interested in studying the solutions of the semiclassical Einstein-Klein-Gordon systems perturbed around the {\bf vacuum solution} $(\eta_\ab,\omega_0)$, where $\eta_\ab$ is the Minkowski metric and $\omega_0$ is the Poincar\'e  vacuum for Klein-Gordon,  we shall set a (FP) for the perturbations employing a temporal cutoff function $\chi\in C^\infty_\pc(\RR)$. The goal is to construct a metric $g_\ab$ that agrees with $\eta_\ab$ in the past of supp$(\chi)$ and a bidistribution $\omega\in \cD'(\cM\times\cM)$ that solves
$$    \begin{cases}
          \mathcal{P}_{x} \;\omega(x,y)= \Xi(x,y)\\
          \mathcal{P}_{y}\; \omega(x,y)= \Xi(x,y)\\
 \omega(x,y)|_{\Sigma_{-\epsilon}\times \Sigma_{-\epsilon}}  = \frac{1}{2}(-\Delta+m^2)^{-\frac{1}{2}}(x,y)  \\
        \partial_t\omega(x,y)|_{\Sigma_{-\epsilon}\times \Sigma_{-\epsilon}} = \frac{i}{2}  \delta(x,y) \\
           \partial_\tau \omega(x,y)|_{\Sigma_{-\epsilon}\times \Sigma_{-\epsilon}} = -\frac{i}{2}  \delta(x,y) \\
       \partial_t\partial_\tau \omega(x,y)|_{\Sigma_{-\epsilon}\times \Sigma_{-\epsilon}}   = \frac{1}{2}(-\Delta+m^2)^{\frac{1}{2}}(x,y)
    \end{cases}
$$
for a suitable choice of $\Xi\in C_c^\infty(\mathcal{M}\times\mathcal{M})$,  with $\Sigma_{-\epsilon}$ contained in the past of supp$(\chi)$ and with $\Delta$ being the Euclidean Laplacian. In this way, $g_\ab$ will be globally hyperbolic --at least for short time-- and $\omega$ will agree with the Poincaré vacuum in a neighborhood of $\Sigma_{-\epsilon}\times \Sigma_{-\epsilon}$, hence, it will satisfy constraint (I) on account of H\"ormander's propagation of singularities theorem. 
\medskip

Before setting a (FP) for the perturbation of the vacuum solution of \Eqs~\eqref{eq:system-of-equation}, let us clarify how to set a (FP) in a simple toy model. Consider the forward (IVP) 
\begin{equation}\label{eq:incauchyS}
\begin{cases}
\mS \, u  &= 0 \\
 u|_{{\Sigma_0}} &=  u_0 
\end{cases}
\end{equation}
for the first-order hyperbolic operator
$$\mS: C^\infty(\RR^2)\to C^\infty(\RR^2)\,, \qquad \mS:= \partial_t - \partial_x \,.$$
By multiplying on the left with the temporal cutoff function $\chi\in C^\infty_\pc(\RR)$ defined by
\begin{equation}\label{eq:cutoff}
\chi(t):=\begin{cases}
 1 & t\geq 0\\
 0 & t \leq -\epsilon \
\end{cases} 
\end{equation}
with $\epsilon >0$, we obtain 
$$
\begin{cases}
(\mS \, u)\chi  &=  \mS \, (\chi u) - (\partial_t \chi) u  \\
 (\chi u)|_{{\Sigma_0}} &=  \chi u_0  \,.
\end{cases}
$$
We notice that for any solution $u$ of the forward (IVP)~\eqref{eq:incauchyS}, there exists a unique solution $\tilde{v}$ of the (FP)
\begin{equation*}
\begin{cases}
\mS \,\tilde v &= f:=  (\partial_t \chi) u \\
\tilde v|_{\Sigma_{-\epsilon}} &= 0 
\end{cases}
\end{equation*}
given by $\tilde v=\chi u$.
Conversely, given a smooth source $f$ that is compactly supported in the interior of $ J^+(\Sigma_{-\epsilon}) \cap J^-(\Sigma_0)$, for any solution of the (FP)
\begin{equation*}
\begin{cases}
\mS \,\tilde v &= f\in C_\cc^\infty( J^+(\Sigma_{-\epsilon}) \cap J^-(\Sigma_0)) \\
\tilde v|_{\Sigma_{-\epsilon}} &= 0 
\end{cases}
\end{equation*}
there exists a unique solution of the forward (IVP) 
\begin{equation*} 
\begin{cases}
\mS \, u  &= 0 \\
 u|_{\Sigma_0} &=  \Tilde{v}|_{\Sigma_0} \,.
\end{cases} 
\end{equation*} 

In our context --where an (IVP) cannot be defined \emph{a priori}-- we exploit this equivalence in reverse. 
We shall solve a (FP)  with a smooth source $f$ with the support that is compactly supported in the interior of $ J^+(\Sigma_{-\epsilon}) \cap J^-(\Sigma_0)$
\begin{equation*}
\begin{cases}
\mS \,\tilde v &= f\in C_\cc^\infty( J^+(\Sigma_{-\epsilon}) \cap J^-(\Sigma_0)) \\
\tilde v|_{\Sigma_{-\epsilon}} &= 0 
\end{cases}
\end{equation*}
and then we can recover a forward (IVP) 
\begin{equation*} 
\begin{cases}
\mS \, u  &= 0 \\
 u|_{\Sigma_0} &=  \Tilde{v}|_{\Sigma_0} 
\end{cases} 
\end{equation*}
at {\em posteriori.} This strategy shares a theoretical foundation with the method for constructing retarded Green operators, as readers familiar with the subject will recognise. Within this framework, the specification of initial data naturally translates into the selection of an appropriate source term. 
\medskip

Having illustrated this mechanism with the toy model, we now proceed to the actual semiclassical framework. As we anticipated in the beginning of this section, we are interested in the semiclassical Einstein-Klein-Gordon equations linearised around the vacuum solution
$(\eta_\ab ,\omega_0)$. Although the Einstein tensor is identically zero for the Minkowski metric, the expectation value of the stress-energy tensor for the Poincaré vacuum does not vanish identically. To this end, it will be necessary to find a set of renormalisation constants $\alpha_i$, so that the vacuum solution will solve the semiclassical Einstein-Klein-Gordon equations.
The renormalisation constants $\alpha_3$ and $\alpha_4$ do not appear, because they multiply contributions containing curvature tensors, which vanish on Minkowski spacetime. Therefore, they do not play any role on that background. The constant $\alpha_2$ is a renormalisation of the Newton constant which is fixed by experiments and furthermore, it multiplies the Einstein tensor, which vanish on a flat background.
Hence, it also has no new role on the Minkowski background. The only relevant constant is $\alpha_1$, which plays the role of a cosmological constant. As we shall see in Proposition~\ref{prop:background}, the choice of $\alpha_1$ for which the vacuum solution $(\eta_{ab},\omega_0)$ solves the semiclassical Einstein-Klein-Gordon equations
\begin{equation*}
    \begin{cases}
 G_\ab^{(0)}:= G_\ab[\eta_\ab]  = \kappa \langle \Wick{T_{\ab}[\eta_\ab]}\rangle_{\omega_0}  \\
{\mathcal{P}^{(0)}_x}\omega_0(x,y):={\mathcal{P}_{[\eta_\ab]}}_x \;\omega_0(x,y) = 0\\
{\mathcal{P}^{(0)}_y}\omega_0(x,y):={\mathcal{P}_{[\eta_\ab]}}_y \; \omega_0(x,y) = 0 \,
\end{cases}
\end{equation*}
is given by 
	\begin{equation*}
	\alpha_1  =\frac{1}{64 \pi^2} \left( - \frac{3}{2} + 2\gamma + \log \left(\frac{m^2}{2\mu^2}\right)
	\right)\,,
	\end{equation*}
where $m$ is the mass of the Klein-Gordon field, $\mu$ is the parameter appearing in the Hadamard parametrix and $\gamma$ is the Euler-Mascheroni constant.  

At the first-order perturbation, we are then left to solve
$$
\begin{cases}
 G_\ab^{(1)}:=\left.\frac{d}{d\lambda}\Big( G_\ab[\eta_\ab+\lambda h_\ab]\Big)\right|_{\lambda=0} = \left.\kappa \frac{d}{d\lambda} \Big(\langle \Wick{T_{\ab}[\eta_\ab+\lambda h_\ab]}\rangle _{[\omega_0+\lambda\omega_1]}\Big)\right|_{\lambda=0}\\
\left.\frac{d}{d\lambda} \Big({\mathcal{P}_x}_{[\eta_\ab+\lambda  h_\ab]}(\omega_0(x,y)+\lambda\omega_1(x,y))\Big)\right|_{\lambda=0}=0\\
\left.\frac{d}{d\lambda} \Big({\mathcal{P}_y}_{[\eta_\ab+\lambda  h_\ab]}(\omega_0(x,y)+\lambda\omega_1(x,y))\Big)\right|_{\lambda=0}=0
\end{cases}
$$ 
To set a forcing problem in the space of past-compact sections, we follow the toy model and we multiply the linearised equations by $\chi$ defined by \Eq~\eqref{eq:cutoff}. Setting $\tilde{h}_{ab} :=  h_\ab\chi$ and ${\tilde\omega_2^{(1)}} := \omega_2^{(1)}(1\otimes \chi + \chi \otimes 1)$, we finally obtain a forcing problem for the perturbations $(\tilde{h}_\ab,\tilde\omega_1)\in \Gamma_\pc(\otimes_s^2 T^*\cM)\oplus \cD_\pc'(\cM\times \cM)$ 
\begin{equation}
    \label{eq:FPpert}
\begin{cases}
 G_\ab^{(1)}[\tilde{h}_\ab] = 
\left.\kappa \frac{d}{d\lambda} \Big(\langle \Wick{T_{\ab}[\eta_\ab+\lambda\tilde h_\ab]}\rangle _{[\omega_0+\lambda\tilde\omega_1]}\Big)\right|_{\lambda=0}+S_\ab\\
\left.\frac{d}{d\lambda} \Big({\mathcal{P}_x}_{[\eta_\ab+\lambda \tilde h_\ab]}(\omega_0(x,y)+\lambda\tilde\omega_1(x,y))\Big)\right|_{\lambda=0}=\Xi(x,y)\\
\left.\frac{d}{d\lambda} \Big({\mathcal{P}_y}_{[\eta_\ab+\lambda \tilde h_\ab]}(\omega_0(x,y)+\lambda\tilde\omega_1(x,y))\Big)\right|_{\lambda=0}=\Xi(x,y)\\
 \tilde\omega_1|_{\Sigma_{-\epsilon}\times \Sigma_{-\epsilon}}  =\partial_t\tilde\omega_1|_{\Sigma_{-\epsilon}\times \Sigma_{-\epsilon}}  =\partial_\tau \tilde\omega_1|_{\Sigma_{-\epsilon}\times \Sigma_{-\epsilon}}= \partial_t\partial_\tau \tilde\omega_1|_{\Sigma_{-\epsilon}\times \Sigma_{-\epsilon}}   =0  \\
\tilde{h}_\ab|_{\Sigma_{-\epsilon}}=\partial_t\tilde{h}_\ab|_{\Sigma_{-\epsilon}}=0  
\end{cases}
\end{equation}
where $S_\ab$ and $\Xi$ represents the symmetric, smooth past-compact (supported in the future of $\Sigma_{-\epsilon}\times \Sigma_{-\epsilon} $) terms generated by the commutation with the temporal cut-off function. 

\medskip

\subsection{The linearised semiclassical Einstein equations}\label{subsec:Nonlocsemi}

The basic idea of decoupling the semiclassical Einstein-Klein-Gordon equations is quite natural and simple. While the system
$$ \begin{cases}
         G_\ab  = \kappa \langle \Wick{ T_\ab }\rangle_{\omega} \\
          {\mathcal{P}}_{x} \;\omega(x,y)= 0\\
          {\mathcal{P}}_{y}\; \omega(x,y)= 0 
    \end{cases}$$
 is local when we consider the unknown $(g_\ab,\omega)$, it becomes nonlocal with respect to the single unknown $g_\ab$, by solving the equations for $\omega$. 
This will be achieved via the so-called {\bf classical M\o ller operator} as follows.
Consider two globally hyperbolic metrics $g_\ab$ and $\tilde{g}_\ab$ that are {\bf paracausally related} in the sense of~\cite{MMV}, namely there exists a finite sequence of globally hyperbolic metrics starting at $g_\ab$ and ending at $\tilde g_\ab$ such that the open future light cones of the consecutive metrics overlap. Then the algebras of observables $\cA_{ g_\ab}$ and $\cA_{\tilde g_\ab}$ for the Klein-Gordon fields propagating, respectively, on $(\cM, g_\ab)$ and $(\cM,\tilde g_\ab)$  are $*$-isomorphic. By denoting this $*$-isomorphism with $\mathcal{R}:\cA_{ g_\ab}\to\cA_{\tilde g_\ab}$, we can notice that any Hadamard state~$\omega_{ g_\ab}:\cA_{ g_\ab}\to \CC$ can be obtained by pullback along $\mathcal{R}$ a Hadamard state~$\tilde{\omega}_{\tilde{\g}_\ab}:\cA_{\tilde g_\ab}\to \CC$. Furthermore, it holds that
$$ {\mathcal{P}}_{[ g_\ab]} \;\omega_{\g} = {\mathcal{P}}_{[ g_\ab]} \mathcal{R}^*\;\tilde\omega_{\tilde g_\ab} = {\mathcal{P}}_{[\tilde g_\ab]} \;\tilde\omega_{\tilde g_\ab} =  0 \,. $$
Using this $*$-isomorphism, we can decouple the semiclassical Einstein-Klein-Gordon equations as 
$$ \begin{cases}
 G_\ab[ g_\ab]  = \kappa \langle (\Wick{\mathcal{R}( T_\ab)[\tilde g_\ab]}\rangle_{\tilde\omega_{\tilde{g}_\ab}} \\
{\mathcal{P}_x}_{[\tilde g_\ab]} \;\tilde\omega_{\tilde g_\ab}(x,y) =  0\\{\mathcal{P}_y}_{[\tilde g_\ab]} \;\tilde\omega_{\tilde g_\ab}(x,y) =  0
        \end{cases}$$
but the price to pay is that the first equation becomes nonlocal in time and space. \medskip

We shall now combine this idea with the Cauchy problem we have obtained in the previous section. Since we are interested in studying the perturbations of the vacuum solution, we have a natural candidate for  $\tilde{g}_\ab$, namely 
$$\tilde{g}_\ab=\eta_\ab\,.$$ 
For sufficiently small time, $ g_\ab$ is globally hyperbolic and the intersections of the open lightcones of $ g_\ab$ and $\eta_\ab$ is always non-empty. Therefore, we can conclude that they are paracausally related ---see e.g.~\cite[Prop. 22]{MMV}.
It remains to write a state $\omega_{ g_\ab}$ as the pullback of a state $\omega_{\eta_\ab}$ for a Klein-Gordon field propagating on Minkowski spacetime. To this end, we notice that  
$$\tilde\omega_{\eta_\ab}=\omega_0 + \tilde{\omega}\,,$$
where $\tilde{\omega}$ is a smooth bisolution of an inhomogeneous Klein-Gordon equation that is past-compact  and that contributes to $\tilde{\omega}_1$ in \eqref{eq:FPpert}. Putting all together, the linearisation of stress energy tensor in the Cauchy problem~\eqref{eq:FPpert} will be given by
\begin{equation*}
\left.\kappa \frac{d}{d\lambda} \Big(\langle \Wick{T_{\ab}[\eta_\ab+\lambda\tilde h_\ab]}\rangle _{[\omega_0+\lambda\tilde\omega_1]}\Big)\right|_{\lambda=0}
=
- \kappa\langle N_\ab[\tilde h_\ab]\rangle_{\omega_0} +
\kappa\langle \Wick{ T_\ab[\tilde h_\ab]}\rangle_{\omega_0}
+
\kappa\langle\Wick{{T}_{ab}[\eta_\ab]}\rangle_{\tilde\omega}\,,
\end{equation*}
where
$N_\ab [ h_\ab]$ is a suitable nonlocal contribution depending linearly on $ h_\ab$,  which is due to the linear part in $\tilde{h}_{ab}$ of the  
pullback of $\omega_0$ to $\mathcal{A}_{\tilde{g}_{ab}}$, namely to the remaining contribution in $\tilde{\omega}_1$ in \eqref{eq:FPpert}.  From a technical point of view, it is easier to compute the operator $N_\ab [ h_\ab]$ using the so-called {\bf quantum M\o ller map} -- we refer to  subsection~\ref{sec:nonlocalcontr} for more details.  The price to pay is that we have to match the renormalisation constant $\alpha_1$ with the one ${\tilde{\alpha}}_1$ obtained using the quantum M\o ller map.
Up to this technical detail, the Cauchy problem for the semiclassical Einstein-Klein-Gordon equations can be rewritten as two distinct Cauchy problems: 
\begin{itemize}
    \item the Cauchy problem for the past-compact perturbation of a Hadamard state 
    \begin{equation*}
\begin{cases}
{\mathcal{P}_x}_{[\eta_\ab]}\tilde\omega(x,y)=\Xi'(x,y)\\
{\mathcal{P}_y}_{[\eta_\ab]}\tilde\omega(x,y)=\Xi'(x,y)\\
 \tilde\omega|_{\Sigma_{-\epsilon}\times \Sigma_{-\epsilon}}  =\partial_t\tilde\omega|_{\Sigma_{-\epsilon}\times \Sigma_{-\epsilon}}  =\partial_\tau \tilde\omega|_{\Sigma_{-\epsilon}\times \Sigma_{-\epsilon}}= \partial_t\partial_\tau \tilde\omega|_{\Sigma_{-\epsilon}\times \Sigma_{-\epsilon}}   =0 
\end{cases}
\end{equation*}
\item  the Cauchy problem for the {\bf linearised semiclassical  Einstein equations} \begin{equation}
    \label{eq:FPmetric}
\begin{cases}
G_{ab}^{(1)}[\tilde{h}_\ab]  + \kappa \langle N_\ab[\tilde h_\ab]\rangle_{\omega_0} -
\kappa \langle \Wick{T_{ab}[\tilde h_\ab]}\rangle_{\omega_0}
=
S'_\ab\\
\tilde{h}_\ab|_{\Sigma_{\epsilon}}=\partial_t\tilde{h}_\ab|_{\Sigma_{-\epsilon}}=0  
\end{cases}
\end{equation}
where $S'_\ab=S_\ab - \kappa \langle\Wick{{T}_{ab}[\eta_\ab]}\rangle_{\tilde\omega}$.
\end{itemize}
Like for the linearised Einstein equation, we cannot expect the Cauchy problem~\eqref{eq:FPmetric} to be well-posed. Indeed, the semiclassical equations inherit a {\bf linear gauge symmetry} parameterised by past-compact smooth (0,1)-tensor $X_{a}$  
$$\tilde h_\ab\mapsto \tilde h_\ab+ \nabla_{(a}X_{b)}\,,$$ where $ \nabla_{(a}X_{b)} =\frac{1}{2}(\partial_a X_b + \partial_b X_a)$ denotes the symmetrised gradient.
Being invariant under linear gauge transformations implies that the linearised semiclassical Einstein equations are not hyperbolic. Therefore, we should study the Cauchy problem up to a gauge transformation, namely 
$$\frac{\mathfrak{G}^{-1}(\{S_{ab}\})} { \text{ran}  \nabla^S|_{\Gamma_{\pc}(T^*\mathcal{M})}}\,,$$
for a source $S_{ab}\in\Gamma_{\mathrm{pc}}(T^*\cM \otimes_s T^*\cM)$,
where $\mathfrak{G}$ and $\nabla_S$ are the operators defined by
\begin{align*}
 &\mathfrak{G}: \Gamma_\pc(T^*\cM \otimes_s T^*\cM)\to \Gamma_\pc(T^*\cM \otimes_s T^*\cM) \qquad 
 \mathfrak{G}(\tilde{h}_\ab) := G_\ab^{(1)}[\tilde{h}_\ab]  + \kappa\langle N_\ab[\tilde h_\ab]\rangle_{\omega_0} -
\kappa\langle \Wick{T_{ab}[\tilde h_\ab]}\rangle_{\omega_0} \\
& \nabla^S: \Gamma_\pc(T^*\cM)\to \Gamma_\pc(T^*\cM \otimes_s T^*\cM)  \hspace{2cm}\nabla^S_{a}X_b:=\nabla_{(a}X_{b)}\,. 
\end{align*}

\subsection{Summary of the main results and structure of the paper}

To discuss the well-posedness of the Cauchy problem in de Donder gauge, we should further decouple the linearised semiclassical Einstein equations. This will be achieved by means of a suitable decomposition of past-compact symmetric (0,2)-tensors performed in Section~\ref{se:metric-perturbation-decomposition}. In this section, we will prove that this decomposition is compatible with the de Donder gauge, it realises a {\bf complete gauge fixing}, and, more importantly, it allows to decouple the linearised semiclassical Einstein equations in a Scalar and in a Transverse-Traceless Tensor part. Building on top of this decomposition, we shall discuss the explicit form of the nonlocal operator appearing in the linearised semiclassical Einstein equations using the quantum M\o ller operator. To this end, we should match the renormalisation constants computed in this setting with the ones obtained imposing that the vacuum solution solves the semiclassical Einstein-Klein-Gordon system. This analysis will be performed in Section~\ref{sec3}, where the principle of general
local covariance will be employed together with the principle of perturbative agreement to fix some of the renormalisation
constants.
A summary of the main results of these two sections can be found in the following theorem.

\begin{theorem}
Let $\bar{h}_{ab}\in\Gamma_{\mathrm{pc}}(T^{\ast}\mathbb{R}^{n}\otimes_{s}T^{\ast}\mathbb{R}^{n})$ be such that $\partial^{a}\bar{h}_{ab}=0$. Then, there are unique smooth and past-compact fields $\bar w\in C^{\infty}_{\mathrm{pc}}(\mathbb{R}^{n})$ and $ \bar  h_\ab^{\mathrm{T}\mathrm{T}}\in\Gamma_{\mathrm{pc}}(T^{\ast}\mathbb{R}^{n}\otimes_{s}T^{\ast}\mathbb{R}^{n})$ 
with $\partial^{a}\bar h_{ab}^{\mathrm{T}\mathrm{T}}=0$ and $\eta^{ab}\bar h_{ab}^{\mathrm{T}\mathrm{T}}=0$, such that $\bar h_{ab}$ decomposes as 
\begin{align*}
\bar h_{ab}= \bar h_{ab}^{\mathrm{S}} +\bar h_{ab}^{\mathrm{T}\mathrm{T}}.
\end{align*}
All the solutions of the constrained Cauchy problems
\begin{align}
   & \begin{cases}\label{eq:S}
    -\frac{1}{2 \kappa} \square \bar{h}_{ab}^{\mathrm{S}}
	= - \frac{1}{4} \cS \cK_0 [\bar{h}_{ab}^{\mathrm{S}}] +\frac{1}{64 \pi^2} m^4 \bar{h}_{ab}^{\mathrm{S}}  -\frac{1}{2}
{\tilde{\alpha}}^{\mathrm{S}}_2
 m^2 \square \bar{h}_{ab}^{\mathrm{S}}+
 {\tilde{\alpha}}^{\mathrm{S}}_3 \square\square \bar{h}_{ab}^{\mathrm{S}}
 + S_{ab}^{\text{S}}
	\\
    \partial^a \bar{h}_{ab}^{\mathrm{S}}=0\\
    \bar{h}^{\mathrm{S}}_{ab}|_{\Sigma_{-\epsilon}}=\partial_t\bar{h}^{\mathrm{S}}_{ab}|_{\Sigma_{-\epsilon}}=0
\end{cases}
\\
&\begin{cases}\label{eq:TT}
    -\frac{1}{2 \kappa} \square \bar{h}_{ab}^{\mathrm{TT}}
		=\frac{1}{2} 
		\cT \cK_0 [\bar{h}_{ab}^{\mathrm{TT}}] -\frac{1}{2}
{\tilde{\alpha}}^{\mathrm{TT}}_2 m^2 \square \bar{h}_{ab}^{\mathrm{TT}}+{\tilde{\alpha}}^{\mathrm{TT}}_4 \square\square \bar{h}_{ab}^{\mathrm{TT}}+S_{ab}^{\mathrm{TT}}
	\\
    \partial^a \bar{h}_{ab}^{\mathrm{TT}}=0\\
    \bar{h}^{\mathrm{TT}}_{ab}|_{\Sigma_{-\epsilon}}=\partial_t\bar{h}^{\mathrm{TT}}_{ab}|_{\Sigma_{-\epsilon}}=0
\end{cases}
\end{align}
are also solution of the Cauchy problem~\eqref{eq:FPmetric}.
Here, $\square$ 
denotes the connection d'Alembert operator of the Minkowski metric, the operators $\cS$ and $\cT$ are defined by
$$\cS = \frac{2}{3} \at m^2 + \frac{1}{2} (1 - 6\xi) \square \ct^2,  \qquad  \cT = \frac{1}{60} \at \square - 4m^2 \ct ^2,$$ respectively,
while $\mathcal{K}_0$ is constructed from
	\begin{align*}
		\cK_{0} (x) = 	-\mathrm{i} \at \mathsf{G}_F^2 (-x) - \mathsf{G}_+^2 (-x) \ct_{\mathrm{reg}} 	= \frac{1}{16\pi^2} \square  \int_{\mathbb{R}^+} \frac{1}{M }\sqrt{1-\frac{4m^2}{M}}\Theta(M-4m^2) \mathsf{G}_{\mathrm{ret}}(x,M)\, dM,
	\end{align*}
	which is a regularised integral kernel that maps past-compact smooth functions to past-compact smooth functions. As always, $\mathsf{G}_{\mathrm{ret}}(x,M)$ denotes the retarded propagator of the flat Klein-Gordon operator on the Minkowski background with mass $\sqrt{M}$,  $\mathsf{G}_F$ is the Feynman propagator while $\mathsf{G}_+=\omegazd$ is the two point distribution of the Poincar\'e vacuum. As usual, $\Theta$ denotes the Heaviside step function.
\end{theorem}

Now that we have settled the Cauchy problems~\eqref{eq:S} and~\eqref{eq:TT} for the linearised semiclassical Einstein equations in the de Donder gauge, we can finally discuss its well-posedness. This is performed in Section~\ref{se:prototype-and-solution}. As a preliminary step, we shall rewrite the equations~\eqref{eq:S} and~\eqref{eq:TT} as
\begin{equation}\label{eq:Protointro}
\iota_{\mathcal{M}}\tilde{\mathsf{G}}_{\mathrm{ret}}( \square (\square-a_1)(\square-a_2)\phi\otimes \check\rho) +\sum_{j=0}^{2}b_j \square^j \phi=S \,,
\end{equation}
where $S$ is some past-compact source,  $\rho$ and $\check{\rho}$ are suitable $L^2$-functions (see Eq.~\eqref{eq:nonlocalK0} and Eq.~\eqref{eq:rhocheck}), $a_{i}$ and $b_{i}$ are real constants, $\tilde{\mathsf{G}}_{\mathrm{ret}}$ is the retarded fundamental solution of the Klein-Gordon operator $\tilde{\square}-4m^{2}$ on \emph{$5$-dimensional Minkowski spacetime} $\mathcal{M}_{5}:=\mathcal{M}\times\mathbb{R}$ and $\iota_{\mathcal{M}}$ realises the restriction to $\mathcal{M}\cong\mathcal{M}\times\{0\}\subset\mathcal{M}_{5}$.

As written, the analysis of the prototypical equations~\eqref{eq:Protointro} is challenging, since the highest-derivative term is governed by a \emph{nonlocal} operator, thereby precluding the direct application of standard techniques for local partial differential equations. The central idea of the following analysis, therefore, is to suitably \emph{invert} the nonlocal contribution, thereby recasting the equation into a form that is more amenable for an analysis of its Cauchy problem and asymptotic behaviour to be performed.

The analysis will be performed in two steps. First of all, by introducing suitable auxiliary functions $\mathfrak{h}_{j}$ for $j=0,1,2$, we will rewrite the prototypical Eq.~\eqref{eq:Proto} in the form 
\begin{align}\label{eq:Proto2intro}
\iota_{\mathcal{M}}\tilde{\mathsf{G}}_{\mathrm{ret}} ((\square^3 \phi+\beta_2\square^{2} \phi + \beta_1 \square \phi +\beta_0\phi) \otimes ( \check \rho+b_2\mathfrak{h}_2 )) =S
\end{align}
for suitable constants $\beta_{i}$. Subsequently, we show that, provided that $\check \rho+b_2\mathfrak{h}_2$ satisfies suitable regularity and positivity conditions, which ultimately will translate into restrictions on the parameters $a_i$ and $b_i$, the operator $\iota_{\mathcal{M}}\tilde{\mathsf{G}}_{\mathrm{ret}}$ becomes invertible, thereby leading to a reformulation of the equation into a fully local problem. As a second step, we show that a similar procedure can actually always been applied to the highest-order terms without any restrictions on the parameters $b_{i}$, hence leading to a 6th order hyperbolic equation with \emph{subleading} nonlocal contributions. More precisely, we will show that Eq.~\eqref{eq:Proto2intro} can equivalently be written in the form
\begin{align*}
\square^3 \phi+\beta_2\square^{2} \phi + \beta_1 \square \phi +\beta_0\phi
+
\mathcal{G}_{\varsigma}^{-1}( 
(\tilde{b}_1-b_1) \square \phi + 
(\tilde{b}_0-b_0) \phi)
=
\mathcal{G}_{\varsigma}^{-1}(S)\,,
\end{align*}
where $\mathcal{G}_{\varsigma}(\phi):=\iota_{\mathcal{M}}\tilde{\mathsf{G}}_{\mathrm{ret}}(\phi\otimes\check{\varsigma})$ with $\check{\varsigma}:=\check{\varrho}+b_{2}\mathfrak{h}_{2}$ denotes the nonlocal operator appearing in Eq.~\eqref{eq:Proto2intro} and where $\overline{b}_{i}$ are suitable real constants.\\
Having recast the system in this form, we establish the well-posedness of the Cauchy problem by employing perturbation theory, treating the nonlocal potential as a perturbative term (see Theorem~\ref{thm:convergence}).

In the last part of Section~\ref{se:prototype-and-solution}, we then perform an asymptotic analysis by studying the large-time behaviour of solutions to the prototypical semiclassical equations~\eqref{eq:Protointro}. This is the second main result of this section, Theorem~\ref{thm:stability}, in which we show that for some ranges of parameters $a_{i}$ and $b_{i}$, solutions are bounded and decay at most as $t^{-\frac{3}{2}}$, while for some others we obtain solutions with exponential growth. The implications of this analysis of the toy model~\eqref{eq:Protointro} when applied to the semiclassical Einstein equations will be discussed in Section~\ref{sec:5}. In particular, we shall prove that for any Cauchy data, the solutions of the system are linearly unstable. However, we will show that the growth is at most exponential and can be controlled by a factor that depends only on the mass of the Klein-Gordon fields. 
This is the main result of our paper.
\begin{theorem}
For any choice of past-compact smooth tensor $S_\ab(x)$ and for any symmetric past-compact smooth function $\Xi(x,y)$, there exists unique past-compact solution $(\tilde\omega_1,\tilde h_\ab)$ of the Cauchy problem for the linearised semiclassical Einstein-Klein-Gordon equations
$$\begin{cases}
 G_\ab^{(1)}[\tilde{h}_\ab] = 
\left.\kappa \frac{d}{d\lambda} \Big(\langle \Wick{T_{\ab}[\eta_\ab+\lambda\tilde h_\ab]}\rangle _{[\omega_0+\lambda\tilde\omega_1]}\Big)\right|_{\lambda=0}+S_\ab\\
\left.\frac{d}{d\lambda} \Big({\mathcal{P}_x}_{[\eta_\ab+\lambda \tilde h_\ab]}(\omega_0(x,y)+\lambda\tilde\omega_1(x,y))\Big)\right|_{\lambda=0}=\Xi(x,y)\\
\left.\frac{d}{d\lambda} \Big({\mathcal{P}_y}_{[\eta_\ab+\lambda \tilde h_\ab]}(\omega_0(x,y)+\lambda\tilde\omega_1(x,y))\Big)\right|_{\lambda=0}=\Xi(x,y)\\
\tilde{h}_\ab|_{\Sigma_{\epsilon}}=\partial_t\tilde{h}_\ab|_{\Sigma_{-\epsilon}}=0  \\
 \tilde\omega_1|_{\Sigma_{-\epsilon}\times \Sigma_{-\epsilon}}  =\partial_t\tilde\omega_1|_{\Sigma_{-\epsilon}\times \Sigma_{-\epsilon}}  =\partial_\tau \tilde\omega_1|_{\Sigma_{-\epsilon}\times \Sigma_{-\epsilon}}= \partial_t\partial_\tau \tilde\omega_1|_{\Sigma_{-\epsilon}\times \Sigma_{-\epsilon}}   =0  \\
\end{cases}$$
for any time $t\in[0,\infty)$. For any nontrivial choice of the sources $S_\ab$ and $\Xi$, the solution $\tilde h_\ab$ grows exponentially in time, but it can be bounded from above by the conformal rescaling of the Minkowski metric using a universal factor $H$, i.e.
$$\Tilde{h}_\ab \leq e^{-Ht}\eta_\ab \,.$$

\end{theorem}

Although one might suspect that semiclassical models generally lack linear stability --and are therefore physically meaningless-- it is possible to establish the linear stability of the Cauchy problem for the \textbf{semiclassical Maxwell–Dirac system} in Minkowski spacetime. The details of this analysis are deferred to a forthcoming paper.

\subsection{Future outlook}

The analysis presented here 
provides the first step 
towards the study of the full 
nonlinear problem.
In the case of cosmological spacetimes, the full nonlinear analysis has been carried out in 
\cite{MedaPinSiemcosmo}. In that paper, it has been shown that the biggest difficulties in the treatments of the corresponding nonlinear semiclassical model are already present at the linear order. Actually, it has been shown that, once the linearised problem around various backgrounds is solved,  the analysis of the full nonlinear model follows by applying certain fixed-point methods and from standard estimates. This holds because the nonlinear contribution to the fixed point equation is a continuous functional on the space of metric perturbations and no loss of derivatives arises there.
Although that paper considers directly an initial value formulation of the problem and not a forcing problem and although it uses directly the classical M\o ller map to control the expectation value of the stress-energy tensor, the very same difficulties as those addressed in the present paper, like the presence of nonlocalities and higher derivatives, already show up at the linear order. 
We thus believe that the careful analysis of the linear problem studied in this paper consists in a 
robust foundation for extending our results to the full nonlinear and nonlocal regime. 
However, while the conceptual path toward local well-posedness in the nonlinear case is clear, its formal implementation would significantly increase the length and complexity of the present manuscript; therefore, we postpone this detailed derivation to future work.

Beyond these analytical goals, our results offer a striking physical insight, namely that the semiclassical Minkowski background $(\mathbb{R}^4, \eta_{ab})$ is fundamentally unstable. We have shown that, regardless of renormalisation constants, linear perturbations grow exponentially for arbitrary Cauchy data. However, this instability is not merely divergent; it is structured. Because these fluctuations are bounded by the conformal transformation $e^{tH} \eta_{ab}$, the quantum backreaction effectively drives the static flat spacetime toward an expanding, de Sitter-like geometry. 
This mirrors the dynamics studied in inflationary models, such as the Starobinsky model \cite{Starobinsky}, and suggests a ``geometric compensation'' where the expansion of the manifold tames the backreaction. 
Consequently, we believe that linear stability will be naturally recovered when the theory is formulated directly on an expanding de Sitter universe rather than Minkowski spacetime. We stress once more that there is, however, a striking difference in comparison to these models of inflation where de Sitter like expansion arises. We have actually seen that for a quantum field whose mass is in the regime of masses present in the standard model, the effective cosmological constant which arises is comparable to the Dark Energy measured at present day in our universe and not orders of magnitude higher than expected (this discrepancy is usually referred to in the Physics literature as ``cosmological constant problem'' \cite{Weinberg}).

\subsection{Historical survey}

The local well-posedness of the Cauchy problem for the Einstein equations, established by Fourès-Bruhat \cite{Foures-Bruhat}, with earlier works by de Donder \cite{deDonder}, Lanczos \cite{Lanczos}, Lichnerowicz \cite{Lichnerowicz2}, Darmois \cite{Darmois1,Darmois2} and Stellmacher \cite{Stellmacher1,Stellmacher2}, among others, and extended to maximal globally hyperbolic developments by Choquet-Bruhat and Geroch \cite{CBG}, naturally redirected the field towards the study of the global behaviour of solutions.\footnote{For a detailed treatment of the Cauchy problem for the Einstein equations, see the monographs of Choquet-Bruhat \cite{ChoquetBruhatBook} and Ringström \cite{RingstromBook}, the latter being based on a \emph{coordinate-free approach} developed by DeTurck \cite{DeTurck1,DeTurck2}; for a comparison, see Graham and Lee \cite{GrahamLee}. An alternative prove, in which the Einstein equations are rewritten into the form of a first-order quasilinear hyperbolic system, was studied by Fischer and Marsden \cite{FischerMarsdenCauchy}. A historical account on the Cauchy problem for Einstein's equation can be found, for instance, in \cite{Ringstrom2,ChoquetBruhatHistory}.} Within this framework, asymptotic analysis determines whether a distinguished spacetime $(\mathcal{M}, g_{ab})$ remains stable under small perturbations $h_{ab}$. Historically, the theory of nonlinear stability first matured through the study of highly symmetric background geometries. 

For spacetimes with a positive cosmological constant, Friedrich \cite{Friedrich} established the global nonlinear stability of de Sitter space, rigorously proving both asymptotic simplicity and null geodesic completeness. 

In the asymptotically flat regime, a major breakthrough was realised by Christodoulou and Klainerman \cite{Christodoulou-Klainerman}, who proved the global nonlinear stability of Minkowski space. This foundational result was later expanded by Lindblad and Rodnianski \cite{Lindblad-Rodnianski}, who introduced a robust alternative proof using the wave-coordinate (harmonic) gauge to control nonlinearities, and by Hintz and Vasy \cite{HVannPDE}, who utilised microlocal methods to derive exact asymptotic series for the metric's behaviour at infinity. 

The transition from vacuum stability to black hole stability represents a significant increase in analytical complexity. A landmark achievement in this area was realised by Dafermos, Holzegel, and Rodnianski \cite{DHR}, who established the linear stability of the static, spherically symmetric Schwarzschild spacetime by proving rigorous decay to a linearised Kerr solution. Extending these results to rotating black holes, however, remained a formidable challenge for decades. This difficulty persisted despite a rich body of literature that had successfully established mode stability and decay bounds for the associated Teukolsky equation \cite{AMPW, Whiting, FSKerr, STeix, Teixeira, Millet, MaZhang, DHRKerr}. Recently, the nonlinear stability of slowly rotating black holes was achieved through a monumental series of works by Klainerman and Szeftel \cite{KleiSze1, KleiSze2, KleiSze3}, Giorgi, Klainerman, and Szeftel \cite{GKS}, and Shen \cite{She23}. Parallelly, Hintz and Vasy \cite{HV} established the nonlinear stability of the slowly rotating Kerr–de Sitter family in the presence of a positive cosmological constant. The ultimate frontier of this program remains the proof of global nonlinear stability for the full subextremal range. A pivotal first step toward this goal was taken by H\"afner, Hintz, and Vasy \cite{HHVslow, HHV}, who demonstrated linear stability for the entire subextremal regime.

The coupling of the Einstein equations to a massive classical Klein-Gordon field  results in the Einstein–Klein–Gordon system. Upon fixing a gauge, this system reduces to a quasilinear hyperbolic system coupled to a wave equation. While local well-posedness follows from standard energy estimates, the dispersive characteristics of the massive field present unique challenges; specifically, the lack of scaling symmetry complicates the extension of vacuum global existence results to the massive case. The global nonlinear stability of Minkowski spacetime for this self-gravitating massive system was first established by LeFloch and Ma, who proved in~\cite{LeMa} that small, asymptotically flat initial data lead to globally dispersive solutions. This was followed by a definitive proof of global asymptotic stability for the full Einstein-Klein-Gordon system by Ionescu and Pausader, who utilised a combination of vector field methods and Fourier analysis to handle small, localised perturbations in~\cite{IoPa}.
\medskip

While a fairly extensive literature is currently available on the Cauchy problem for the classical Einstein–Klein–Gordon system, including both linear and nonlinear results, the situation changes drastically in the semiclassical setting. 
As previously delineated, in contrast to the classical theory, establishing the hyperbolic character of the semiclassical Einstein equations constitutes a highly non-trivial endeavour, even subsequent to the imposition of standard gauge-fixing conditions. Within this framework, 
 the presence of the expectation value of the quantum stress-energy tensor as matter source introduces severe analytical pathologies: it is intrinsically state-dependent, acts as a highly nonlocal functional of the background metric $g_{ab}$, and generically incorporates fourth-order derivatives in time. These characteristics collectively render the formulation of the {\it Cauchy problem} substantially more challenging than its counterpart in general relativity.

The first mathematical treatment of the Cauchy problem for the semiclassical Einstein--Klein--Gordon system was performed on {\em spatially flat Robertson--Walker spacetimes}. In~\cite{Pinamonti}, it was proved the local well-posedness and analysed the backreaction of a quantised, massive, conformally coupled scalar field. Subsequently, this analysis was extended in~\cite{PS}, where it was establish the global well-posedness. More precisely, it was proved that any local solutions admit a maximal globally hyperbolic development, terminating either at a spacetime singularity or a critical threshold of the Hubble parameter corresponding to a divergence for the scalar curvature. The problem of local well-posedness for a massive scalar field with an arbitrary curvature coupling $\xi$ was subsequently resolved in~\cite{MedaPinSiemcosmo}. The key idea was to isolate the highest-order derivative contributions --which generically manifest nonlocally through the action of a linear, retarded, unbounded operator-- and inverting it on a suitable space of functions. In this manner, the Cauchy problem can be rewritten as fixed-point problem.

A complementary perspective on the semiclassical Einstein-Klein-Gordon equations in spatially flat Robertson--Walker spacetimes was established by Gottschalk and Siemssen~\cite{GS}, building upon the work of Eltzner and Gottschalk~\cite{Eltzner}. By recasting the equations as an infinite-dimensional dynamical system, they proved the existence of maximal global solutions for vacuum-like states and local solutions for thermal-like states. However, a significant conceptual hurdle in this framework is the problem of state reconstruction: the quantum state is encoded by an infinite sequence of moments, meaning the full state cannot be recovered without supplying an infinite amount of initial information. This analysis was later complemented by Gottschalk, Rothe and Siemssen in~\cite{GRScosmo,GRSdesitter} in which it was fixed specific conditions on the quantum state  to reduce the infinite-dimensional problem into a manageable 4th-order differential equation.
These structural difficulties are mirrored in the approach of Kay, Juárez-Aubry, Miramontes, and Sudarsky~\cite{Juarez-Aubry2020init, Juarez-Aubry2023coll, Juarez-Aubry2024Cauch}, who addressed the initial value problem in generic globally hyperbolic spacetimes.
Also, in these approaches, to keep the regularity of the underlying quantum state, an infinite tower of initial conditions needs to be imposed.

Let us remark that, to circumvent these difficulties, simplifying assumptions are often imposed on the studied spacetimes~\cite{Juarez-Aubry,Juarez-Aubry2,Sanders}, on the expectation value of the state~\cite{DFP,Juarez-Aubry3}, or on the form of the stress-energy tensor itself~\cite{Juarez-Aubry2021conf}.

Since a general theory for the well-posedness of the Cauchy problem on arbitrary background and for arbitrary Cauchy data is missing, a mathematical discussion regarding the stability of the equations, in the strict PDE sense, is currently unavailable. 
However, in the physical literature, the issue of linear instability of the Minkowski spacetime has been approached,
showing that among the solutions to the linearised semiclassical equations, there exist modes that naturally exhibit exponential growth in time, known as {\bf runaway solutions}\footnote{The terminology of ``runaway solutions'' originates in classical electrodynamics in connection with radiation reaction for a point charge. For instance, the Lorentz–Dirac equation \cite{Dirac}, in which the self-force depends on higher derivatives of the particle worldline. In this case, the enlarged space of solutions contains exponentially growing accelerations even in the absence of external forces, which were later called runaway solutions.}. Crucially, while the global well-posedness of the underlying equations is not formally established, these physical treatments explicitly analyse the relationship between these runaway solutions and the associated Cauchy data, demonstrating how specific initial configurations inevitably excite catastrophic growth.

The first investigation of this nature was performed by Horowitz \cite{Horowitz}, who analysed Minkowski spacetime in the weak-field limit. Rather than formulating a well-posed initial value problem, Horowitz demonstrated that the equations admit exponentially growing perturbations on the Planck timescale and showed how these spurious solutions are triggered by the initial data, rendering the semiclassical evolution inherently unstable.

This phenomenological instability was subsequently corroborated and extended by several other physicists \cite{Kay,Yamagishi,Jordan,Suen,MW}. Further generalisations by Flanagan and Wald \cite{Flanagan} incorporated a broader class of first-order perturbations alongside non-vanishing incoming quantum states as sources. Again, without establishing a general well-posedness framework, this work reaffirmed that flat spacetime is destabilised by Cauchy data that naturally couple to these runaway modes. Finally, Anderson, Molina-París, and Mottola \cite{Anderson} investigated these exponentially growing solutions under the influence of a massive quantum field. Recognising the catastrophic role of the initial configurations, they proposed a stringent validity criterion for semiclassical gravity: a background spacetime can only be considered physically admissible if the allowed Cauchy data do not excite these unbounded modes, thereby ensuring that all physical solutions remain globally bounded in time.

\paragraph{General notation and conventions:}

\begin{itemize}
    \item[$\bullet$]For a given vector bundle $E\xrightarrow{\pi}\mathcal{M}$, we denote by $\Gamma(E)$ the $C^{\infty}(\mathcal{M})$-module of \emph{smooth sections} and by $\Gamma_{\mathrm{c}}(E)$, $\Gamma_{\mathrm{sc}}(E)$, $\Gamma_{\mathrm{pc}}(E)$ and $\Gamma_{\mathrm{fc}}(E)$ the subspaces consisting of sections with \emph{compact}, \emph{space-compact}, \emph{past-compact} and \emph{future-compact} support, respectively.
    \item[$\bullet$]If $(\mathcal{M},g_{ab})$ is Lorentzian manifold, we will use the \emph{mostly plus} signature convention $(-,+,\dots,+)$. In particular, the Minkowski metric on $\mathcal{M}:=\mathbb{R}^{1+n}$ is given by $\eta_{ab}=\mathrm{diag}(-1,1,\dots,1)$. 
    \item[$\bullet$]Both the \emph{abstract index notation} and \emph{Einstein summation convention} are used to denote tensor fields and their contractions. We refer to the extensive discussion in \cite[Sec.~2.4]{Waldbook} and\cite[Chap.~2]{PenroseBook} for further details on this notation.
    \item[$\bullet$]If $(\mathcal{M},g_{ab})$ is a pseudo-Riemannian manifold with Levi-Civita connection $\nabla$, we use the following sign and index conventions for the Riemann and Ricci curvature tensors:
    \begin{align*}
        (\nabla_{a}\nabla_{b}-\nabla_{b}\nabla_{a})v^{c}=\tensor{R}{_a_b^c_d}v^{d}\,,\qquad R_{ab}:=\tensor{R}{^c_a_c_b}\qquad\text{and}\qquad R:=g^{ab}R_{ab}\,,
    \end{align*}
    which agree with the one used, for instance, in the textbooks \cite{MTW,RingstromBook,ChoquetBruhatBook,WaldBook2}.
    \item[$\bullet$]If $(\mathcal{M},g_{ab})$ is a Lorentzian manifold with Levi-Civita connection $\nabla$, we define the \emph{(connection or rough) d'Alembertian} with sign convention \begin{align*}
        \square_{g}:=g^{\ab}\nabla_{a}\nabla_{b}\, ,
    \end{align*} 
    which agrees with the \emph{Bochner d'Alembertian} up to a sign, i.e.~$\square_{g}=-\nabla^{\ast}\nabla$. In the special case of Minkowski spacetime $(\mathcal{M}:=\mathbb{R}^{1+n},\eta_{ab})$, we write $\square:=\square_{\eta}=\eta^{ab}\partial_{a}\partial_{b}$, i.e.~$\square=-\partial_{t}^{2}+\Delta$, where $\Delta$ denotes the Euclidean Laplacian. The corresponding \emph{retarded Green operator} is denoted by $\mathsf{G}_{\mathrm{ret}}\colon C^{\infty}_{\mathrm{pc}}(\mathbb{R}^{n})\to C^{\infty}_{\mathrm{pc}}(\mathbb{R}^{n})$ and similarly when acting on tensor fields.  
    \item[$\bullet$]For the Fourier transform $\mathcal{F}\colon L^{2}(\mathbb{R}^{n})\to L^{2}(\mathbb{R}^{n})$, we employ the normalisation convention
    \begin{align*}
        (\mathcal{F} f)(\textbf{p}) := \int_{\mathbb{R}^{n}} e^{-\mathrm{i}\textbf{p}\cdot \textbf{x}}f(\textbf{x})\, d^n\textbf{x},\qquad\qquad f(\textbf{x}) = \frac{1}{(2\pi)^{n}} \int_{\mathbb{R}^{n}} e^{\mathrm{i}\textbf{p}\cdot \textbf{x}}(\mathcal{F}f)(\textbf{p})\, d^n\textbf{p}.
    \end{align*}
    \item[$\bullet$] Denoting with $H^2(\mathbb{C}_+)$ the Hardy space of holomorphic functions in the right complex half-plane $\mathbb{C}_+$, we employ for the Laplace transform $\mathscr{L}\colon L^{2}(0,\infty)\to H^2(\mathbb{C}_+)$ the standard convention
    \begin{align*}
        (\mathscr{L} f)(s) := \int_0^\infty e^{-s t} f(t)\, dt.
    \end{align*}
    
\end{itemize}

\paragraph{Acknowledgements and Funding:}  
We are grateful to Edoardo D'Angelo, Nicolò Drago, Benito A. Juárez-Aubry and Valter Moretti for helpful discussions. 

The work of P.M.~is financed by the European Union - NextGenerationEU - National Recovery and Resilience Plan (NRRP) - Mission 4 Component 2 Investment 1.2 - ``Funding projects presented by young researchers'' Seal of Excellence PNRR Young Researchers, ``SPACE project'' - SOE20240000129 - CUP E63C24002410003. The work of S.G.~is funded by the EPSRC Open Fellowship EP/Y014510/1. S.M. was partially supported by the INdAM - GNFM Project ``Hadamard states for linearized Yang Mills theories'', CUP E5324001950001. The work of G.S.~is supported by a research grant from the German Academic Exchange  Service (DAAD) through the programme ``Forschungsstipendien in Deutschland, 2026'' (funding programme number: 57811724). 
The research was supported in part by the MIUR Excellence Department Project
2023–2027 awarded to the Department of Mathematics of the University of Genoa, 
CUP D33C23001110001, and by the INFN project ``Bell''.

\section{Past-compact metric perturbations}\label{se:metric-perturbation-decomposition}

In this section, we consider general relativity with past-compact metric perturbations linearised around a Minkowski background. We demonstrate that any such past-compact perturbation can be uniquely decomposed into a scalar, a vector, and a transverse-traceless tensor component, all of which are past-compact. This decomposition will be crucial for the subsequent analysis of the linearised semiclassical Einstein–Klein–Gordon system, as it enables us to decouple the equations into two independent sectors: one governing the scalar part and the other governing the tensorial part, with the vectorial part vanishing due to a natural gauge fixing condition. Each sector is then described by a hyperbolic equation with nonlocal contribution of the same type, hence facilitating a more tractable analysis of the system's dynamics and linear (in)stability.

\subsection{Degrees of freedom and metric perturbations}
We consider $(1+3)$-dimensional Minkowski spacetime $(\mathbb{R}^{4},\eta_{ab})$ as a background and discuss the linearisation of the metric around $\eta_{ab}$. Moreover, as explained in the introduction, we restrict ourselves to \emph{past-compact} metric perturbations. Formally, we consider a one-parameter family of Lorentzian metrics $g_{ab}(\lambda)$ on $\mathcal{M}=\mathbb{R}^{4}$, depending smoothly on $\lambda$, such that $g_{ab}(0)=\eta_{ab}$ and $\frac{d}{d\lambda}\big\vert_{\lambda=0}g_{ab}(\lambda)=:h_{ab}\in\Gamma_{\mathrm{pc}}(T\mathbb{R}^{n}\otimes_{s}T\mathbb{R}^{n})$. In other words, the metric $g_{ab}$ is formally expanded as $g_{ab}(\lambda)=\eta_{ab}+\lambda h_{ab}+\mathcal{O}(\lambda^{2})$. For the sole purpose of deriving linearised expressions, it is sufficient to work on a formal level; we refer to the approaches discussed in \cite{StewartWalker,FischerMarsden} for rigorous perturbation theory in the context of general relativity. 

In the context of linearised gravity, it is convenient to introduce the \emph{trace-reversed perturbation}, i.e.~the symmetric $(0,2)$-tensor defined by
\begin{equation*}
	\bar{h}_{ab} := h_{ab} - \frac{1}{2}\eta_{ab} \eta^{cd} h_{cd}= h_{ab} - \frac{1}{2}\eta_{ab}  {h_c}^c.
\end{equation*}
Clearly, $\bar{h}_{ab}$ is past-compact as well. Moreover, as the name suggests, the trace-reversal satisfies
\[
	{\bar{h}_c}^c = - {h_c}^c.
\]

As a next step, consider the Einstein tensor $G_{ab}(\lambda)$ of the full metric $g_{ab}(\lambda)$, which can be expanded as $G_{ab}(\lambda)=G_{ab}^{(0)}+\lambda G_{ab}^{(1)}+\mathcal{O}(\lambda^{2})$ with $G_{ab}^{(0)}=G_{ab}(0)$ and $G_{ab}^{(1)}=\frac{d}{d\lambda}\big\vert_{\lambda=0}G_{ab}(\lambda)$. Similarly, we expand the Ricci tensor $R_{ab}(\lambda)$ and scalar curvature $R(\lambda)=g^{ab}(\lambda)R_{ab}(\lambda)$. For Minkowski spacetime, it clearly holds that $G_{ab}^{(0)}=0$ as well as $R_{ab}^{(0)}=0$ and $R^{(0)}=0$. Moreover, if we impose the \emph{(flat) de Donder} (also called \emph{harmonic} or \emph{Bianchi}) \emph{gauge} condition, 
\begin{align}\label{eq:deDonderGauge}
\pa^c \bar{h}_{cd} = 0,
\end{align}
a straightforward computation shows that the first-order contributions are given by
\begin{subequations}
\label{eq:line-hbar}
\begin{align}
R_{ab}^{(1)}&=-\frac{1}{2}\square (\bar{h}_{ab} - \frac{1}{2}\eta_{ab} {\bar{h}_c}^c )
=-\frac{1}{2}\square {h}_{ab}\\
R^{(1)} &=\eta^{ab}R^{(1)}_{ab}-h^{ab}R^{(0)}_{ab}=\frac{1}{2}\square {\bar{h}_c}^c,\\
G_{ab}^{(1)}&=R_{ab}^{(1)}-\frac{1}{2}\eta_{ab}R^{(1)}-\frac{1}{2}h_{ab}R^{(0)}=-\frac{1}{2}\square \bar{h}_{ab},
\end{align}
\end{subequations}
where $\square=\eta^{ab}\partial_{a}\partial_{b}$ is the d'Alembert operator of Minkowski spacetime. Moreover, indices are raised/lowered with respect to the background (Minkowski) metric. For a detailed derivation of the linearised expressions, including the case of general background spacetimes, we refer to \cite[Sec.~2.2.4]{SchmidThesis} or the classical reference \cite{Lichnerowicz}.

\subsection{Decomposition of past-compact symmetric tensor fields}

In this section, we present a decomposition of past-compact tensors into scalar, vectorial and transverse–traceless components. This construction is similar in spirit to the decomposition of metric perturbations in \emph{cosmological perturbation theory}, in which one isolates the gauge-invariant degrees of freedom of gravity linearised around cosmological backgrounds (see e.g.~\cite{Muk,AQFTCos}). The difference, however, is that here, we use retarded fundamental solutions of \emph{hyperbolic} equations to decompose the metric perturbation $h_{ab}$ on the (\emph{Lorentzian}) spacetime manifold $\mathcal{M}$, rather than solving \emph{elliptic} problems to perform a decomposition separately for the metric components $h_{00}$, $h_{0i}$ and $h_{ij}$ on the spatial (\emph{Riemannian}) hypersurfaces.

To keep some generality in this subsection the dimension of the spacetime is kept generic and equal to $n\geq 2$, one for the time coordinate plus $n-1$ for the spatial coordinates. 

\begin{proposition}[Past-compact tensor decomposition]\label{prop:decomposition}
For any smooth and past-compact symmetric (0,2)-tensor $h_{ab}\in\Gamma_{\mathrm{pc}}(T^*\mathbb{R}^n\otimes_s T^*\mathbb{R}^n)$, there are unique smooth and past-compact fields
\begin{align*}
    w\in C^{\infty}_{\mathrm{pc}}(\mathbb{R}^{n}),\qquad v_{a}^{\mathrm{T}}\in\Gamma_{\mathrm{pc}}(T^{\ast}\mathbb{R}^{n}),\qquad h_{ab}^{\mathrm{T}\mathrm{T}}\in\Gamma_{\mathrm{pc}}(T^{\ast}\mathbb{R}^{n}\otimes_{s}T^{\ast}\mathbb{R}^{n}) 
\end{align*}
with $\partial^{a}v_{a}^{\mathrm{T}}=0$ as well as $\partial^{a}h_{ab}^{\mathrm{T}\mathrm{T}}=0$ and $\eta^{ab}h_{ab}^{\mathrm{T}\mathrm{T}}=0$, such that $h_{ab}$ decomposes as 
\begin{align}\label{eq:TT-decomp}
    h_{ab}=h_{ab}^{\mathrm{S}}+h_{ab}^{\mathrm{V}}+h_{ab}^{\mathrm{T}\mathrm{T}},
\end{align}
where the scalar and vectorial components, $h_{ab}^{\mathrm{S}}$ and $h_{ab}^{\mathrm{V}}$, respectively, are given by
\begin{align*}
    h_{ab}^{\mathrm{S}} := 
 \Big(\partial_a\partial_b - \frac{1}{n}\,\eta_{ab}\,\square\Big) w
+ \frac{1}{n}\,\eta_{ab}\,h, \quad\qquad h_{ab}^{\mathrm{V}} := 2\partial_{(a}v^{\mathrm{T}}_{b)}=\partial_a v^{\mathrm{T}}_b +\partial_b v^{\mathrm{T}}_a.
\end{align*}
More explicitly, the scalar $w$ and vector $v_{a}^{\mathrm{T}}$ are fully determined by $h_{ab}$ and given by
\begin{align}\label{eq:decompComp}
    w = \frac{n}{n-1}\,
\mathsf{G}^{\circ 2}_{\mathrm{ret}}
\Big(\partial^a\partial^b h_{ab} - \frac{1}{n}\,\square h\Big),\qquad v_b^{\mathrm{T}} = 
\mathsf{G}_{\mathrm{ret}}
\,\Big(\partial^a h_{ab} - \frac{1}{n}\partial_b h - \frac{n-1}{n}\,\partial_b\square w \Big).
\end{align}
\end{proposition}

\begin{remark}
    A similar decomposition for symmetric $(0,2)$-tensor fields has been employed in the discussion of Anderson et al.~in \cite{Anderson}, see Appendix~A therein. The difference lies in the choice of perturbations and Green’s operators used in the construction: in our setting, we restrict to \emph{past-compact} tensor fields and make use of the \emph{retarded fundamental solution}, which yields a canonical (and hence \emph{unique}) decomposition for this class of metric perturbations. By contrast, the approach in \cite{Anderson} is formulated in terms of arbitrary symmetric $(0,2)$-tensors and an arbitrary (and hence non-unique) propagator. This naturally allows for additional freedom in the construction and therefore does not single out a unique decomposition. This distinction becomes particularly important in the subsequent analysis, where uniqueness will play a crucial role.
\end{remark}

\begin{proof}[Proof of Proposition~\ref{prop:decomposition}.] Let $h_{ab}\in\Gamma_{\mathrm{pc}}(T^*\mathbb{R}^n\otimes_s T^*\mathbb{R}^n)$ be arbitrary. Obtaining a decomposition of the form~\eqref{eq:TT-decomp} is equivalent to finding $w\in C^{\infty}_{\mathrm{pc}}(\mathbb{R}^{n})$ and $v_{a}^{\mathrm{T}}\in\Gamma_{\mathrm{pc}}(T^{\ast}\mathbb{R}^{n})$ with $\partial^{a}v_{a}^{\mathrm{T}}=0$ such that the $(0,2)$-tensor
    \begin{align*}
        h_{ab}^{\mathrm{T}\mathrm{T}}:=h_{ab}- \frac{1}{n}\,\eta_{ab}\,h-\Big(\partial_a\partial_b - \frac{1}{n}\,\eta_{ab}\,\square\Big) w-\partial_a v^{\mathrm{T}}_b -\partial_b v^{\mathrm{T}}_a
    \end{align*}
    is transverse and traceless. A straightforward computation shows that the requirement $\eta^{ab} h_{ab}^{\mathrm{T}\mathrm{T}}=0$ is equivalent to the condition $\partial^{a}v_{a}^{\mathrm{T}}=0$, while the requirement $\partial^{a}h_{ab}^{\mathrm{T}\mathrm{T}}=0$ yields the hyperbolic equation
    \begin{align*}
        \square v^{\mathrm{T}}_b = K_b(w)\qquad\text{with}\qquad K_b(w):=\partial^a h_{ab}- \frac{1}{n}\partial_b h - \frac{n-1}{n}\,\partial_b\square w .
    \end{align*}
    Hence, we are left to find a smooth and past-compact solution $(w,v_{a}^{\mathrm{T}})$ to the coupled system 
    \begin{align*}
        \begin{cases}
            \square v^{\mathrm{T}}_b &= K_b(w)\\
            \partial^{a}v_{a}^{\mathrm{T}}&=0
        \end{cases}.
    \end{align*}
    First of all, we note that the second equation implies that necessarily $\partial^{a}K_a(w)=0$, which results into the 4th order hyperbolic equation
    \begin{align*}
        \square^2 w 
= \frac{n}{n-1}\,\Big(\partial^a\partial^b h_{ab} - \frac{1}{n}\,\square h\Big).
    \end{align*}

    The unique smooth and past-compact solution $w\in C^{\infty}_{\mathrm{pc}}(\mathbb{R}^{n})$ to this equation with fixed $h_{ab}\in\Gamma_{\mathrm{pc}}(T^*\mathbb{R}^n\otimes_s T^*\mathbb{R}^n)$ can be written as 
    \begin{align*}
        w = \frac{n}{n-1}\,
\mathsf{G}^{\circ 2}_{\mathrm{ret}}
\Big(\partial^a\partial^b h_{ab} - \frac{1}{n}\,\square h\Big),
    \end{align*}
    as claimed. With this scalar at our disposal, we obtain the unique past-compact solution $v_{a}^{\mathrm{T}}\in\Gamma_{\mathrm{pc}}(T^{\ast}\mathbb{R}^{n})$ of the hyperbolic problem $\square v^{\mathrm{T}}_b= K_b(w)$, which reads
    \begin{align*}
        v_{b}^{\mathrm{T}}:=\mathsf{G}_{\mathrm{ret}}K_{b}(w)=\Big(\partial^a h_{ab} - \frac{1}{n}\partial_b h - \frac{n-1}{n}\,\partial_b\square w \Big).
    \end{align*}
    This yields the required decomposition~\eqref{eq:TT-decomp}. 
  
For uniqueness, suppose that $(w,v^{\mathrm{T}}_a,h^{\mathrm{TT}}_{ab})$ and 
$(\hat w,\hat v^{\mathrm{T}}_a,\hat h^{\mathrm{TT}}_{ab})$ are two triples of smooth and past-compact fields providing the required decomposition \eqref{eq:TT-decomp} of $h_{ab}$. Now, consider the differences
\[
\delta w := w - \hat w,\qquad
\delta v^{\mathrm{T}}_a := v^{\mathrm{T}}_a - \hat v^{\mathrm{T}}_a,\qquad
\delta h^{\mathrm{TT}}_{ab} := h^{\mathrm{TT}}_{ab} - \hat h^{\mathrm{TT}}_{ab}.
\]
By assumption, $\delta w$, $\delta v^{\mathrm{T}}_a$ and $\delta h^{\mathrm{TT}}_{ab}$ are past-compact. Moreover, $\delta v^{\mathrm{T}}_a$ is itself transverse, while $\delta h^{\mathrm{TT}}_{ab}$ is transverse and traceless. Taking the difference of the corresponding decompositions of $h_{ab}$, we obtain
\begin{align*}
    0&=h_{ab}-h_{ab}=(h_{ab}^{\mathrm{S}}-\hat{h}_{ab}^{\mathrm{S}})+(h_{ab}^{\mathrm{V}}-\hat{h}_{ab}^{\mathrm{V}})+(h_{ab}^{\mathrm{T}\mathrm{T}}-\hat{h}_{ab}^{\mathrm{T}\mathrm{T}})=\\&= \Big(\partial_a\partial_b - \frac{1}{n}\,\eta_{ab}\,\square\Big) \delta w
+\partial_a \delta v^{\mathrm{T}}_b +\partial_b \delta v^{\mathrm{T}}_a+\delta h_{ab}^{\mathrm{T}\mathrm{T}}.
\end{align*}
Now, taking the divergence of this equation, we obtain the hyperbolic problem
\begin{align*}
    \square \delta v_{b}^{\mathrm{T}}=\frac{1-n}{n}\square\partial_{b}\delta w.
\end{align*}
Moreover, since the left-hand side of this equation is divergence-free, by assumption, we conclude that $\square^{2}\delta w=0$. To sum up, $(\delta w,\delta v_{a}^{\mathrm{T}})$ is a solution to the coupled hyperbolic system
\begin{align*}
    \begin{cases}
        \square \delta v^{\mathrm{T}}_b &= \frac{1-n}{n}\,\partial_b\square \delta w\\
        \square^2 \delta w &= 0
    \end{cases}
\end{align*}
and since $(\delta w,\delta v_{a}^{\mathrm{T}})$ are smooth and past-compact, we conclude that $\delta w=0$ and $\delta v_{a}^{\mathrm{T}}=0$. As consequence, also $\delta h_{ab}^{\mathrm{T}\mathrm{T}}=0$, which concludes the proof.
\end{proof}

As discussed in the beginning of this section, it is convenient to discuss also the \emph{trace-reversed perturbation}, which in the $n$-dimensional setting is given by
\begin{align*}
    \bar{h}_{ab}=h_{ab} - \frac{2}{n}\eta_{ab} {h_c}^{c}.
\end{align*}
Being a past-compact symmetric $(0,2)$-tensor field in its own right, it admits a decomposition as in Proposition~\ref{prop:decomposition}. The following proposition follows immediately from the definition of the trace-reversal and Proposition~\ref{prop:decomposition}.

\begin{corollary}[Trace-reversal and decomposition]\label{Cor:decomposition}
    Let $h_{ab}\in\Gamma_{\mathrm{pc}}(T^{\ast}\mathbb{R}^{n}\otimes_{s} T^{\ast}\mathbb{R})$ be arbitrary and consider its trace-reversal $\bar{h}_{ab}$. Moreover, consider the decompositions of $h_{ab}$ and $\bar{h}_{ab}$ as in Proposition~\ref{prop:decomposition} with corresponding $w,\bar{w}\in C^{\infty}_{\mathrm{pc}}(T^{\ast}\mathbb{R}^{n})$ and $v_{a}^{\mathrm{T}},\bar{v}_{a}^{\mathrm{T}}\in\Gamma_{\mathrm{pc}}(T^{\ast}\mathbb{R}^{n})$. Then,
    \begin{align*}
        w=\bar{w}\qquad\text{and}\qquad v_{a}^{\mathrm{T}}=\bar{v}_{a}^{\mathrm{T}}.
    \end{align*}
    As a consequence, the trace-reversal is compatible with the decomposition in Proposition~\ref{prop:decomposition}, i.e.
    \begin{align*}
        \bar{h}_{ab} = \bar{h}_{ab}^{\mathrm{S}}  + \bar{h}_{ab}^{\mathrm{V}}+ \bar{h}_{ab}^{\mathrm{TT}}\qquad\text{with}\qquad \bar{h}_{ab}^{\mathrm{S}/\mathrm{V}/\mathrm{T}\mathrm{T}}=\overline{h_{ab}^{\mathrm{S}/\mathrm{V}/\mathrm{T}\mathrm{T}}},
    \end{align*}
    where $\overline{h_{ab}^{\mathrm{S}/\mathrm{V}/\mathrm{T}\mathrm{T}}}$ denotes the trace-reversal of $h_{ab}^{\mathrm{S}/\mathrm{V}/\mathrm{T}\mathrm{T}}$, respectively. In fact, it holds that $\bar{h}_{ab}^{\mathrm{V}}=h_{ab}^{\mathrm{V}}$ and $\bar{h}_{ab}^{\mathrm{T}\mathrm{T}}=h_{ab}^{\mathrm{T}\mathrm{T}}$, since these two components are traceless tensors.
\end{corollary}

The linearised Einstein equations inherit a linear gauge symmetry, originating from the diffeomorphism invariance on the level of the background spacetime. More precisely, in the setting of \emph{past-compact perturbations}, the gauge transformations are parameterised by $X_{a}\in\Gamma_{\mathrm{pc}}(T^{\ast}\mathbb{R}^{n})$ and given by $h_{ab}\mapsto h_{ab}+\mathcal{L}_{X}\eta_{ab}$, where $\mathcal{L}_{X}\eta_{ab}=2\partial_{(a}X_{b)}$ denotes the Lie-derivative of $\eta_{ab}$ along the vector field $X_{a}$, i.e.
\[
h_{ab}\mapsto h_{ab} + \partial_a X_b + \partial_b X_a. 
\]
On the level of trace-reversed perturbations, the gauge transformations act correspondingly as 
\[
\bar{h}_{ab}\mapsto \bar{h}_{ab} + \partial_a X_b + \partial_a X_b   -\frac{4}{n}\eta_{ab} \partial^c X_c.
\]

Now, since we restrict ourselves to past-compact gauge transformations, we note that also $X_a$ admits a unique decomposition into a \emph{transverse} and \emph{longitudinal part}, i.e.~for every $X_{a}\in\Gamma_{\mathrm{pc}}(T^{\ast}\mathbb{R}^{n})$ there exist unique $Y\in C^{\infty}_{\mathrm{pc}}(\mathbb{R}^{n})$ and $X_{a}^{\mathrm{T}}\in\Gamma_{\mathrm{pc}}(T^{\ast}\mathbb{R}^{n})$ with $\partial^{a}X_{a}^{\mathrm{T}}=0$, such that 
\[
X_a = X^{\mathrm{T}}_{a} + \partial_a Y.
\]
Indeed, for a given $X_{a}\in\Gamma_{\mathrm{pc}}(T^{\ast}\mathbb{R}^{n})$, this is equivalent to finding a scalar $Y\in C^{\infty}_{\mathrm{pc}}(\mathbb{R}^{n})$ such that $X_{a}^{\mathrm{T}}:=X_{a}-\partial_{a}Y$ is divergence-free. This, in turn, yields the hyperbolic equation $\square Y=\partial^{a}X_{a}$, whose unique smooth and past-compact solution is given by
\[
Y:=\mathsf{G}_{\mathrm {ret}} (\partial^a X_a).
\]

If $h_{ab}$ and $\bar{h}_{ab}$ are (uniquely) decomposed according to Proposition~\ref{prop:decomposition} and Corollary~\ref{Cor:decomposition} with components \begin{align*}
    (w,v_{a}^{\mathrm{T}})\in C^{\infty}_{\mathrm{pc}}(\mathbb{R}^{n})\times\Gamma_{\mathrm{pc}}(T^{\ast}\mathbb{R}^{n})\qquad\text{and}\qquad (\bar{w},\bar{v}_{a}^{\mathrm{T}})\in C^{\infty}_{\mathrm{pc}}(\mathbb{R}^{n})\times\Gamma_{\mathrm{pc}}(T^{\ast}\mathbb{R}^{n}),
\end{align*}
respectively, and if we decompose a gauge transformation $X_{a}\in\Gamma_{\mathrm{pc}}(T^{\ast}\mathbb{R}^{n})$ (uniquely) as $X_{a}=X_{a}^{\mathrm{T}}+\partial_{a}Y$ with $Y\in C^{\infty}_{\mathrm{pc}}(\mathbb{R}^{n})$ and $X_{a}^{\mathrm{T}}\in\Gamma_{\mathrm{pc}}(T^{\ast}\mathbb{R}^{n})$ satisfying $\partial^{a}X_{a}^{\mathrm{T}}=0$, as explained above, then uniqueness of the decomposition implies that the individual components of $h_{ab}$ and $\bar{h}_{ab}$ admit the following gauge transformation laws:
\begin{equation}\label{eq:gauge-transf-decomposition}
\begin{cases}
h &\mapsto\qquad h + 2 \square Y \\
w &\mapsto\qquad w + 2  Y \\
v^{\mathrm{T}}_b &\mapsto\qquad v^{\mathrm{T}}_b + X^{\mathrm{T}}_b  \\
h^{\mathrm{TT}}_{ab} &\mapsto\qquad h^{\mathrm{TT}}_{ab}
\end{cases}\qquad \text{and}\qquad\begin{cases}
\bar{h} &\mapsto\qquad \bar{h} - 2 \square Y \\
\bar{w} &\mapsto\qquad \bar{w} + 2  Y \\
\bar{v}^{\mathrm{T}}_b &\mapsto\qquad \bar{v}^{\mathrm{T}}_b + X^{\mathrm{T}}_b  \\
\bar{h}^{\mathrm{TT}}_{ab} &\mapsto\qquad \bar{h}^{\mathrm{TT}}_{ab}
\end{cases}
\end{equation}

In the following proposition, we show that the \emph{de Donder gauge condition} (cf.~Eq.~\eqref{eq:deDonderGauge}) provides a \emph{complete gauge fixing} in the setting of \emph{past-compact} metric perturbations. More precisely, given $h_{ab}\in\Gamma_{\mathrm{pc}}(T^{\ast}\mathbb{R}^{n}\otimes_{s} T^{\ast}\mathbb{R}^{n})$, its gauge equivalence class $[h_{ab}]_{\sim}$, where two past-compact tensors are identified if they differ by a linear gauge transformation generated by some $X_{a}\in\Gamma_{\mathrm{pc}}(T^{\ast}\mathbb{R}^{n})$, contains \emph{exactly one} representative satisfying the de Donder gauge condition $\partial^{a}\bar{h}_{ab}=0$.
\begin{proposition}[Complete gauge fixing]\label{prop:gauge-trasformation}
Let $h_{ab}\in\Gamma_{\mathrm{pc}}(T^{\ast}\mathbb{R}^{n}\otimes_{s}T^{\ast}\mathbb{R}^{n})$. Then, there is exactly one vector field $\tilde{X}_a \in  \Gamma_{\mathrm{pc}}(T^*\mathbb{R}^n)$, which makes the trace-reversal $\bar{h}_{ab}$ divergence-free, i.e.
\begin{align*}
    \partial^{a}\bar{h}_{ab}^{\prime}=0\qquad\text{with}\qquad \bar{h}_{ab}^{\prime}:=h_{ab}+\partial_a \tilde{X}_b + \partial_b \tilde{X}_a-\frac{4}{n}\eta_{ab}\partial^{c}\tilde{X}_{c}
\end{align*}
If $\bar{h}_{ab}$ is uniquely decomposed as in Proposition~\ref{prop:decomposition} (cf.~Corollary~\ref{Cor:decomposition}), the field $\tilde{X}_{a}$ is given by
\[
\tilde{X}_a^{\mathrm{T}} = - \bar{v}_a^{\mathrm{T}} , \qquad
\tilde{Y}= -\mathsf{G}_{\mathrm{ret}} \Big(    \frac{n-1}{2(n-2)}\square \bar{w} 
+ \frac{1}{2(n-2)}\,\bar{h} 
   \Big),\qquad \tilde{X}_{a}:=\tilde{X}_{a}^{\mathrm{T}}+\partial_{a}\tilde{Y}.
\]
\end{proposition}
\begin{proof}
We aim to find $\tilde{X}_{b}=\tilde{X}_{b}^{\mathrm{T}}+\partial_{b}\tilde{Y}\in\Gamma_{\mathrm{pc}}(T^{\ast}\mathbb{R}^{n})$ such that
\[
\bar{h}_{ab}'=\bar{h}_{ab}+\partial_a \tilde{X}_b+\partial_b \tilde{X}_a - \frac{4}{n}\eta_{ab}\partial^c \tilde{X}_c
\]
is divergence-free. In this case, utilising the decomposition given in Proposition \ref{prop:decomposition} and Corollary~\ref{Cor:decomposition} and, recalling the transformation rules under gauge transformation given in \Eq~\eqref{eq:gauge-transf-decomposition}, yields the two equations
\begin{align*}
0=\partial^a \bar{h}_{ab}' 
&= 
 \frac{n-1}{n}\partial_b\,\square \bar{w} +
 2\frac{n-1}{n}\partial_b\,\square \tilde{Y}
+ \frac{1}{n}\,\partial_b\,\bar{h} -\frac{2}{n}\,\partial_b\,\square \tilde{Y}
+\square (\bar{v}_b^{\mathrm{T}}+\tilde{X}_b^{\mathrm{T}})
\\
&= 
 \frac{n-1}{n}\partial_b\,\square \bar{w} 
+ \frac{1}{n}\,\partial_b\,\bar{h} 
+
 2\frac{n-2}{n}\partial_b\,\square \tilde{Y}
 +\square (\bar{v}_b^{\mathrm{T}}+\tilde{X}_b^{\mathrm{T}})\\
 0=\partial^b \partial^a \bar{h}_{ab}'&=\frac{n-1}{n}\square\square \bar{w} 
+ \frac{1}{n}\,\square\,\bar{h} 
+
 2\frac{n-2}{n}\square\square \tilde{Y},
\end{align*}
whose unique past-compact and smooth solution is given by $\tilde{X}_{b}^{\mathrm{T}}=-\bar{v}_{b}^{\mathrm{T}}$ and 
\[
\tilde{Y}= -\mathsf{G}_{\mathrm{ret}} \Big(    \frac{n-1}{2(n-2)}\square \bar{w} 
+ \frac{1}{2(n-2)}\,\bar{h}
   \Big).
\]
We conclude that $\tilde{X}_{b}=\tilde{X}_{b}^{\mathrm{T}}+\partial_{b}\tilde{Y}$ is the required vector field. Uniqueness follows from similar arguments as above.
\end{proof}

    As a next step, we observe that the vector component in the decomposition of Proposition~\ref{prop:decomposition} vanishes upon imposing the de Donder gauge. The decomposition therefore simplifies to include only the scalar and transverse-traceless sectors in this case.

\begin{proposition}[Vanishing of the vector component]\label{prop:decomp-divfree}
Let $\bar{h}_{ab}\in\Gamma_{\mathrm{pc}}(T^{\ast}\mathbb{R}^{n}\otimes_{s}T^{\ast}\mathbb{R}^{n})$ be such that $\partial^{a}\bar{h}_{ab}=0$. Then, the decomposition of Proposition \ref{prop:decomposition} simplifies to
\begin{align*}
\bar{w} =  - \frac{1}{n-1}\,\mathsf{G}_{\mathrm{ret}}( \bar{h}), \qquad\qquad
\bar{v}^{\mathrm{T}}_a =  0. 
\end{align*}
In particular, the scalar and vectorial components in the decomposition are given by
\[
\bar{h}_{ab}^{\mathrm{S}} = \frac{1}{n-1} \tau_{ab} \bar{h}, \qquad\qquad \bar{h}_{ab}^{\mathrm{V}} = 0,
\] 
where $\tau_{ab}\colon C^{\infty}_{\mathrm{pc}} (\mathbb{R}^n)\to \Gamma_{\mathrm{pc}} (T^*\mathbb{R}^n\otimes_{s} T^*\mathbb{R}^n)$ is the operator that acts on $f\in C^{\infty}_{\mathrm{pc}} (\mathbb{R}^n)$ as
\[
\tau_{ab} f := \eta_{ab} f -\partial_a\partial_b \mathsf{G}_{\mathrm{ret}} (f).
\]
\end{proposition}

\begin{proof}
    The first claim follows directly from the explicit expressions of $(\bar{w},\bar{v}_{a}^{\mathrm{T}})$ in terms of $\bar{h}_{ab}$ as specified in Eq.~\eqref{eq:decompComp}. On the other hand, this implies that $\bar{h}^{\mathrm{V}}_{ab}=0$ and similarly
\begin{align*}
\bar{h}_{ab}^{\mathrm{S}} &= 
 \Big(\partial_a\partial_b - \frac{1}{n}\,\eta_{ab}\,\square\Big) \bar{w}
+ \frac{1}{n}\,\eta_{ab}\,\bar{h}
\\
&= \frac{n}{n-1}\,
\partial_a\partial_b \mathsf{G}^{\circ 2}_{\mathrm{ret}}
\Big(\partial^c\partial^d \bar{h}_{cd}\Big)  
- \frac{1}{n-1} \eta_{ab} \mathsf{G}_{\mathrm{ret}}
\Big(\partial^c\partial^d \bar{h}_{cd}\Big)
-\frac{1}{n-1}\partial_a\partial_b \mathsf{G}_{\mathrm{ret}}( \bar{h})
+ \frac{1}{n-1}\,\eta_{ab}\,\bar{h}\\&=\frac{1}{n-1}\bigg(\eta_{ab}\bar{h}-\partial_{a}\partial_{b}\mathsf{G}_{\mathrm{ret}}(\bar{h})\bigg),
\end{align*}
where we used that $\partial^{a}\bar{h}_{ab}=0$, which proves the claim.
\end{proof}

\begin{remark}
    To sum up, once the de Donder gauge condition is imposed, either for a past-compact symmetric tensor field $h_{ab}$ or its trace-reversal $\bar{h}_{ab}$, the gauge freedom is completely fixed when restricting to past-compact gauge transformations. This is also the reason why only the trace $\bar{h}$ appears in the scalar component $\bar{h}_{ab}^{\mathrm{S}}$ of $h_{ab}$. 
\end{remark}

Having imposed the de Donder gauge condition, the decomposition reduces to scalar and transverse-traceless degrees of freedom. As a final step, we derive explicit expressions for the corresponding projectors, which will be useful for decomposing various tensors. To start with, let us consider the operator already introduced in Proposition \ref{prop:decomp-divfree}, i.e.
\begin{equation}\label{eq:tau}
\tau_{ab}\colon C^{\infty}_{\mathrm{pc}} (\mathbb{R}^n)\to \Gamma_{\mathrm{pc}} (T^*\mathbb{R}^n\otimes_{s} T^*\mathbb{R}^n),\qquad \tau_{ab} := \eta_{ab} -  \pa_a \pa_b \mathsf{G}_{\mathrm{ret}}.
\end{equation}
With this notation, we have the following.

\begin{proposition}\label{Prop:Projectors}
Consider $\Gamma_{\mathrm{pc}}^{\delta}(T^{\ast}\mathbb{R}^n\otimes T^{\ast}\mathbb{R}^n)$, which is the set of smooth, symmetric past-compact $2$-tensor which are also divergence free.
Consider the following operators on $\Gamma_{\mathrm{pc}}^{\delta}$, 
\begin{equation}
	\label{eq:projector-S}
	{{P^{(\mathrm{S})}}_{ab}}^{cd} := \frac{1}{n-1} {\tau}_{ab} {\tau}^{cd},
\end{equation}
which is called  {\bf scalar projector},
and 
\begin{equation}
	\label{eq:projector-T}
	{{P^{(\mathrm{TT})} }_{ab}}^{cd} := \frac{1}{2} \at {{\tau}_a}^c {{\tau}_b}^d + {{\tau}_a}^d {{\tau}_b}^c \ct - \frac{1}{n-1} {\tau}_{ab} {\tau}^{cd},
\end{equation}
which is called  {\bf transverse traceless projector} and where both
are given in terms of $\tau_{ab}$ introduced in \eqref{eq:tau}.
They are projectors which extract the scalar and transverse traceless part of $\bar{h}_{ab}\in  \Gamma_{\mathrm{pc}}^{\delta}$, that is, according to the decomposition introduced in Proposition \ref{prop:decomposition},
\[
{P^{(\mathrm{S})}_{ab}}^{cd} \,\bar{h}_{cd} = \bar{h}_{cd}^{\mathrm{S}}, \qquad  {P^{(\mathrm{TT})}_{ab}}^{cd}\, \bar{h}_{cd} = \bar{h}_{cd}^{\mathrm{TT}}\, .
\]  
\end{proposition}

\begin{remark}
    The projectors ${P^{(\mathrm{S})}_{ab}}^{cd}$ and ${P^{(\mathrm{TT})}_{ab}}^{cd}$ defined on $\Gamma_{\mathrm{pc}}^{\delta}(T^{\ast}\mathbb{R}^n\otimes T^{\ast}\mathbb{R}^n)$ are orthogonal with respect to the natural fibre metric on $T^{\ast}\mathbb{R}^n\otimes T^{\ast}\mathbb{R}^n$ induced by the background metric $\eta_{ab}$. Moreover, it holds that
    \begin{align*}
        {P^{(\mathrm{S})}_{ab}}^{cd}{P^{(\mathrm{TT})}_{cd}}^{ef}={P^{(\mathrm{TT})}_{ab}}^{cd}{P^{(\mathrm{S})}_{cd}}^{ef}=0\,,\qquad {P^{(\mathrm{S})}_{ab}}^{cd}+{P^{(\mathrm{TT})}_{ab}}^{cd}=\delta^{c}_{a}\delta_{b}^{d}\quad\text{on}\quad \Gamma_{\mathrm{pc}}^{\delta}(T^{\ast}\mathbb{R}^n\otimes T^{\ast}\mathbb{R}^n)\,.
    \end{align*}
\end{remark}

\begin{proof}[Proof of Proposition~\ref{Prop:Projectors}.]
When applied to past-compact and divergence-free symmetric $(0,2)$-tensors, it holds that
\begin{equation}\label{eq:algbra-tau}
\partial^b\tau_{ab}=0, \qquad {\tau_{a}}^a = n-1,\qquad \tau_{ab}\tau^{ab}=n-1, \qquad \tau_{ab}{\tau^{b}}_d = \tau_{ad}.
\end{equation}
In particular, we obtain the identity
\[
\frac{1}{2} \at {{\tau}_a}^c {{\tau}_b}^d + {{\tau}_a}^d {{\tau}_b}^c \ct \bar{h}_{cd} = \bar{h}_{ab}\,.
\]
Hence, using Eq.~\eqref{eq:algbra-tau}, a direct inspection shows that
\[
{P^{(\mathrm{S})}_{ab}}^{cd} +  {P^{(\mathrm{TT})}_{ab}}^{cd} = {\mathrm{I}_{ab}}^{cd}, \qquad  {P^{(\mathrm{S})}_{ab}}^{cd}{P^{(\mathrm{TT})}_{cd}}^{ef} = 0 = {P^{(\mathrm{TT})}_{ab}}^{cd}{P^{(\mathrm{S})}_{cd}}^{ef}\,,
\] 
where ${\mathrm{I}_{ab}}^{cd}=\delta_{a}^{c}\delta_{b}^{s}$ denotes the identity on $\Gamma_{\mathrm{pc}}^\delta$. Moreover, on divergence-free $(0,2)$-tensors, it holds that
\[
{P^{(\mathrm{S})}_{ab}}^{cd} \,\bar{h}_{cd} = \frac{1}{n-1}\tau_{ab} {\bar{h}_c}^c = \bar{h}^{\mathrm{S}}_{cd}\,,
\]
where we used Proposition \ref{prop:decomp-divfree} in the last equality. The same proposition further implies that 
\[
{P^{(\mathrm{TT})}_{ab}}^{cd} \bar{h}_{cd} = \bar{h}^{\mathrm{TT}}_{ab}\,,
\]
which shows that ${P^{(\mathrm{S})}_{ab}}^{cd}$ and ${P^{(\mathrm{TT})}_{ab}}^{cd}$ are the required projectors.
\end{proof}

\subsection{The four-dimensional case and decomposition of some tensors}

For later purposes, we specialise the analysis now to the $(1+3)$-dimensional case. After imposing the de Donder gauge, i.e.~the gauge condition Eq.~\eqref{eq:deDonderGauge}, we obtain the decomposition
\begin{equation}
\label{eq:metric-decomposition}
	\bar{h}_{ab} = \bar{h}_{ab}^{{\mathrm{S}}} + \bar{h}_{ab}^{{\mathrm{TT}}},
\end{equation}
where we recall that the vector-component vanishes by Proposition~\ref{prop:decomp-divfree} and where
\begin{align*}
	\bar{h}_{ab}^{\mathrm{S}}  &:= {P^{(\mathrm{S})}_{ab}}^{cd}\bar{h}_{cd}  
	= \frac{1}{3} \tau_{ab} {\bar{h}^c}_c\,, 
	\qquad  
	&{{\bar{h}^{\mathrm{S}}}_c}^c = {\bar{h}_c}^c\;,\\
	\bar{h}_{ab}^{{\mathrm{TT}}}  &:= {P^{(\mathrm{TT})}_{ab}}^{cd}\bar{h}_{cd}
	= \bar{h}_{ab} - \frac{1}{3} \tau_{ab} {\bar{h}^c}_c\,,
	&{{\bar{h}^{{\mathrm{TT}}}}_c}^c = 0\,
	, 
\end{align*}
are the \emph{gauge-invariant components} of the metric perturbation. These components correspond to the linearly independent scalar and tensorial degrees of freedom of the model, because they are related to the trace and the transverse traceless part of ${\bar{h}}_{ab}$. 

In particular, notice that if $\bar{h}_{ab}$ satisfies the de Donder gauge, we find
\begin{equation}\label{eq:projectors-onhab}
{P^{(\mathrm{S})}_{ab}}^{cd}{h}_{cd}  =-\frac{1}{2}  \bar{h}^{\mathrm{S}}_{cd}, \qquad {P^{(\mathrm{TT})}_{ab}}^{cd}{h}_{cd} = h^{\mathrm{TT}}_{ab}.  
\end{equation}

Based on the decomposition introduced in Eq.~\eqref{eq:metric-decomposition}, and after imposing the de Donder gauge, as will be shown in detail below, all contributions entering the linearised semiclassical Einstein equations can be decoupled into scalar and transverse-traceless sectors, i.e.
\begin{equation}
	\label{eq:Einstein-tensor-1-decomposition}
	{G}_{ab}^{(1)} = -\frac{1}{2} \square \bar{h}_{ab} = {{G}_{ab}^{(1)}}^{\mathrm{S}} + {{G}_{ab}^{(1)}}^{\mathrm{TT}} = -\frac{1}{2} {{P}^{(\mathrm{S})}_{ab}}^{cd} \square \bar{h}_{cd} - \frac{1}{2} {{P}^{(\mathrm{TT})}_{ab}}^{cd} \square \bar{h}_{cd},
\end{equation}
where the scalar and transverse-traceless sectors are given by
\begin{subequations}\label{eq:Einstein-tensor-1}
\begin{align}
	{{G}_{ab}^{(1)}}^{\mathrm{S}} &= -\frac{1}{2} \square \bar{h}_{ab}^{\mathrm{S}} = - \frac{1}{6} \tau_{ab} \square \bar{h}\,; \\
	{{G}_{ab}^{(1)}}^{\mathrm{TT}} &= -\frac{1}{2} \square \bar{h}_{ab}^{\mathrm{TT}}\,.
\end{align}
\end{subequations}

The linearised local curvature tensors $I$ and $J$, which appears in Eq.~\eqref{eq:T_ab-omega} and which we recall in Eq.~\eqref{eq:def-I-J} of the appendix, can accordingly be written as
\begin{subequations}\label{eq:local-curvature-tensor-A}
\begin{align}
	I^{(1)}_{ab} &= -2 \square {{G}_{ab}^{(1)}}^{\mathrm{TT}} =  \square\square \bar{h}_{ab}^{\mathrm{TT}}\,; \\
	J^{(1)}_{ab} &= 6 \square {{G}_{ab}^{(1)}}^{\mathrm{S}} = \square\square \bar{h}_{ab}^{\mathrm{S}}\,.
\end{align}
\end{subequations}
Note that $I_{ab}$ depends exclusively on ${\bar{h}}^{\mathrm{TT}}_{ab}$, while $J_{ab}$ depends only on ${\bar{h}}^{\mathrm{S}}_{ab}$. Moreover, although $\tau_{ab}$ is in general nonlocal, the presence of the d’Alembertian operator $\square$ acting on ${\bar{h}}^{\mathrm{TT}}_{ab}$ and ${\bar{h}}^{\mathrm{S}}_{ab}$ removes these nonlocal contributions. As a consequence, ${{G}_{ab}^{(1)}}^{\mathrm{S}}$, ${{G}_{ab}^{(1)}}^{\mathrm{TT}}$, $I^{(1)}_{ab}$, and $J^{(1)}_{ab}$ are all local (differential) operators acting on $\bar{h}_{ab}$.

\section{Linearised semiclassical Einstein equations}\label{sec3}

The central aim of this chapter is to derive the {\bf linearised semiclassical Einstein equations} in the form~\eqref{eq:FPmetric} on a Minkowski background. We begin by imposing that the vacuum solution $(\eta_\ab,\omega_0)$ solves the semiclassical Einstein-Klein-Gordon system,
$$ G_\ab^{(0)}:= G_\ab[\eta_\ab]  = \kappa \langle \Wick{T_{\ab}[\eta_\ab]}\rangle_{\omega_0} \,. $$
This will be achieved once the renormalisation constants $\alpha_i$ in \Eq~\eqref{eq:stress-tensor} are chosen appropriately. Following the strategy outlined in subsection~\ref{subsec:Nonlocsemi} of the introduction, we can then set a forcing problem, and using the decomposition obtained in Section~\ref{se:metric-perturbation-decomposition}, we can further decouple the Einstein-Klein-Gordon system using the quantum M\o ller map. The price to pay is that we have to match the renormalisation constant $\alpha_1$ with ${\tilde{\alpha}}_1$, which is the one obtained using the quantum M\o ller map. The main result of the chapter is presented in the final section, where we employ the \emph{principle of general local covariance} together with the \emph{principle of perturbative agreement} to fix some of the renormalisation constants that arise at linear order in this construction.

In this section, we shall use the following notation: given a state $\omega$ we denote its two-point correlation distribution by $\omega_2$. With this notation,  the corresponding two-point distribution of the Poincar\'e vacuum $\omega_0$ will be denoted by $\omegazd$.

\subsection{Background solution}
The aim of this section is to find the suitable renormalisation constant that makes the vacuum solution satisfy
$$ G_\ab^{(0)}= \kappa \langle \Wick{T_{\ab}[\eta_\ab]}\rangle_{\omega_0} =:\kappa \langle \Wick{T_{\ab}^{(0)}}\rangle_{\omega_0} \,. $$
We begin by noticing that 
the semiclassical Einstein equations simplifies to 
\begin{equation}
\label{eq:SEE-0}
	0 =  G^{(0)}_{ab} = \kappa \omega_0(\Wick{T^{(0)}_{ab}}).
\end{equation}
This equation, in turn, simply becomes a constraint on the renormalisation constants contained in $T_{ab}$, which are $\alpha_1$,  $\alpha_2$, $\alpha_3$, $\alpha_4$. Actually, on Minkowski spacetime, the curvature tensors multiplying $\alpha_2$, $\alpha_3$ and $\alpha_4$ as appearing in \eqref{eq:T_ab-omega} vanish. Therefore, these renormalisation freedoms do not appear in \eqref{eq:SEE-0}. Furthermore, $\alpha_2$ is a renormalisation of the Newton constant which is fixed by experiments. The only relevant constant is $\alpha_1$, which plays the role of a cosmological constant.\\

Although $\alpha_1$ gets fixed in this way, the freedom of choosing the corresponding constants in $\Wick{\Phi}$ and $\Wick{\Phi_{ab}}$, keeping their linear combination equal to $\alpha_1$, is left (see Appendix \ref{app:reno-freedom} for further details). We choose to fix this remaining freedom in a way that the vacuum expectation values of $\Wick{\Phi}$ and $\Wick{\Phi_{ab}}$ vanish. We have the following.

\begin{proposition}\label{prop:background}
Consider the Klein-Gordon field of mass $m$ and generic coupling to the scalar curvature $\xi$ propagating on a Minkowski spacetime. Assume the theory to be in vacuum $\omega_0$. We have 
	\[
	\omega_0(\Wick{\Phi}) = [\mathsf{w}]+d_1 m^2, \qquad \omega_0(\Wick{\Phi_{ab}}) = [\partial_a\partial_b\mathsf{w}]+c_1 m^4\eta_{ab}\,,
	\]
	where $\mathsf{w}(x,y) = \omegazd(x,y)-\mathcal{H}(x,y)$
    is the smooth part of the two-point function, $[\cdot]$ denotes the coinciding point limit and $d_1$, $c_1$ are the regularisation constants that have a role on the Minkowski background. 
    Here, $\mathcal{H}$ is the Hadamard parametrix.
	If $d_1$ and $c_1$ are such that 
	\[
	\begin{aligned}
d_1  &= -\frac{1}{16 \pi^2} \left(-1+ 2\gamma + \log \left(\frac{m^2}{2\mu^2}\right)\right)
\\
c_1 &= 
-\frac{1}{64 \pi^2} \left({-\frac{5}{2}} + 2\gamma + \log \left(\frac{m^2}{2\mu^2}\right)
\right),
\end{aligned}
\] 
with $\gamma$ the Euler-Mascheroni constant, the following expectation values vanish
	\begin{equation*}
	\omega_0(\Wick{\Phi}) = 0 , \qquad \omega_0(\Wick{\Phi_{ab}})=0
	\end{equation*}
	and thus
	\[
	\omega_0(\Wick{T^{(0)}_{ab}})=0.
	\]
	With these choices, the renormalisation parameter $\alpha_1$ given in \Eq \eqref{eq:T_ab-omega} is
	\begin{equation*}
	\alpha_1  =\frac{1}{64 \pi^2} \left( - \frac{3}{2} + 2\gamma + \log \left(\frac{m^2}{2\mu^2}\right)
	\right).
	\end{equation*}
and the pair $(\eta_{ab},\omega_0)$ is a solution of \Eq \eqref{eq:SEE-0}.
\end{proposition}
\begin{proof}
The two-point function of the  Poincar\'e vacuum in four-dimensional Minkowski spacetime is given in terms of the halved squared geodesic distance 
\[
    \sigma(x,y)=\frac{1}{2} (x-y)_\mu (x-y)^{\mu}
\]
and takes the form 
\[
    \omegazd = \lim_{\epsilon\to 0^+}\frac{1}{4\pi^2} \frac{m}{\sqrt{2\sigma_\varepsilon}}
    K_1\left(m \sqrt{2\sigma_\varepsilon}\right),
\]
where the limit is taken in the distributional sense and $\sigma_{\epsilon}= \sigma(x,y) + \mathrm{i} \epsilon (x_0-y_0)$ (see, e.g., the appendix of \cite{BDF}). It can be expanded as follows:
\[
\begin{aligned}
\omega_0 = & \frac{1}{8\pi^2 \sigma_\varepsilon} + \left(\frac{m^2}{16 \pi^2} + \frac{m^4}{64 \pi^2} \sigma + O(\sigma^2)  \right)\log (\sigma_\varepsilon \mu^2) +
\\
&  +\frac{m^2}{16 \pi^2} \left(-1+ 2\gamma + \log \left(\frac{m^2}{2\mu^2}\right)\right)
+  \frac{m^4}{64 \pi^2} \left({-\frac{5}{2}} + 2\gamma + \log \left(\frac{m^2}{2\mu^2}\right)
\right)\sigma + O(\sigma^2),
\end{aligned}
\]
where $\sigma(x,0)= \frac{(-t^2+|\mathbf{x}|^2)}{2}$, and $\mu$ is the mass parameter that appears in the Hadamard singularity. The first two terms are in $\mathcal{H} = (8\pi^2 \sigma_\epsilon)^{-1} + \mathsf{v} \log(\mu^2\sigma_\epsilon)$, whereas the last term is $\mathsf{w}$. Thus, we obtain
\begin{equation}
\label{eq:wwab0}
\begin{aligned}
    [\mathsf{w}] &= \frac{m^2}{16 \pi^2} \left(-1+ 2\gamma + \log \left(\frac{m^2}{2\mu^2}\right)\right), \\
    [\partial_a\partial_b \mathsf{w}] &= \frac{m^4}{64 \pi^2} \left({-\frac{5}{2}} + 2\gamma + \log \left(\frac{m^2}{2\mu^2}\right) \right) \eta_{ab},
\end{aligned}
\end{equation}
after using $[\partial_a\partial_b \sigma]=\eta_{ab}$. Furthermore,
\[
    [\mathsf{w}_0]= \frac{m^2}{16\pi^2}, \qquad [\mathsf{w}_1] = \frac{m^4}{64\pi^2}.
\]
Hence, we may evaluate the expectation value of the stress-energy tensor on the  Poincar\'e vacuum
\begin{equation}\label{eq:T0}
\begin{aligned}
\omega_0(\Wick{T^{(0)}_{ab}})=0
 & =
 - [\partial_a\partial_b \mathsf{w}]- \frac{1}{2} \eta_{ab} \at  
 m^2  [\mathsf{w}]  -  [\partial_c\partial^c \mathsf{w}]  \ct 
 +\alpha_1 m^4 \eta_{ab} + \frac{1}{4\pi^2} [\mathsf{v}_1] \eta_{ab}
 \\
 & =
\left( 
-\frac{1}{64 \pi^2} \left( {-\frac{3}{2}} + 2\gamma + \log \left(\frac{m^2}{2\mu^2}\right)
\right)
+
\alpha_1
 \right) m^4\eta_{ab}.
\end{aligned}
\end{equation}
The thesis directly follows from \Eq \eqref{eq:wwab0} and from the analysis of
$\omega_0(\Wick{T^{(0)}_{ab}})=0$ 
given in \Eq \eqref{eq:T0}.
\end{proof}

\subsection{Linearisation of the expectation value of the stress-energy tensor}
Following the algebraic viewpoint, in order to linearise the expectation value of the stress-energy tensor, we need to connect the $*$-algebra of the quantum theory on $(\mathcal{M},g_{ab})$ with the corresponding one on $(\mathcal{M},\eta_{ab})$. The $*$-algebra $\mathcal{A}(\mathcal{M},g_{ab})$ we refer to is generated by the local Wick polynomials, endowed with the quantum product and with an involution, see, e.g., \cite{HW01,BDF,Rejzner,AdvancesAQFT,Hack}. Then, it is possible to construct an operator which maps the generators of $\mathcal{A}(\mathcal{M},g_{ab})$ to those of $\mathcal{A}(\mathcal{M},\eta_{ab})$ and later extend it to a $*$-isomorphism \cite{DHP}. This operator descends from perturbation theory, and relates (local) observables on $(\mathcal{M},g_{ab})$ with (possibly nonlocal) observables on $(\mathcal{M},\eta_{ab})$. 
The connection with perturbation theory is established as follows. Consider the quadratic Lagrangian density of the theory on $(\mathcal{M},g_{ab})$ given by
\[
\mathcal{L}[g_{ab}] := \frac{1}{2} \phi \mathcal{P}_g\phi  \sqrt{-g}, \qquad \mathrm{for}\qquad\mathcal{P}_g := \square_g-m^2 -\xi R\,,
\]
where $g=\mathrm{det}(g_{ab})$ denotes the metric determinant, and the corresponding Lagrangian $\mathcal{L}[\eta_{ab}]$ of the theory on $(\mathcal{M},\eta_{ab})$. Then, their difference 
\begin{equation*}
    \mathcal{L}_I=\mathcal{L}[g_{ab}]-\mathcal{L}[\eta_{ab}],
\end{equation*}
which is quadratic in the quantum fields and past-compact in our setting, is considered as an interacting Lagrangian and allows us to define the map
\[
    \mathcal{R}_{\mathcal{L}_I}:\mathcal{A}(\mathcal{M},g_{ab})\to \mathcal{A}(\mathcal{M},\eta_{ab})
\]
in terms of a convergent perturbative expansion. $\mathcal{R}_{\mathcal{L}_I}$ realises the $*$-isomorphism and can be constructed in two different ways, usually known in the literature as {\bf classical M\o ller map} or {\bf quantum M\o ller map}; in fact, these two methods must agree by the principle of perturbative agreement \cite{HW05,DHP}. We shall discuss in Appendix \ref{se:source-perturbation}, and more specifically in Appendix \ref{se:source-cutoff}, how to make $h$ past-compact with the suitable use of a time cutoff function.  The quantum M\o ller map, also known as {\bf Bogoliubov map}, is a linear map on the generators of $\mathcal{A}(\mathcal{M},g_{ab})$ constructed out of the $\mathsf{S}$-matrix of $\mathcal{L}_{I}$ (time ordered exponential of $\mathcal{L}_I$) which respects the $*$-operation and the product. Its action on a local field $A$ is
\begin{equation}
    \label{eq:Bogoliubov}
    \mathsf{R}_{V}(A) := -\mathrm{i}\left.\frac{d}{d\lambda}\mathsf{S}(V)^{-1}\mathsf{S}(\lambda A+V)\right|_{\lambda=0},
\end{equation}
where $\mathsf{S}(A):=\mathscr{T}(\exp(i\mathscr{T}^{-1}A))$ is the time ordered exponential of $iA$ and
\[
V := \int_{\mathcal{M}}  (\mathcal{L}_I)f d^4 x  = 
\int_{\mathcal{M}} \frac{1}{2}\left(\phi(\mathcal{P}_{g})\phi \sqrt{-g}- \phi(\mathcal{P}_{\eta})\phi \sqrt{-\eta}\right) f  d^4x
\]
where $f\in C_{\mathrm{sc}}^\infty(\mathcal{M})\cap C_{\mathrm{fc}}^\infty(\mathcal{M})$ is equal to $1$ on a sufficiently large domain. In the particular case of interest here, if we let $x \in \mathcal{M}$ the point where we compute $\Wick{T_{ab}}(x)$, we require $f$ equal to $1$ on a domain that contains $J^{-}(x)$. Notice that with this choice, $\mathcal{L}_I f$ is of compact support because $\mathcal{L}_I$ is past-compact. Finally, notice that by the causal factorisation property of the Bogoliubov map, the result does not depend on the particular form $f$. \\

States, in the algebraic framework, are described by positive normalised linear functionals over $\mathcal{A}(\mathcal{M},g_{ab})$. Hence, if we have a state $\omega$ on $\mathcal{A}(\mathcal{M},\eta_{ab})$, we can pull it back via $\mathcal{R}_{\mathcal{L}_I}$ to a state on $\mathcal{A}(\mathcal{M},g_{ab})$. Usually, the action of $\mathcal{R}_{\mathcal{L}_I}$ is best described when it is realised by the {\bf classical M\o ller map} \cite{DHP, Dappiaggi, HW05}.
In this case, its action on the two-point function of the state is 
\[
\mathsf{R}_{\mathcal{L}_I}^*\omega(f_1,f_2) = \omega(\mathsf{R}_{\mathcal{L}_I} f_1 , \mathsf{R}_{\mathcal{L}_I} f_2), \qquad f_i\in C^{\infty}_c(M).
\] 
The classical M\o ller map on a smooth compactly supported function $f$ is defined as
\[
\mathsf{R}_{\mathcal{L}_I} f :=  (1-\mathcal{L}_I^{(1)}\mathsf{G}_{\mathrm{adv}}^{\mathcal{L}_I}) f,
\]
and it is given in terms of $\mathsf{G}_{\mathrm{adv}}^{\mathcal{L}_I}$  which is the advanced fundamental solution of the linear equation of motion associated to $\mathcal{L}[g_{ab}]$.\\

The stress-energy tensor $\Wick{T_{ab}[g_{ab}]}$ is constructed out of the metric, $\Phi$ and $\Phi_{ab}$. Both $\Phi$ and $\Phi_{ab}$ are elements of $\mathcal{A}(\mathcal{M},g_{ab})$ and can be mapped to $\mathcal{A}(\mathcal{M},\eta)_{ab}$ by means of $\mathsf{R}_{V}$. The state $\omega$ we use is a linear functional on $\mathcal{A}(\mathcal{M},\eta_{ab})$, hence, if we evaluate the expectation values of $\Wick{T_{ab}[g_{ab}]}$ in $\mathcal{R}_{\mathcal{L}_I}^*\omega$, the pull-back under $\mathcal{R}_{\mathcal{L}_I}$, we have 
\[
\expval{\Wick{T_{ab}[g_{ab}]}}_{\mathcal{R}_{\mathcal{L}_I}^*\omega} =
\mathcal{R}_{\mathcal{L}_I}^*\omega((\Wick{T_{ab}[g_{ab}]})).
\]
We now discuss properties of the linearisation of this formula around 
$\expval{\Wick{T_{ab}[\eta_{ab}]}}_{\omega_{0}}$ where $\omegaz$ is the  Poincar\'e vacuum.
We need to linearise three objects each of which will give a separate contribution to the linearised expectation value of $\Wick{T_{ab}}$ 
\begin{itemize}
\item[(i)] the state on $\mathcal{A}(\mathcal{M},\eta_\ab)$ admits the expansion $\omega=\omegaz +  \omegau + O(2)$. The two-point function $\omegaud$ is smooth and adds to $\mathsf{w}$, the smooth part of $\omegazd$. In the following discussion, the contribution at first order due to $\omegau$ is included in a redefinition of the source term (Appendix \ref{se:source-omega1}).

\item[(ii)] the Bogoliubov map $\mathsf{R}_{V}(F)$ for any $F \in \mathcal{A}(\mathcal{M},g_{ab})$ can be expanded as
\begin{equation}
    \label{eq:Bogoliubov-lin}
        \mathsf{R}_{V}(F) = F + \mathrm{i} \left( \mathscr{T}(V F)_{\mathrm{reg}} - V   F\right) + O(2),
\end{equation}
where $\mathscr{T}(V F)_{\mathrm{reg}}$ is the regularised time ordered product of $V$ and $F$, while $VF$ is the quantum product. In addition, $V$ admits a power expansion in the metric perturbation. Recalling that 
\[
T_{ab}= - \frac{2}{\sqrt{-g}} \frac{\delta \mathcal{L}}{\delta g^{ab}}, 
\]
the first order expansion of $V$ is 
\begin{equation*}
V = -\frac{1}{2}\int  f h^{cd} T^{(0)}_{cd}(y) d \mu_\eta + O(2).
\end{equation*}

\item[(iii)] the way in which the classical $T_{ab}$, and thus $\Wick{T_{ab}}$, depend on $h$ to first order. In particular
\[
T_{ab}=T^{(0)}_{ab}+T^{(1)}_{ab} + O(2). 
\]
For any observable, depending on $g_\ab=\eta_\ab+h_\ab$, the operator
\[
\mathfrak{L}:F[g_{ab}]\mapsto F[\eta_{ab}]+ \int h^{ab}(y)\left.\frac{1}{\sqrt{|g|}}\frac{\delta F}{\delta g^{ab}(y)}\right|_{g=\eta} d\mu_\eta(y)
\] 
extracts up to the linear contribution in $h$. Hence, ${T}^{(0)}_{ab}(x) = \left.-\frac{2}{\sqrt{|g|}}\frac{\delta \mathcal{L}}{\delta g^{ab}(x)}\right|_{g^{ab}=\eta^{ab}}$ and
\begin{align*}
T^{(1)}_{ab}(x)
&=-\frac{2}{\sqrt{|\eta|}}
\int \left.\frac{\delta^2 \mathcal{L}} {\delta g^{cd}(y) \delta g^{ab}(x)}\right|_{g^{ab}=\eta^{ab}}h^{cd}(y) \frac{1}{\sqrt{|\eta|}} d\mu_\eta(y)
\\
&=
\left(\frac{\delta T^{(0)}_{ab}(x)}{\delta g^{cd}(x)}-\frac{1}{2} T^{(0)}_{ab}(x) g_{cd}(x) \right)h^{cd}(x).
\end{align*}
\end{itemize} 
Combining these observations, we have proven the following proposition.
\begin{proposition}\label{prop:T-linearised}
Consider a metric $g_{ab}=\eta_{ab}+h_{ab}$, past-compact perturbation of the Minkowski background, and a Hadamard state $\omega=\omega_0+\tilde \omega$ on $\mathcal{A}(\mathcal{M},\eta_\ab)$, perturbation of the  Poincar\'e vacuum $\omega_0$. 
Assume that $h_{ab} \in \Gamma_{\mathrm{pc}} (T^*\mathbb{R}^4\otimes_{s} T^*\mathbb{R}^4)$. Then, up to first order in $h_{ab}$ and $\tilde\omega$, the expectation value of the stress-energy tensor $\Wick{T_{ab}}$ is
\begin{align*}
\langle
\Wick{T_{ab}(x)}
\rangle_\omega
&= (\omega_0+\tilde\omega)(\mathsf{R}_V(\Wick{T^{(0)}_{ab}(x)}+\Wick{T^{(1)}_{ab}(x)})) + O(2) \\
&= \omega_0(\Wick{{T}^{(0)}_{ab}(x)}) + \omega_0(\Wick{T^{(1)}_{ab}(x)})
+ \tilde\omega(\Wick{{T}^{(0)}_{ab}(x)})
\\
&- \frac{\mathrm{i}}{2} \int \omega_0\left(\mathscr{T}(\Wick{{T}^{(0)}_{cd}(y)} \Wick{{T}^{(0)}_{ab}(x)})_{\mathrm{reg}} - \Wick{{T}^{(0)}_{cd}(y)}  \Wick{{T}^{(0)}_{ab}(x)}\right)h^{cd}(y) d\mu_\eta(y) + O(2)\,,
\end{align*}
where $\mathscr{T}(\Wick{T^{(0)}_{cd}(y)} \Wick{T^{(0)}_{ab}(x)})_{\mathrm{reg}}$ denotes the regularised time ordering product.
\end{proposition}

We recall, and later carefully discuss, that in the analysis of $\mathscr{T}(\Wick{T^{(0)}_{cd}(y)} \Wick{T^{(0)}_{ab}(x)})_{\mathrm{reg}}$ a renormalisation procedure needs to be employed \cite{HW02,BF}. This extra renormalisation freedom is of the expected form according to \cite{HW05}, and the anomaly $A$ given in \eqref{eq:Anomaly} gets also corrected by linear contributions in $h$.

\subsubsection{Preliminary discussion}
From now on, starting from Proposition \ref{prop:T-linearised}, we shall specialise the expansion up to linear order of the semiclassical Einstein equation, choosing Minkowski spacetime as background geometry. The choice of background geometry, as solution of the semiclassical Einstein equations (see Proposition \ref{prop:background}), gives
\[
    \omega_0(\Wick{{T}^{(0)}_{ab}(x)})=0.
\]
With similar reasoning:
\begin{equation*}
    \begin{aligned}
       \omega_0(\Wick{T^{(1)}_{ab}})
        =&\,\omegaz \left(\left(\xi-\frac{1}{4} \right) \square \Wick{\Phi}-\frac{1}{2}m^2\Wick{\Phi}+\frac12\Wick{{\Phi_c}^c} \right) h_{ab} \\
        +&\,\omegaz\left(\left(\xi-\frac{1}{4} \right) \partial_a\partial_b \Wick{\Phi}+\frac12{\Wick{\Phi_{ab}}}  \right) {h^{c}}_c -\frac{1}{2}\omegaz\left(  \Wick{T^{(0)}_{ab}} \right) {h^{c}}_c
\end{aligned}
\end{equation*}
With the choices of the renormalisation freedom listed in Proposition \ref{prop:background}, it holds that 
\[
   \omegaz(\Wick{T^{(1)}_{ab}(x)}) =0.
\]

Therefore, only two contributions in $\langle \Wick{T_{ab}}\rangle_\omega$ are left. On the one hand, the first one, can be included as a contribution to the source term in the linearised semiclassical equation
\[
    S'_{ab}(x)=\tilde\omega(\Wick{T_{ab}^{(0)}(x)})+ S_{ab}(x).
\]
According to the analysis presented in the introduction, see also Appendix \ref{se:source-perturbation} and in particular Appendix \ref{se:source-omega1} for further details, we recall that the source $S'_{ab}$ can be chosen to be past-compact keeping the information of the state perturbation in the future, with a suitable use of a time cutoff. The second term is intrinsically nonlocal and will be analysed in the next section.

\subsubsection{The nonlocal contribution to $\langle \Wick{T_{ab}}\rangle_\omega$} \label{sec:nonlocalcontr}
We recall, by our previous discussion, that the linearised expectation value to the stress-energy tensor is 
\begin{equation}
\label{eq:linearisedT1}
    \langle \Wick{T_{ab}(x)} \rangle_{\omega} = N_{ab}(x) +\omega_0({T}^{(1)}_{ab}(x))+{\tilde\omega}({T}^{(0)}_{ab}(x)),
\end{equation}
where the nonlocal contribution to $\langle\Wick{T_{ab}}\rangle_\omega$ at linear order is denoted, for simplicity, by
\begin{equation}
\label{eq:def-N}
    N_{ab}(x):=- \frac{\mathrm{i}}{2} \int\omega_0\left(\mathscr{T}(\Wick{{T}^{(0)}_{cd}(y)} \Wick{{T}^{(0)}_{ab}(x)})_{\mathrm{reg}} - \Wick{{T}^{(0)}_{cd}(y)} \Wick{{T}^{(0)}_{ab}(x)}\right)h^{cd}(y) d\mu_\eta(y).
\end{equation}
The expectation value of the stress-energy tensor is covariantly conserved, in other words it is divergence-free. Therefore, according to Proposition \ref{prop:decomposition}, the expectation value of $\Wick{T_{ab}}$ at linear order can be decomposed in a scalar and a transverse traceless part as a past-compact, divergence-free symmetric 2-tensor. Moreover, both the contributions  $\omega_0(\Wick{T_{ab}^{(1)}})$  and ${\tilde\omega}(\Wick{T_{ab}^{(0)}})$ are covariantly conserved, because
$\omega_0(\Wick{T_{ab}^{(1)}})$ vanishes thanks to the choices made in Proposition \ref{prop:background}, while ${\tilde\omega}(\Wick{T_{ab}^{(0)}})$ is the difference of the expectation values of the stress-energy tensor on Minkowski in two states. Even when ${\tilde\omega}(\Wick{T_{ab}^{(0)}})$ is taken into account in the past-compact source, the conservation of the corresponding term is assured by the analysis presented in the appendix \ref{se:source-omega1}. Hence, we conclude that the contribution $N_{ab}$ must be covariantly conserved as well.\\

In the analysis of the nonlocal contribution arising at first order, there is also a non trivial renormalisation procedure which needs to be taken into account. However, according to \cite{HW01,HW05}, the expected renormalisation freedom must be chosen also at the linearised level, in a way that the resulting tensor is covariantly conserved. Therefore, still according to \cite{HW01,HW05}, the nonlocal contribution must have the form
\begin{equation}\label{eq:ren-freedom-nonlocal}
\begin{aligned}
N_{ab}  = \Pi_{ab}(\bar{h})  
+{\tilde{\alpha}}^{\mathrm{S}}_1 m^4 \bar{h}_{ab}^{\mathrm{S}}  
 +{\tilde{\alpha}}^{\mathrm{TT}}_1 m^4 \bar{h}_{ab}^{\mathrm{TT}} 
 -\frac{1}{2}{\tilde{\alpha}}^{\mathrm{S}}_2 m^2 \square \bar{h}_{ab}^{\mathrm{S}}
 -\frac{1}{2}{\tilde{\alpha}}^{\mathrm{TT}}_2 m^2 \square \bar{h}_{ab}^{\mathrm{TT}}
+ {\tilde{\alpha}}_3^{\mathrm{S}} J^{(1)}_{ab} + {\tilde{\alpha}}^{\mathrm{TT}}_4I^{(1)}_{ab},  
\end{aligned}
\end{equation}
where the constants ${\tilde{\alpha}}^{\mathrm{S}/\mathrm{TT}}_{i}$ are redefined renormalisation constants. 
Actually, these parameters differ from those that descend from the action in
\begin{equation*}
    \mathcal{L}_F = -\frac{1}{2} \left( \alpha_1 m^4 + \alpha_2 m^2 R + \alpha_3 R^2 + \alpha_4 C_{abcd}C^{abcd} \right) \sqrt{-g},
\end{equation*}
with $C_{abcd}$ representing the Weyl tensor,
and appear in \Eq \eqref{eq:T_ab-omega}. Although, this redefinition is crucial as the renormalisation procedure used in the derivation of $\Pi_{ab}$ may differ from the point splitting renormalisation discussed in \Eq \eqref{eq:T_ab-omega}. It should be pointed out that some of the constants appear to be doubled in the above equations, because the $\mathrm{S}$ and $\mathrm{TT}$ modes decouple at the linear order. As such, when the Bogoliubov map \eqref{eq:Bogoliubov} is employed at the linear order, the renormalisation constants arise from separate renormalisation procedures, see \Eq \eqref{eq:Bogoliubov-lin}. Finally, the eventual contribution at linear order arising from the trace anomaly in \Eq \eqref{eq:T_ab-omega}, which is necessary to make $\omega(\Wick{T_{ab}})$ covariantly conserved, is implicitly contained in the expression $N_{ab}$ that we are going to derive.\\

With all these observations at disposal, we can now present the following proposition for the explicit form of the nonlocal contribution $N_{ab}$ in \Eq \eqref{eq:def-N}.

\begin{proposition}\label{prop:non-local}
Let $\bar{h}_{ab} \in \Gamma^{\delta}_{\mathrm{pc}}(T^{\ast}\mathbb{R}^{n}\otimes_{s}T^{\ast}\mathbb{R}^{n})$ be the trace reversed metric perturbation in the de Donder gauge and consider its decomposition in Proposition \ref{prop:decomp-divfree}. Then, the nonlocal contribution $N_{ab}$ to the linearised expectation value of the quantum stress-energy tensor in \Eq \eqref{eq:linearisedT1} reads
\begin{equation}\label{eq:stress-energy-tensor-expval-lin}
\begin{aligned}
N_{ab}  = \Pi_{ab}(\bar{h}) 
&+{\tilde{\alpha}}^{\mathrm{S}}_1 m^4 \bar{h}_{ab}^{\mathrm{S}} 
 -\frac{1}{2}{\tilde{\alpha}}^{\mathrm{S}}_2 m^2 \square \bar{h}_{ab}^{\mathrm{S}} + {\tilde{\alpha}}^{\mathrm{S}}_3 \square\square \bar{h}_{ab}^{\mathrm{S}}\\
 & +{\tilde{\alpha}}^{\mathrm{TT}}_1 m^4 \bar{h}_{ab}^{\mathrm{TT}} -\frac{1}{2}{\tilde{\alpha}}^{\mathrm{TT}}_2 m^2 \square \bar{h}_{ab}^{\mathrm{TT}}
+{\tilde{\alpha}}^{\mathrm{TT}}_4 \square\square \bar{h}_{ab}^{\mathrm{TT}},
\end{aligned}	
\end{equation}
where $\Pi_{ab}$ is an operator that acts on the metric perturbation and decomposes diagonally with respect to the $\mathrm{S}$ and $\mathrm{TT}$ sectors
\begin{equation}\label{eq:polarisation}
\begin{aligned}
		2 \Pi_{ab}(\bar{h}) = \Pi_{abcd} [h^{cd}](x) 
		= - \frac{1}{2} \cS \cK_0 [\bar{h}_{ab}^{\mathrm{S}}] (x) + \cT \cK_0 [\bar{h}_{ab}^{\mathrm{TT}}] (x).
\end{aligned}
\end{equation}
In the above,
\begin{equation}\label{eq:S-T-operator}
		\cS = \frac{2}{3} \at m^2 + \frac{1}{2} (1 - 6\xi) \square \ct^2,  \qquad  \cT = \frac{1}{60} \at \square - 4m^2 \ct ^2
\end{equation}
are quadratic differential operators taken in the distributional sense and
\begin{align}\label{eq:K}
		\cK_{0} (x) = 	-\mathrm{i} \at \mathsf{G}_F^2 (-x) - \mathsf{G}_+^2 (-x) \ct_{\mathrm{reg}} 	= \square  \int_{\mathbb{R}^+} \frac{\varrho(M)}{M } \mathsf{G}_{\mathrm{ret}}(x,M) dM
\end{align}
is a regularised integral kernel which maps past-compact smooth functions to past-compact smooth functions with $\mathsf{G}_{\mathrm{ret}}(x,M)$ the retarded propagator of the Klein-Gordon operator with mass $\sqrt{M}$ on the Minkowski background. 
Finally, $\varrho$ is the positive-defined spectral density
\begin{equation}\label{eq:spectral-rho}
\begin{aligned}
		\varrho(M) &= \frac{1}{(2\pi)^3}\int \frac{d^3 \vec{q}}{4(|\vec{q}|^2 + m^2)} \delta(\sqrt{M}- 2\sqrt{|\vec{q}|^2 + m^2})\\
		&= \frac{1}{16\pi^2} \sqrt{1-\frac{4m^2}{M}} \Theta(M-4m^2),
\end{aligned}
\end{equation}	
where $\delta$ is the Dirac delta function and $\Theta$ the Heaviside step function.
\end{proposition}
\begin{proof}
Up to renormalisation freedom and up to the regularisation of the time ordered product \cite{HW01,HW02}, the integral kernel of $\Pi_{ab}$ reads
\begin{equation}
\label{eq:Piabcd}
\Pi_{abcd}(x-y) = -\mathrm{i}\omega_0\left(\mathscr{T} \at \Wick{{T}^{(0)}_{cd}(y)} \Wick{{T}^{(0)}_{ab}(x)}\ct_{\mathrm{reg}} - \Wick{{T}^{(0)}_{cd}(y)} \Wick{{T}^{(0)}_{ab}(x)}\right). 
\end{equation}
The first step of the analysis consists of studying $L_{abcd}(x-y)=\omega^{(0)}\left(\Wick{{T}^{(0)}_{cd}(y)} \Wick{{T}^{(0)}_{ab}(x)}\right)$ and eventually extract its retarded part. We recall that 
\begin{align*}
    \Wick{T^{(0)}_{ab}} 
    &= \frac{1}{2} \nabla_a\nabla_b \Wick{\Phi} -\Wick{\bar{\Phi}_{ab}} - \frac{1}{2}\eta_{ab} \left(\frac{1}{2}\square\Wick{\Phi}+m^2\Wick{\Phi} \right) - \xi\left(\nabla_a\nabla_b-\eta_{ab}\square \right)\Wick{\Phi} \\
    &= \mathscr{D}_{ab} \Wick{\Phi} + {\mathscr{P}_{ab}}^{cd} \Wick{\Phi_{cd}},
\end{align*}
where $\bar{\Phi}_{ab}$ is the trace reversal of $\Phi_{ab}$ and
\begin{align*}
	\mathscr{D}_{ab} &= \at \frac{1}{2} - \xi \ct \pa_a \pa_b - \at \frac{1}{4} -\xi \ct \eta_{ab} \square_\eta - \frac{1}{2}m^2 \eta_{ab}, 
    \\
	{\mathscr{P}_{ab}}^{cd} &= -\delta^c_a \delta^d_b + \frac{1}{2} \eta_{ab} \eta^{cd}. 
\end{align*}
Thus, in order to compute $L_{abcd}(x-x')$ we have to first evaluate 
\begin{align*}
\omegaz(\Wick{\Phi(x')}  \Wick{\Phi(x)}) &=  2 \omegazd(x',x)^2, \\
\omegaz(\Wick{\Phi(x')}  \Wick{\Phi_{ab} (x)}) &=  2\omegazd(x',x) \pa_{a} \pa_{b} \omegazd(x',x), \\
\omegaz(\Wick{\Phi_{cd} (x')}  \Wick{\Phi_{ab} (x)}) &= \omegazd(x',x) \pa_{c'} \pa_{d'} \pa_{a} \pa_{b} \omegazd(x',x) + 
\pa_{c} \pa_{d} \omegazd(x',x) \pa_{a} \pa_{b} \omegazd(x',x),
\end{align*}
and then apply either $\mathscr{P}_{ab}^{\,\,\,\,\,\, cd}$ or $\mathscr{D}_{ab}$ to the obtained results. Notice that in the above the products of distributions are well defined due to the microlocal spectrum condition satisfied by the two-point function of Hadamard states.\\  

To evaluate these contributions, we use the K\"{a}llén-Lehmann spectral representations of the squared flat propagators and its derivatives. 
To obtain these representations, we denote  the Poincar\'e vacuum two-point distribution by $ \mathsf{G}_+:=\omegazd$ and the integral kernel of the Feynman propagator by $ \mathsf{G}_F$. Working in Fourier domain, $\mathsf{G}_+$ admits the following representation 
\[
\hat{\mathsf{G}}_+(p) = 2\pi \delta(p^2+m^2) \theta(p_0),
\]
which is supported on the positive mass hyperboloid. Thus, we have that the convolution with itself can be represented by
\begin{align*}
\widehat{\mathsf{G}_+^2}(p) =  
   \frac{1}{(2\pi)^4} 
   \int    \delta^{(4)}{(p -q_1-q_2)} 
   \hat{\mathsf{G}}_+(q_1)\hat{\mathsf{G}}_+(q_2)
   d^4q_1 d^4q_2.
\end{align*}
Since $p$ is supported in the future light cone, we can always find a Lorentz transformation $\Lambda$ which maps $p$ to $(s,0,0,0)$ for $s=\sqrt{p_0^2-|\vec{p}|^2}$. 
Exploiting the Lorentz invariance of the involved measures, we have that 
\begin{align*}
\widehat{\mathsf{G}_+^2}(p) 
   &=  \frac{1}{4\pi^2} \iint   \frac{\delta(s-\omega(\vec{q}_1)-\omega(\vec{q}_2))}{4\omega(\vec{q}_1)\omega(\vec{q}_2)} 
   \delta^{(3)}{(\vec{q}_1+\vec{q}_2)} 
   d^3\vec{q}_1 d^3\vec{q}_2
   \\
   &=  \frac{1}{4\pi^2} \int   \frac{\delta(s-2\omega(\vec{q}))}{4\omega(\vec{q})^2 } 
   d^3\vec{q} 
   \\
   &=  \frac{1}{4\pi^2} \int_{4m^2}^\infty  \int   \frac{\delta(\sqrt{M}-2\omega(\vec{q}))}{4\omega(\vec{q})^2 } 
    \delta(M+s^2) \theta(s)  dM d^3\vec{q}
   \\
   &=  \frac{1}{(2\pi)^3} \int_{4m^2}^\infty    \frac{\delta(\sqrt{M}-2\omega(\vec{q}))}{4\omega(\vec{q})^2 }  \hat{\mathsf{G}}_+(p,M) dM d^3\vec{q},
\end{align*}
where $\omega(\vec{q}) = \sqrt{|\vec{q}|^2+m^2}$ and $M=4\omega^2$.
Denoting by $\mathsf{G}_+(p,s^2)$ the two-point function of the vacuum state of the Klein-Gordon field of mass $s$ and in terms of the spectral density $\varrho(M)$ introduced in \Eq \eqref{eq:spectral-rho}
	\begin{equation*}
		\mathsf{G}_+^2(x) = \int_{4m^2}^\infty \varrho(M) \mathsf{G}_+(x,M) dM,
	\end{equation*}
	which is the kernel of a globally defined distributions over $\cM$. \\
	
To evaluate the contributions containing derivatives of $\mathsf{G}_+$, we consider the pure Lorentz transformation $\Lambda \in \mathbf{SO}(1,3)^+$  
\begin{align*}
		\Lambda = \begin{bmatrix}
		\frac{p^0}{s} & -\frac{\vec{p}^t}{s} \\ -\frac{\vec{p}}{s}  & \mathbb{I} + \frac{\vec{p} \vec{p}^t}{\at 1+\frac{p^0}{s} \ct s^2}\end{bmatrix}\,,
\end{align*}
which maps a generic momentum $p_a$ in the forward light cone to its rest frame representation $(s,\vec{0})$, with $s = \sqrt{p_0^2 - |\vec{p}|^2}$
\begin{align*}
		-\mathsf{G}_+ \pa_a \pa_b \mathsf{G}_+ &= \cF^{-1} \ag \frac{1}{(2\pi)^4}\hat{\mathsf{G}}_+ \ast p_a p_b \hat{\mathsf{G}}_+ \cg \\
		&= \cF^{-1}  \frac{1}{4\pi^2}\ag \int \frac{d^3 \vec{q}}{(2\omega)^2} \delta(2\omega - s) \at \Lambda^{-1} q \ct_a \at \Lambda^{-1} q \ct_b  \cg\\
		&= \cF^{-1} \ag {(\Lambda^{-1})_a}^c  Q_{cd}  {(\Lambda^{-1})_{b}}^d\cg,
\end{align*}
where $\omega=\sqrt{|\vec{q}^2|+m^2}$. Notice that, in view of the rotation invariance of the measure, $Q_{ab}$ as defined above is diagonal and $Q_{jj}$ for $j\in\{1,2,3\}$ does not depend on $j$. Furthermore, the two-point functions are weak solutions of the Klein-Gordon equation of mass $m$, hence the trace satisfies ${Q_{c}}^c=Q_{cd}\eta^{cd} = -m^2 \widehat{\mathsf{G}_+^2}(p)$ and 
therefore, 
\[
Q = \text{diag}(Q_{00},\frac{1}{3} (Q_{00}-{Q_{c}}^c)).
\]
To estimate $Q_{00}$, we observe that the corresponding spectral density is 
\[
\frac{1}{(2\pi)^3} \int \frac{d^3 \vec{q}}{(2\omega)^2}  \delta(\sqrt{M}- 2\omega(\vec{q}))\omega(\vec{q})^2 = \frac{M}{4} \varrho(M)
\]
and hence,
\[
 \cF^{-1} \{Q_{00}\} =  \int_{4m^2}^\infty \frac{M}{4} \varrho(M)  \mathsf{G}_+(x,M)  dM = \frac{\square}{4} \mathsf{G}_+^2(x),
 \qquad
 \cF^{-1} \{Q_{jj}\} =  \frac{1}{3} \frac{\square-4m^2}{4} \mathsf{G}_+^2(x).
 \]
Finally, taking into account the action of $\Lambda^{-1}$ before evaluating the inverse Fourier transform, we obtain
	\begin{equation*}
		-\mathsf{G}_+ \pa_a \pa_b \mathsf{G}_+ (x) = D_{ab} \square \int_{4m^2}^\infty \frac{\varrho(M)}{M} \mathsf{G}_+(x,M) \,dM,
	\end{equation*}
	where
	\begin{equation}
	\label{eq:operator-DeltaPPDelta} 
		D_{ab} = \frac{1}{4} \pa_a \pa_b - \frac{1}{12} (\square - 4m^2) \tau_{ab}
	\end{equation}
	and
	\[
	\tau_{ab}\square = - \pa_a \pa_b + \eta_{ab} \square.
	\]	
	The evaluation of the contributions having four derivatives of $\mathsf{G}^2_+$ can be obtained in a similar way, after taking into account the minus sign in the spatial component of $\bar{q}_a = (\omega(\vec{q}),-\vec{q})$ in some of them. Thus, to study
	\begin{align*}
		\pa_a \pa_b \mathsf{G}_+ \pa_c \pa_d \mathsf{G}_+ 
		&= \cF^{-1} \ag \frac{1}{(2\pi)^2}\int \frac{d^3 \vec{q}}{(2\omega)^2} \delta(2\omega - s) \at \Lambda^{-1} q \ct_a \at \Lambda^{-1} q \ct_b \at \Lambda^{-1} \bar{q} \ct_c \at \Lambda^{-1} \bar{q} \ct_d \cg,
	\\
				 \mathsf{G}_+ \pa_a \pa_b \pa_c \pa_d \mathsf{G}_+ 
		&= \cF^{-1} \ag \frac{1}{(2\pi)^2} \int \frac{d^3 \vec{q}}{(2\omega)^2} \delta(2\omega - s) \at \Lambda^{-1} q \ct_a \at \Lambda^{-1} q \ct_b \at \Lambda^{-1} {q} \ct_c \at \Lambda^{-1} {q} \ct_d \cg,
	\end{align*}
	we first of all analyse the very same objects without the application of $\Lambda^{-1}$
	\begin{align*}
	Q_{abcd} &=  \frac{1}{(2\pi)^2}\int \frac{d^3 \vec{q}}{(2\omega)^2} \delta(2\omega - s)  q_a q_b \bar{q}_c\bar{q}_d,
	\\
	Q_{abcd}' &=  \frac{1}{(2\pi)^2}\int \frac{d^3 \vec{q}}{(2\omega)^2} \delta(2\omega - s)  q_a q_b {q}_c{q}_d\, .
	\end{align*}
	Then, both of them are non vanishing only when the corresponding indices are equal in pairs. 
Those components can be computed to be
\begin{align*}
Q_{0000}&=\int_{4m^2}^\infty  \frac{M^2}{16}\varrho(M)\hat{\mathsf{G}}_+(p,M)dM,
\\ 
Q_{iiii}&=\frac{1}{5}\int_{4m^2}^\infty\left({1-\frac{4m^2}{M}}\right)^2\frac{M^2}{16}\varrho(M)\hat{\mathsf{G}}_+(p,M) dM,
\\ 
Q_{00jj}&=-Q_{0j0j}=-Q_{j0j0}=Q_{jj00} = \frac{1}{3}\int_{4m^2}^\infty\left({1-\frac{4m^2}{M}}\right)\frac{M^2}{16} \varrho(M)\hat{\mathsf{G}}_+(p,M)dM,
\\ 
Q_{iijj}&=Q_{ijij}=Q_{jiji}=Q_{jjii} = \frac{1}{15}\int_{4m^2}^\infty\left({1-\frac{4m^2}{M}}\right)^2\frac{M^2}{16} \varrho(M)\hat{\mathsf{G}}_+(p,M)dM,
 \qquad i\neq j,
\end{align*}
where the factor $1/3$, $1/5$ and $1/15$ arise because of the appropriate integration over the two dimensional unit sphere.
The result is similarly for $Q'$ up to the sign without the minus signs.
After applying $\Lambda^{-1}$ we get
	\begin{align*}
		\pa_a \pa_b \mathsf{G}_+ \pa_c \pa_d \mathsf{G}_+ (x) &= D_{abcd} \square \int_{4m^2}^\infty \frac{\varrho(M)}{M} \mathsf{G}_+(x,M) dM, 	
		\\
		\mathsf{G}_+ \pa_a \pa_b \pa_c \pa_d \mathsf{G}_+ (x) &= D_{abcd}' \square \int_{4m^2}^\infty \frac{\varrho(M)}{M} \mathsf{G}_+(x,M) dM,
\end{align*}
	where 
	\begin{align*}
		D_{abcd} &= \frac{1}{16} \pa_a \pa_b \pa_c \pa_d \nonumber \\
		&- \frac{1}{48} \at \square - 4m^2 \ct \at \pa_a \pa_b \tau_{cd} + \pa_c \pa_d \tau_{ab} - \pa_a \pa_c \tau_{bd} - \pa_b \pa_c \tau_{ad} - \pa_a \pa_d \tau_{bc} - \pa_b \pa_d \tau_{ac} \ct \nonumber \\ 
		&+ \frac{1}{240} \at \square - 4m^2 \ct^2 \at \tau_{ab}\tau_{cd} + \tau_{ac} \tau_{bd} + \tau_{ad} \tau_{bc} \ct, 
	\end{align*}
while
	\begin{align*}
		D'_{abcd} &= \frac{1}{16} \pa_a \pa_b \pa_c \pa_d \nonumber \\
		&- \frac{1}{48} \at \square - 4m^2 \ct \at \pa_a \pa_b \tau_{cd} + \pa_c \pa_d \tau_{ab} + \pa_a \pa_c \tau_{bd} + \pa_b \pa_c \tau_{ad} + \pa_a \pa_d \tau_{bc} + \pa_b \pa_d \tau_{ac} \ct \nonumber \\ 
		&+ \frac{1}{240} \at \square - 4m^2 \ct^2 \at \tau_{ab}\tau_{cd} + \tau_{ac} \tau_{bd} + \tau_{ad} \tau_{bc} \ct. 
	\end{align*}
Combining those terms, and recalling the form of $D_{ab}$ given in \eqref{eq:operator-DeltaPPDelta} we have that 
	\begin{align*}
				D_{abcd}+D'_{abcd} =&  2 D_{ab}D_{cd}+
		\frac{1}{120} \at \square - 4m^2 \ct^2 \at \tau_{ac} \tau_{bd} + \tau_{ad} \tau_{bc} -\frac{2}{3}\tau_{ab}\tau_{cd} \ct .
	\end{align*}
Let us now collect all the terms, getting
	\begin{align*}
	L_{abcd} =&    
	\left(2(\mathscr{D}_{ab}- \overline{D}_{ab} )(\mathscr{D}_{cd}- \overline{D}_{cd} ) +  \frac{1}{120} \at \square - 4m^2 \ct^2 \at \tau_{ac} \tau_{bd} + \tau_{ad} \tau_{bc} -\frac{2}{3}\tau_{ab}\tau_{cd} \ct\right) 
	\\
	&\square \int_{\mathbb{R}} \frac{\varrho(M)}{M}  \mathsf{G}_+(x,M) dM,
	\end{align*}
	where $\overline{D}_{ab}$ is the trace reversal of  $D_{ab}$, hence 
	\[
	\mathscr{D}_{ab}- \overline{D}_{ab} = -\frac{1}{3}\left(m^2+\frac{1}{2} (1-6\xi) \square \right) \tau_{ab}\, .
	\]
Thus,
	\begin{equation*}
		\begin{aligned}
			L_{abcd}(x) &= \aq \frac{1}{3} \tau_{cd} \tau_{ab} \cq \frac{2}{3} \at m^2 + \frac{1}{2}(1-6\xi) \square \ct^2 \square\int_{\mathbb{R}} \frac{\varrho(M)}{M} \mathsf{G}_+(x,M) dM \\
			&+ \aq \frac{1}{2} \at \tau_{ac} \tau_{bd} + \tau_{ad} \tau_{bc} \ct - \frac{1}{3} \tau_{ab} \tau_{cd} \cq \frac{1}{60} \at \square - 4m^2\ct^2  \square \int_{\mathbb{R}} \frac{\varrho(M)}{M}  \mathsf{G}_+(x,M) dM.
		\end{aligned}
	\end{equation*}
Now, in order to obtain the expression for $\Pi_{abcd}(x-y)$, as in \Eq \eqref{eq:Piabcd}, we need to combine $L_{abcd}(x)$ with the contribution arising from the regularised time ordered product. The latter, is defined up to the ordinary renormalisation ambiguity in the extension of the distribution $\mathsf{G}^2_F(x-x')$, originally defined for $x\neq x'$, to test functions which are supported also on the diagonal \cite{BDF}. In particular, if the action of $\Pi_{abcd}$ on the metric perturbation decouples at linear order the $\mathrm{S}$ and $\mathrm{TT}$ modes, the renormalisation freedoms for the two components become, in principle, independent. Moreover, as a consequence of \cite{HW01,HW05}, these renormalisation freedoms appear in a form permitting their reabsorption in a redefinition of the $\alpha_i$ already present in the stress energy tensor. Therefore, using the decomposition at first order of $G_{ab}$, $I_{ab}$ and $J_{ab}$ in Eq.~\eqref{eq:Einstein-tensor-1}
and Eq.~\eqref{eq:local-curvature-tensor-A}, if we prove that the action of $\Pi_{abcd}$ decouples, then the way in which the constants ${\tilde{\alpha}}^{\mathrm{S}}_i$ and ${\tilde{\alpha}}^{\mathrm{TT}}_i$ appear in \eqref{eq:stress-energy-tensor-expval-lin} is justified.\\

One of the possible extensions to the diagonal of $\mathsf{G}_F^2$ is
	\begin{equation*}
		\mathsf{G}_F^2(x) = \square \int_{4m^2}^\infty \frac{\varrho(M)}{M} \mathsf{G}_F(x,M) dM.
	\end{equation*}
	Thus, using that $\mathsf{G}_F = \mathsf{G}_+ + i \mathsf{G}_{\mathrm{adv}}$, that $\mathsf{G}_{\mathrm{adv}}(-x) = \mathsf{G}_{\mathrm{ret}}(x)$ and \Eq \eqref{eq:K} we have that \Eq \eqref{eq:Piabcd} becomes
	\begin{equation*}
		\begin{aligned}
			\Pi_{cdab}(x-x') &= \aq \frac{1}{3} \tau_{cd} \tau_{ab} \cq \frac{2}{3} \at m^2 + \frac{1}{2}(1-6\xi) \square \ct^2 \cK_{0} (x-x') \\
			&+ \aq \frac{1}{2} \at \tau_{ac} \tau_{bd} + \tau_{ad} \tau_{bc} \ct - \frac{1}{3} \tau_{ab} \tau_{cd} \cq \frac{1}{60} \at \square - 4m^2\ct^2 \cK_{0} (x-x').
		\end{aligned}
	\end{equation*}
	Finally, recognising the form of the operators $P^{(\mathrm{S})}$ and $P^{(\mathrm{TT})}$ as given respectively in \Eqs \eqref{eq:projector-S} and \eqref{eq:projector-T}, and using their action on $h_{cd}$ as given in \eqref{eq:projectors-onhab}, we have that  
	\begin{align*}
	\int_\cM \Pi_{abcd}(x-x') h^{cd}(x') dx' &= 
	-\frac{1}{3} \at m^2 + \frac{1}{2}(1-6\xi) \square \ct^2 \cK_{0} (\bar{h}^{\mathrm{S}}_{ab}) 
			+\frac{1}{60} \at \square - 4m^2\ct^2 \cK_{0} (\bar{h}^{\mathrm{TT}}_{ab}).
	\end{align*}
\end{proof}

Finally, consistently with the result of the previous proposition, we have the following result.
\begin{proposition}
Consider the operator $\Pi_{abcd}$, introduced in \Eq \eqref{eq:polarisation}, which acts on smooth past-compact metric perturbations. The operator $\Pi_{abcd}$ is invariant under gauge transformations and vanishes when applied on $h^{\mathrm{V}}_{ab}$.  
\end{proposition}
\begin{proof}
We observe that $\partial_x^c \Pi_{abcd}=0=\partial_x^d \Pi_{abcd}$. Hence, for any past-compact smooth vector field $X_a$, it holds that 
\[
\Pi_{abcd}(\partial_y^cX^d) = \partial_x^c(\Pi_{abcd}X^d)=0.
\]
This implies that $\Pi_{abcd}(h^{cd} + \partial_y^cX^d+\partial_y^dX^c) = \Pi_{abcd}(h^{cd})$, namely $\Pi_{abcd}$ is invariant under gauge transformations. Furthermore, since $h^{\mathrm{V}}_{ab}=\partial_a v_b^{\mathrm{T}} +\partial_b v_a^{\mathrm{T}}$, $\Pi_{abcd}$ vanishes when applied on $h^{\mathrm{V}}_{ab}$.
\end{proof}

\subsection{The linearised semiclassical Einstein equations}
Combining the results of the previous section, the linearised semiclassical Einstein equations for $h_\ab$ in the de Donder gauge with a past-compact source $S_{ab}$, can now be expressed in more explicit terms. In particular, the $\mathrm{S}$ and $\mathrm{TT}$ contributions decouple both in the expectation value of the linearised stress-energy tensor, see \Eq \eqref{eq:linearisedT1}, and in the linearisation of $G_{ab}$, see \Eq \eqref{eq:Einstein-tensor-1-decomposition}. Furthermore, as discussed in the Appendix \ref{se:source-omega1}, the state perturbations are fully reabsorbed in a redefinition of the source term $S_{ab}$, leaving an equation for the unknown $h_{ab}$ only. Finally, as the source term $S_{ab}$ is also past-compact, according to Proposition \ref{prop:decomposition}, admits itself a unique decomposition in scalar, vector and transverse traceless part. Moreover, since $S_{ab}$ must be covariantly conserved, according to Proposition \ref{prop:decomp-divfree}, we have
\[
    S_{ab} =  S_{ab}^{\mathrm{S}}+S_{ab}^{\mathrm{TT}}.
\]
Therefore, we conclude that the scalar and transverse traceless degrees of freedom in the metric perturbation decouple completely in the linearised equation \eqref{eq:FPmetric}. To summarise and express this in an explicit form, we have proven that the linearised semiclassical Einstein equations are
\begin{equation}\label{eq:semiclassical-sources}
\begin{aligned}
	-\frac{1}{2 \kappa} \square \bar{h}_{ab}^{\mathrm{S}}
	&= - \frac{1}{4} \cS \cK_0 [\bar{h}_{ab}^{\mathrm{S}}] +{\tilde{\alpha}}^{\mathrm{S}}_1 m^4 \bar{h}_{ab}^{\mathrm{S}}  -\frac{1}{2}
{\tilde{\alpha}}^{\mathrm{S}}_2
 m^2 \square \bar{h}_{ab}^{\mathrm{S}}+
 {\tilde{\alpha}}^{\mathrm{S}}_3 \square\square \bar{h}_{ab}^{\mathrm{S}}
 + S_{ab}^{\text{S}}
	\\
	-\frac{1}{2 \kappa} \square \bar{h}_{ab}^{\mathrm{TT}}
		&=\frac{1}{2} 
		\cT \cK_0 [\bar{h}_{ab}^{\mathrm{TT}}]+{\tilde{\alpha}}^{\mathrm{TT}}_1 m^4 \bar{h}_{ab}^{\mathrm{TT}} -\frac{1}{2}
{\tilde{\alpha}}^{\mathrm{TT}}_2 m^2 \square \bar{h}_{ab}^{\mathrm{TT}}+{\tilde{\alpha}}^{\mathrm{TT}}_4 \square\square \bar{h}_{ab}^{\mathrm{TT}}+S_{ab}^{\mathrm{TT}}.
\end{aligned}
\end{equation}
We remind the reader that the new renormalisation constants ${\tilde{\alpha}}_i$ should take into account the new renormalisation freedom which arises in the analysis of the nonlocal contribution discussed in Proposition \ref{prop:non-local}. Also, due to the decoupling of the metric perturbation into scalar and transverse traceless components, the corresponding renormalisation constants ${\tilde{\alpha}}_i^{\mathrm{S}}$ and ${\tilde{\alpha}}_i^{\mathrm{TT}}$ are independent.

\begin{remark}
Our formulas are in analogy with those in \cite{Anderson}, where a similar analysis of the linear response of the semiclassical Einstein equations around Minkowski spacetime was studied. The difference though, compared to that derivation, is the choice and form of the renormalisation constants ${\tilde{\alpha}}_i^{\mathrm{S}}$, ${\tilde{\alpha}}_i^{\mathrm{TT}}$ appearing in \Eqs \eqref{eq:semiclassical-sources}. In particular, as it is a crucial point of distinction and a relevant observation for later discussions, we shall focus on it in the next section.
\end{remark}

\subsection{Renormalisation constants and their background values}
\label{se:renormalisation-constants}

In the perturbative treatment of interacting quantum field theory, it is known that counterterms and renormalisation constants admit a power expansion in the coupling constants.
By the principle of general local covariance \cite{BFV,HW01}, if these constants can be fixed with some principle in a globally hyperbolic spacetime, then they must be coherently fixed in all other globally hyperbolic spacetimes. Therefore, we can use general local covariance, together with a deformation argument and the requirement that Minkowski spacetime is a solution of the semiclassical Einstein equation, to fix, to all order, some of these constants.
Indeed, we did fix $\alpha_1$ in \Eq \eqref{eq:T_ab-omega}, as explained in Proposition \ref{prop:background} in order to impose the validity of the semiclassical equation on the background. Therefore, this choice fixes by general local covariance also the values of ${\tilde{\alpha}}_1^{\mathrm{S}/\mathrm{TT}}$. 
We have the following Proposition.

\begin{theorem}\label{thm:fixing-alpha1}
Consider the semiclassical Einstein equations given in \eqref{eq:semiclassical-sources}. In the hypothesis of Proposition \ref{prop:background} and by the principle of general local covariance, the fixed value of $\alpha_1$ gives that the values taken by ${\tilde{\alpha}}_1^{\mathrm{S}}$ and ${\tilde{\alpha}}^{\mathrm{TT}}_1$ are
\begin{equation*}
    {\tilde{\alpha}}^{\mathrm{S}}_1 = \frac{1}{64 \pi^2}\, , \qquad {\tilde{\alpha}}^{\mathrm{TT}}_1 = 0.
\end{equation*}
\end{theorem}
\begin{proof}
In Proposition \ref{prop:non-local} we denoted by ${\tilde{\alpha}}_1^{\mathrm{S}}$ and ${\tilde{\alpha}}_1^{\mathrm{TT}}$ the renormalisations arising from the Bogoliubov map. These freedoms have to be related to $\alpha_1$ which was fixed in Proposition \ref{prop:background} by the requirement that the  Poincar\'e vacuum is a solution of the semiclassical Einstein equation. In order to fix these spurious renormalisation constants, we make use of the principle of general local covariance by means of which the renormalisation constant $\alpha_1$ fixed on Minkowski spacetime fixes the constants in all other globally hyperbolic spacetime. Therefore, in order to fix this extra renormalisation freedom, we just need to compute $T_{ab}$ at linear order in perturbation theory, use the classical M\o ller map for some particular perturbation for the already chosen $\alpha_1$ and equate the obtained expression with the corresponding one derived in Proposition \ref{prop:non-local} given in terms of ${\tilde{\alpha}}^{\mathrm{S}}_1$ and ${\tilde{\alpha}}^{\mathrm{TT}}_1$.  

\bigskip
\noindent
{\bf The S-degrees of freedom and  ${\tilde{\alpha}}^{\mathrm{S}}_1$:}
	
\noindent	
Consider the perturbation of $\eta_{ab}$ given by a conformal rescaling 
 	\[
		g_{ab} = \Omega^2 \eta_{ab}, \qquad \Omega^2(t,\mathbf{x}) = (1 + W f_n(\mathbf{x}) \chi_\nu(t))^2\,,
	\]
where $f_n$ here is a family of compactly supported smooth functions that are constant in time and tend to $1$ for large $n$. 
Furthermore, $W$ is a positive constant that, as we will see, corresponds to a variation of the mass of the scalar field. The function $\chi_\nu$ is given by $\chi_\nu(t) = \int_{-\infty}^t \lambda\left(\frac{t'}{\nu}\right)\frac{1}{\nu} dt'$, where $\lambda$ is a positive compactly supported smooth function in the region $(-1,-\epsilon)$ with small positive $\epsilon$, and such that $\chi_\nu(t)=1$ for $t\geq -\epsilon \nu$ and $\chi_\nu(t) = 0 $ for $t<-\nu$. Hence, $\chi_\nu$ is a function of time which realises a smooth switch on of the interaction. In the limit $\nu \to \infty$, the switch on $\chi_\nu$ converges pointwise to $1$ and for larger $\nu$ the region where $\chi_\nu\neq1$ translates back in the past. For later purposes we also assume that, 
\begin{equation}\label{eq:constraint-model}
	|W \chi_\nu''(6\xi-1)|< m^2 \,,
\end{equation}
which holds for any choice of $W$ if $\nu$ is sufficiently large.\\

The linearised action of this perturbative model reads
	\begin{align*}
		V_{1,tr} &= -  \int_{\cM} {\phi}^2 (x') \at m^2 + \frac{1}{2} \at 1- 6\xi \ct \square \ct W f_n(\mathbf{x}') \chi_\nu (t(x')) dx'.
	\end{align*}
and in the limit of large $n$ and large time, the action of this interaction Lagrangian corresponds to a change in the mass of the Klein Gordon field. Therefore, we denote by 
\[
\delta m^2 = 2 m^2 W
\] 
the change in the mass occurring in the limit $\nu\to\infty$ and $n\to\infty$ by means of this interaction Lagrangian. The two-point function of the system can be determined through the classical M\o ller map applied to the vacuum of the scalar field theory of mass $m$ on the Minkowski spacetime
	\begin{align*}
		\bar{\omega}_2(x',x) &= \mathcal{R}_{\mathcal{L}_I}\omegazd \mathcal{R}_{\mathcal{L}_I}^{\ast}(x',x).
	\end{align*}
In particular, as shown in the appendix of \cite{drago}, in the limits $n \to \infty$ and $\nu \to \infty$ the system converges to the  Poincar\'e vacuum ${\omega}^\delta_2$ of the Klein-Gordon field of mass $m^{2}+\delta m^2$.

\bigskip
\noindent 
{\it State in the large time limit $\nu\to \infty$:}\\
We consider the limit $n\to\infty$ so that $f_n$ tends to $1$. The state for $t<-\nu$ is the  Poincar\'e vacuum, while for later times we consider a generic mode expansion
\[
\bar{\omega}_2(x,x') = \frac{1}{(2\pi)^3}\int \frac{\overline{\zeta_k(t)}}{(1+\chi_\nu(t) W)} \frac{\zeta_{k}(t')}{{(1+\chi_\nu(t') W)}}   e^{i \mathbf{k} (\mathbf{x}-\mathbf{x}') } d^3 k.
\]
Here, the modes $\zeta_k$ are taken to satisfy the equation 
\begin{equation}\label{eq:modes-zeta-q}
\zeta_k(t)'' + q_k^2(t) \zeta_k = 0,
\end{equation}
subject to the constraint
\[
\zeta_k(t) = \frac{e^{\mathrm{i} \Omega_k t }}{\sqrt{2\Omega_k}}, \qquad \Omega_k^2 = \mathbf{k}^2 +m^{2} \qquad t<-\nu,
\]
and for
\begin{align*}
    q_k(t)^2 &= \left(k_1^2 +k_2^{2}+k_3^2+v(t))\right)\\
    v(t)&=(1+\chi_\nu W)^2m^2+\left(6\xi-1\right)\frac{(1+\chi_\nu W)''}{ (1+\chi_\nu W)}.
\end{align*}
Notice that the inequality \eqref{eq:constraint-model} implies that $v(t)>0$. Furthermore, the modes for $t<-\nu$ fix the state to be the ordinary  Poincar\'e vacuum in the past. To estimate the form of these modes for large positive times, we use a method developed in the appendix of \cite{DHP} and we start with a WKB (Wentzel–Kramers–Brillouin) approximation 
\[
\zeta_k^0 = \frac{e^{\mathrm{i} \int_{t_0}^{t} q_k(t')dt'}}{\sqrt{2 q_k(t)}}.
\]
We observe that these modes satisfy the equation
\begin{equation}\label{eq:equation-modes-WKB}
{\zeta_k^0}(t) '' +q_k^2(t){\zeta_k^0(t)}  + w_k(t){\zeta_k^0(t)}  = 0,
\end{equation}
where
\begin{equation*}
    w_k(t) = \frac{1}{4}  \frac{(q_k^2(t))''}{q_k^2(t)}-\frac{5}{16}
\left(\frac{(q_k^2(t))'}{q_k^2(t)}\right)^2
\end{equation*}
and the function $w_k(t)$ is past-compact in time. Furthermore, $q_k^2(t) = \Omega_k^2(t)$ for $t< -\nu$ and thus we have that 
\[
\zeta_k^0(t) = \zeta_k(t) , \qquad t<-\nu.
\]  
The modes $\zeta_k(t)$ are then seen as perturbations of $\zeta^0_k(t)$ due to an external potential $-w_k(t)$. Hence, since the retarded operator of \Eq \eqref{eq:equation-modes-WKB} has integral kernel
\begin{equation*}
{g_{\text{ret}}}_k(t,t') 
\coloneqq -{g_{\text{PJ}}}_k(t,t')\theta(t-t' )
\end{equation*}
with
\begin{align*}
    {g_{\text{PJ}}}_k(t,t') &\coloneqq -\mathrm{i} \left(\overline{\zeta^0_k(t')}\zeta^0_k(t) - \overline{\zeta^0_k(t)}\zeta^0_k(t')\right)\\
&= \frac{e^{\mathrm{i} \int_{t'}^t q_k(s) ds}-e^{-\mathrm{i} \int_{t'}^t q_k(s) ds}}{2\sqrt{q_k(t)q_k(t')} \mathrm{i}},
\end{align*}
it holds that
\begin{equation*}
\zeta_k(t) = \zeta^{0}_k(t)+ \int_{-\infty}^t  \frac{e^{\mathrm{i} \int_{r}^t q_k(u) du}-e^{-\mathrm{i} \int_{r}^t q_k(u) du}}{2\sqrt{q_k(t)q_k(r)} \mathrm{i}}  w_k(r)\zeta_k(r) dr.
\end{equation*}
Therefore, using recursively this relation, we get the perturbative series for $\zeta_k$
\[
\zeta_k = \sum_{n\geq 0} (({g_{\text{ret}}}_k) w_k)^n \zeta_k^{0}\,,
\]
which can be proven to be absolutely convergent. Actually, observing that $q_k(t)\geq \epsilon>0$ for a sufficiently small positive $\epsilon$, we have the uniform estimate $|{g_{PJ}}_k(t,t')| \leq \frac{1}{\epsilon}$ and
\[
|\zeta_k(t)| \leq \frac{1}{\sqrt{2\epsilon}}e^{\int_{-\infty}^t \frac{|w_k(s)|}{\epsilon} ds }.
\] 
Notice that $\chi_\nu(t)=1$ for $t>0$ and $\nu\to\infty$ and therefore, it holds that
\begin{equation}\label{eq:chi-bounds}
\|\chi_\nu^{(l+2)}\|_1 \leq  \nu^{-(l+1)}\|\lambda^{(l+1)}\|_1  \underset{\nu\to\infty}{\longrightarrow} 0,
\qquad
\|\chi_\nu^{(l+1)}\|_\infty 
\leq  \nu^{-{l+1}}\|\lambda^{(l)}\|_\infty  \underset{\nu\to\infty}{\longrightarrow} 0 
\qquad l\in \mathbb{N},
\end{equation}
while $\|\chi_\nu^{(1)}\|_1 =C$ uniformly in $\nu$ for some positive constant $C$. 
Hence, recalling \Eq \eqref{eq:modes-zeta-q}, we have that 
\[
|w_k(t)| \leq  \frac{1}{4}  \frac{(v(t))''}{\epsilon^2}-\frac{5}{16}
\left(\frac{v(t)'}{\epsilon^2}\right)^2
\]
and
\[
\|v''\|_1+ \|(v')^2\|_1 \leq C(\| \chi_\nu'\|_1\| \chi_\nu'\|_\infty+\| \chi_\nu''\|_1+\| \chi_\nu^{(3)}\|_1+\| \chi_\nu^{(4)}\|_1)\,,
\]
where $C$ is a suitable positive constant uniform in $\nu$ and $k$. Thus, for large $\nu$ we can bound the integral
\[
\int_{-\infty}^t |w_k(s)| ds \leq 
\int_{-\infty}^\infty 
\left(\frac{1}{4}  \frac{(q_k^2(t))''}{\epsilon^2 }-\frac{5}{16}
\left(\frac{(q_k^2(t))'}{\epsilon^2 }\right)^2\right) ds\leq \frac{C}{\nu}\,,
\]  
where $C$ is uniform in $\nu$, $t$ and $k$.
Similarly, we have
\[
|\zeta_k(t)-\zeta^0_k(t)| \leq \frac{1}{\sqrt{2\epsilon}}\left|e^{\int_{-\infty}^t \frac{|w_k(s)|}{2\epsilon} ds }-1\right| \leq \frac{C}{\nu}.
\]
and thus
\[
\lim_{\nu\to\infty}|\zeta_k-\zeta_k^0| = 0.
\]

With this result at disposal, we can now study   $\bar\omega_2(f,g)$ with $f,g$ compactly supported smooth functions with support in the future of the Cauchy surface at $t=0$ and compare it with the  Poincar\'e vacuum. First, let us express the two-point functions in terms of the modes:
\begin{align*}
\bar{\omega}_2(f,g) &= 
\frac{1}{(2\pi)^3}
\iiint
\frac{\overline{\zeta_k(t)}}{(1+\chi_\nu(t) W)} \frac{\zeta_{k}(t')}{{(1+\chi_\nu(t') W)}}
\overline{\hat{f}}(t,\mathbf{k})\hat{g}(t', \mathbf{k}) d^3\mathbf{k}\, dt\,  dt'
\\
\omega^\delta_2(f,g)&= 
\frac{1}{(2\pi)^3}
\iiint 
\frac{\overline{\zeta^0_k(t)}}{(1+\chi_\nu(t) W)} \frac{\zeta^0_{k}(t')}{{(1+\chi_\nu(t') W)}}
\overline{\hat{f}}(t,\mathbf{k})\hat{g}(t', \mathbf{k}) d^3\mathbf{k}\, dt\,  dt'\,,
\end{align*}
where $\omega^\delta_2$ denotes the two-point function of the  Poincar\'e vacuum of the Klein-Gordon field of mass squared  $m^2+\delta m^2$. Then, by the above, in the limit $\nu\to\infty$ the state $\bar{\omega}_2$, which is the  Poincar\'e vacuum of mass squared $m^2$ in the past, tends to $\omega^\delta_2$ 
\[
\lim_{\nu\to\infty}\left(\bar{\omega}_2(f,g) -\omega_2^\delta(f,g)\right) = 0.
\]

\bigskip
\noindent
{\it Evaluation of the stress-energy tensor at late time:}

\noindent
For $t>0$, the expectation value of the stress-energy tensor on $\bar{\omega}$ can thus be computed using the explicit form of $\Wick{T_{ab}^{(0)}}$ given in \eqref{eq:T0} with square mass equal to $m^{2}+\delta m^{2}$ and $\alpha_1$ as fixed in Proposition~\ref{prop:background}. For $t > 0$, it holds that
\begin{equation}\label{eq:Tab-vacuum}
\lim_{\nu\to\infty}\lim_{n\to\infty}\expval{\Wick{T_{ab}}}_{\tilde{\omega}}
	=
-\frac{1}{64 \pi^2} \left(\log \left(\frac{m^2+\delta m^2}{2\mu^2}\right)- \log \left(\frac{m^2}{2\mu^2}\right)\right)(m^2+\delta m^2)^2\eta_{ab}.
\end{equation}
Then, the value of ${\tilde{\alpha}}^{\mathrm{S}}_1$ descends by this expression as a consequence of the principle of perturbative agreement. Namely, we fix it computing the first order contribution in $W$ arising from this expression and compare it with that in \Eq \eqref{eq:linearisedT1} using in particular the results of Proposition \ref{prop:non-local}.

\bigskip
\noindent
{\it Evaluation of the stress-energy tensor with perturbative interacting quantum field theory:}	
	
\noindent	
At linear order in $W$, the metric perturbation we are considering is
\[
	\overline{h}^{\mathrm{S}}_{ab} = -\frac{4}{3}\tau_{ab} W f_n 2\chi_\nu
\]
In particular, in the large time limit, the non-local contribution described by $\mathcal{K}_0$ and evaluated for this perturbation vanishes. Namely, taking the trace, it holds that 
	\[
	 \lim_{\nu\to\infty}\lim_{n\to\infty}\mathcal{K}_0(4 W f_n2\chi_\nu) \to 0. 
	\]
To prove this claim, we notice that 
\begin{align*}
\mathcal{K}_0(8 W f_n\chi_\nu)  
&= \square \int_{4m^2}^\infty dM \frac{1}{M}\rho(M) \mathsf{G}_{\mathrm{ret}}(8 W f_n\chi_\nu,M)\\
&= \square \square \int_{4m^2}^\infty dM \frac{1}{M^2} \rho(M)  \mathsf{G}_{\mathrm{ret}}(8 W f_n\chi_\nu,M)- 
\square 8 W f_n\chi_\nu  \int_{4m^2}^\infty dM \frac{1}{M^2} \rho(M) \,,
\end{align*}
where we used that $\mathsf{G}_{\mathrm{ret}}$ is a fundamental solution of the Klein-Gordon equation of mass square equal to $M$. We now observe that in the limit $n\to\infty$ and for large values of $t$ (where $\chi_\nu=1$), it further holds
\begin{align*}
\mathcal{K}_0(8 W f_n\chi_\nu)  
&= 8 W \square \square \int_{4m^2}^\infty dM \frac{1}{M^2} \rho(M)  \int_{-\infty}^t 
\frac{\sin(\omega(\mathbf{p},M)(t-t'))}{\omega(\mathbf{p},M)}e^{\mathrm{i}\mathbf{p} \mathbf{x}}
\hat{f_n}(\mathbf{p})\chi_\nu(t')dt'd^3 \mathbf{p}\\
&- \square 8 W f_n\chi_\nu  \int_{4m^2}^\infty dM \frac{1}{M^2} \rho(M).   
\end{align*}
Now, notice that the function $\rho(M)/M^2$ is integrable on $[4m^2,\infty)$ and thus, in the limit $n \to \infty$, $\square 4 W f_n \to 0$ and the last contribution in the previous expression tends to $0$. Instead, for the first contribution, we observe that $\hat{f}_n$ tends to the Dirac delta function as $n \to \infty$ and thus
\begin{align*}
\lim_{n\to\infty}\mathcal{K}_0(8 W f_n\chi_\nu)  
&= 
8 W \frac{d^4}{dt^4} \int_{4m^2}^\infty dM \frac{1}{M^2} \rho(M)  \int_{-\infty}^t 
\frac{\sin(\sqrt{M}(t-t'))}{\sqrt{M}}\chi_\nu(t')dt' . 
\end{align*}
Taking the derivatives in time and by partial integration in $dt'$ four times we get 
\begin{align*}
\lim_{n\to\infty}\mathcal{K}_0(8 W f_n\chi_\nu)  
&= 
8 W \int_{4m^2}^\infty dM \frac{1}{M^2} \rho(M)  \int_{-\infty}^t
\frac{\sin(\sqrt{M}(t-t'))}{\sqrt{M}} \chi_\nu^{(4)}(t')dt'  
\end{align*}
and the fourth order derivative of $\chi$ is a smooth and of compactly supported function. Hence if $t$ is in the future of the support of $\chi^{(4)}$, we can extend the domain of the $t'$ integration to the whole $\mathbb{R}$ 
\begin{align*}
\lim_{n\to\infty}\mathcal{K}_0(8 W f_n\chi_\nu)  
&= 
8 W \int_{4m^2}^\infty dM \frac{1}{M^2} \rho(M)  \int_{-\infty}^\infty
\frac{\sin(\sqrt{M}(t-t'))}{\sqrt{M}} \chi_\nu^{(4)}(t')dt'  .
\end{align*} 
To conclude, we notice that $\sqrt{M}^{-{5}} \rho(M) \theta(M-4m^2)$ is integrable and the function $\chi^{(4)}_{\nu}$ is bounded with $L^1$-norm
\[
\|\chi^{(4)}_{\nu}\|_1 \leq \frac{C}{\nu^{3}}.
\]
Hence by dominated convergence, we can take the limit $\nu\to\infty$ before the integration in $M$. So, in this limit 
\[
\lim_{\nu\to\infty}\lim_{n\to\infty} \mathcal{K}_0(8W f_n \chi_\nu) =0
\]
With the same reasoning and still for $t>0$ we can prove
\[
\lim_{\nu\to\infty} \lim_{n\to\infty} \mathcal{S}\mathcal{K}_0(8W f_n\chi_\nu) =0.
\]

Hence, recalling the form of $\langle \Wick{T_{ab}(x)}\rangle_\omega^{(1)}$ given in \Eq \eqref{eq:linearisedT1} and the discussion around \Eq \eqref{eq:ren-freedom-nonlocal}, 
we have that
\begin{equation}\label{eq:pert-agr-T}
	\lim_{\nu\to\infty}\lim_{n\to\infty} \langle
\Wick{{{T_{a}}^a}(x)}
\rangle_\omega^{(1)} = - 8{\tilde{\alpha}}^{\mathrm{S}}_1 W m^4.
\end{equation}
Finally, for consistency, notice that this result holds despite the definition of $h_{ab}$ is not divergence free. This because the possibly additional terms appearing in \eqref{eq:linearisedT1} are uniformly convergent to zero in the limit $\nu,n \to \infty$.

\bigskip 
\noindent
{\it Comparison of the results:}
	
\noindent
Hence, comparing the contribution proportional to $\delta m^2$ in $\lim_{\nu\to\infty}\lim_{n\to\infty}\expval{\Wick{{T_{a}}^a}}_{\tilde{\omega}}$ given in \Eq \eqref{eq:Tab-vacuum} with the one obtained in \Eq \eqref{eq:pert-agr-T} we have 
    \[
	-\frac{4}{64 \pi^2}  m^2\delta{m^2} = -8{\tilde{\alpha}}^{\mathrm{S}}_1 m^4 W,
	\]
and recalling that at linear order in $W$ it holds $\delta m^2 = 2 m^2 W + O(W^2)$, we get 
    \[
	{\tilde{\alpha}}^{\mathrm{S}}_1 = \frac{1}{64 \pi^2}.
	\]
	
\bigskip
\noindent
{\bf The TT-degrees of freedom and  ${\tilde{\alpha}}^{\mathrm{TT}}_1$:}

\noindent
In order to fix the value of ${\tilde{\alpha}}_{1}^{\mathrm{TT}}$ in \eqref{eq:semiclassical-sources} we consider a transverse traceless perturbation ${h}_{ab}^{\mathrm{TT}}$ of the Minkowski metric which is particularly convenient. In this way, we compute the stress-energy tensor explicitly in this spacetime and compare it with the result in \Eq \eqref{eq:semiclassical-sources} and in Proposition \ref{prop:decomposition}. Again, by the principle of perturbative agreement, using the retarded classical M\o ller map, this fixes the freedom ${\tilde{\alpha}}_1^{\mathrm{TT}}$ in \eqref{eq:semiclassical-sources}.\\
 
We consider here the metric perturbation 
\begin{equation}\label{eq:perturbation-z}
		h_{ab} = r_{ab} W f_{\mathsf{k}}(t - x^3) \chi_\nu(t)=  \bar{h}^{\mathrm{TT}}_{ab}. 
\end{equation}
Again, as for the scalar case, $\chi_\nu(t) = \int_{-\infty}^t \lambda(\frac{s}{\nu}) \frac{1}{\nu} ds$ with $\lambda  \in C_{\mathrm{c}}^{\infty}([-1,-\epsilon])$ as before. Furthermore, in the above $f_\mathsf{k}(t-x^3) = \cos(\mathsf{k} (t - x^3))$ is a plane wave propagating forward along the $x^3$-direction, with frequency $\mathsf{k}$ and
 \[
 r_{ab} dx^a dx^b = (dx^1)^2 - (dx^2)^2.
 \] 
Finally, $W$ is positive and smaller than $1$. Then, the perturbed spacetime $(M,\eta + h)$ coincides with the Minkowski background for $t<-\nu$ and is a spacetime which contains a plane gravitational wave propagating along $x^3$ for $t>0$. The equations of motion of a scalar Klein Gordon field on the Minkowski background are then perturbed according to the following interaction Lagrangian
\begin{equation*}
		\mathcal{L}_I = -\frac{1}{2} \pa^{c'} \phi^{(0)} (x') \bar{h}_{c'd'}^{\mathrm{TT}} (x') \pa^{d'} \phi^{(0)} (x') f(x').
\end{equation*}
Notice that by construction, in the limit $\mathsf{k} \to 0$ and $\nu \to \infty$ we have that $\square h_{ab}^{\mathrm{TT}} = 0$. Hence, this perturbation is a simple solution of the TT part of \Eq \eqref{eq:semiclassical-sources} with vanishing ${\tilde{\alpha}}_1^{\mathrm{TT}}$ at least for large times. Indeed, using a similar method as the one discussed above for the S-modes, we can prove that 
\[
\lim_{\nu\to\infty} \mathcal{K}_0(\bar{h}^{\mathrm{TT}}_{ab}) =0.
\] 
Hence, in that limit and for now generic ${\tilde{\alpha}}_1^{\mathrm{TT}}$, it holds that
\begin{equation}\label{eq:TabHtt-caso}
\lim_{\nu\to\infty} 
\lim_{\mathsf{k}\to 0}
\langle
\Wick{{T_{ab}}(x)}
\rangle_\omega^{(1)} 
= {\tilde{\alpha}}_1^{\mathrm{TT}}\bar{h}^{\mathrm{TT}}_{ab} m^4. 
\end{equation}

\bigskip
\noindent
{\it State in the large time limit $\nu\to\infty$:}

\noindent
We analyse the state in the limit $\mathsf{k} \to 0$ in $h_{ab}^{\mathrm{TT}}$. From the given form of the perturbation, the state is the  Poincar\'e vacuum for $t<-\nu$ and has the generic form
\[
\bar{\bar{\omega}}_2(x,x') = \frac{1}{(2\pi)^3}\int \frac{\overline{\eta_k(t)}}{(1-(\chi_\nu(t) W)^2)^{\frac{1}{4}}} \frac{\eta_{k}(t')}{{
(1-(\chi_\nu(t') W)^2)^{\frac{1}{4}}}}   e^{\mathrm{i} \mathbf{k} (\mathbf{x}-\mathbf{x}') }     d^3 k\,,
\]
where, the modes $\eta_k$ satisfy the equation 
\begin{equation}\label{eq:modes-eta-q}
\eta_k(t)'' + \tilde{q}_k^2(t) \eta_k = 0\,,
\end{equation}
where
\begin{equation*}
\begin{aligned}
&\tilde{q}_k(t)^2 = \left(\frac{k_1^2}{1+W\chi_\nu(t)} +\frac{k_2^{2}}{1-W\chi_\nu(t)}+k_3^2 + m^2+\tilde{v}(t)\right)\\
&\tilde{v}(t)= \xi R + S
\end{aligned}
\end{equation*}
with scalar curvature
\[
R = -
\frac{(3+W^2\chi_\nu^2)}{2 (1-W^2 \chi_\nu^2)^2}(W\chi_\nu')^2- 2  \frac{W\chi_\nu}{1-W^2 \chi_\nu^2}(W \chi_\nu'')
\]
and
\[
S = -\frac{R}{4} + 
   \frac{(1+W^2\chi_\nu^2)}{8 (1-W^2 \chi_\nu^2)^2}(W\chi_\nu')^2.
\]
Moreover, at early time, it must hold that
\[
\eta_k(t) = \frac{e^{\mathrm{i} \Omega_k t }}{\sqrt{2\Omega_k}}, \qquad \Omega_k^2 = \mathbf{k}^2 +m^{2} \qquad t<-\nu.
\]
Notice that the inequality \eqref{eq:constraint-model} implies that $m^2 + \tilde{v}(t) > 0$ when $W$ is assumed to be sufficiently small. Furthermore, the condition at $t<-\nu$ fixes the state to be the ordinary  Poincar\'e vacuum in the past.\\

To estimate the form of the state in the future we analyse its two-point function using a WKB approximation similar to the one used for the $\mathrm{S}$ modes. The WKB-approximated state is obtained with the mode 
\[
\eta_k^0 = \frac{e^{\mathrm{i} \int_{t_0}^{t} \tilde{q}_k(t')dt'}}{\sqrt{2 \tilde{q}_k(t)}}.
\]
which satisfies the following equation: 
\begin{equation*}
{\eta_k^0}(t) '' +\tilde{q}_k^2(t){\eta_k^0(t)}  + \tilde{w}_k(t){\eta_k^0(t)}  = 0 ,
\end{equation*}
where
\begin{equation*}
    \tilde{w}_k(t) = \frac{1}{4}  \frac{(\tilde{q}_k^2(t))''}{\tilde{q}_k^2(t)}-\frac{5}{16}
\left(\frac{(\tilde{q}_k^2(t))'}{\tilde{q}_k^2(t)}\right)^2.
\end{equation*}
The function $\tilde{w}_k(t)$ is past-compact in time and $\tilde{q}_k^2(t) = \Omega_k^2(t)$ for $t < -\nu$. Thus, we have that 
\[
\eta_k^0(t) = \eta_k(t), \qquad \mathrm{for} \quad t<-\nu.
\]  
Similar to the $\mathrm{S}$ case it holds
\begin{equation*}
\eta_k(t) = \eta^{0}_k(t)+ \int_{-\infty}^t  \frac{e^{\mathrm{i} \int_{r}^t \tilde{q}_k(u) du}-e^{-\mathrm{i} \int_{r}^t \tilde{q}_k(u) du}}{2\sqrt{\tilde{q}_k(t)\tilde{q}_k(r)} \mathrm{i}}  \tilde{w}_k(r)\eta_k(r) dr,
\end{equation*}
that, when used recursively, gives the perturbative expansion of $\eta_k$ in terms of $\eta_k^0$, i.e.
\[
\eta_k = \sum_{n\geq 0} (({\tilde{g}_{\text{ret}}}_k) \tilde{w}_k)^n \eta_k^{0}.
\]
 
To prove that this perturbative series is absolutely convergent, we use estimates similar to those discussed for the $S$ modes. In particular, if $W$ is sufficiently small, we can find a strictly positive $\epsilon$ such that $\tilde{q}_k(t)\geq \epsilon>0$ and giving the uniform estimate $|{\tilde{g}_{PJ}}_k(t,t')| \leq \frac{1}{\epsilon}$. This leads to
\[
|\eta_k(t)-\eta^0_{k}| \leq \frac{1}{\sqrt{2 \epsilon}}
\left| e^{\int_{-\infty}^t \frac{|\tilde{w}_k(s)|}{\epsilon} ds } -1\right|.
\] 
To test this inequality in the limit $\nu\to \infty$, we analyse $\tilde{w}_k(t)$ given in \eqref{eq:modes-zeta-q}, and we observe that  
\begin{align*}
\tilde{w}_k(t) 
&=  \frac{1}{4}  \log(\tilde{q}_k^2(t))''-\frac{1}{16} (\log(\tilde{q}_k^2(t))')^2.
\end{align*}
Furthermore, $\tilde{q}_k^2$ as given in \Eq \eqref{eq:modes-eta-q} depends on $t$ through $\chi_\nu$, $\chi_\nu'$ and  $\chi_\nu''$.
Hence
\begin{align*}
\log(\tilde{q}_k^2(t))'
&=  
\sum_{\chi \in\{\chi_\nu,\chi_\nu',\chi_\nu''\}}
\frac{\partial}{\partial \chi} \log(\tilde{q}_k^2(t)) \chi'
\\
\log(\tilde{q}_k^2(t))''
&=  
\sum_{\chi \in\{\chi_\nu,\chi_\nu',\chi_\nu''\}}
\frac{\partial}{\partial \chi} \log(\tilde{q}_k^2(t)) \chi''
+
\sum_{\chi, \rho \in\{\chi_\nu,\chi_\nu',\chi_\nu''\}}
\frac{\partial}{\partial \chi}
\frac{\partial}{\partial \rho} \log(\tilde{q}_k^2(t)) \chi'\rho'.
\end{align*}
Recalling the form of $\tilde{q}_k^{2}$, we observe that  
all the derivatives of $\log(\tilde{q}_k^{2})$ with respect to $\chi$ are bounded functions and satisfy a bound uniform in $\nu,t,k$. The behaviour of the derivative of $\chi_\nu$ given in \eqref{eq:chi-bounds} gives the following uniform bound 
\[
\| \tilde{w}_k\|_1 \leq \frac{C}{\nu}
\]
valid for large $\nu$ and with $C$ uniform in $k$. With this bound at disposal, the series defining 
$\eta_k$ is absolutely convergent also in the limit $\nu\to\infty$. Furthermore, for $t>0$
\[
\lim_{\nu\to\infty} |\eta_k-\eta_k^0| = 0.
\] 

We have now enough information to compare $\bar{\bar{\omega}}_2(f,g)$, where now $f,g$ are compactly supported smooth functions with support in the future of the Cauchy surface at $t=0$, with the  Poincar\'e vacuum $\bar{\bar{\omega}}^\delta_2$ with $x_1$ rescaled by $\sqrt{1+W}$ and $x_2$ rescaled by $\sqrt{1-W}$. We have that
\begin{align*}
\bar{\bar{\omega}}_2(f,g) &= 
\frac{1}{(2\pi)^3}\iiint \frac{\overline{\eta_k(t)}}{(1-(\chi_\nu(t) W)^2)^{\frac{1}{4}}} \frac{\eta_{k}(t')}{{
(1-(\chi_\nu(t') W)^2)^{\frac{1}{4}}}}   e^{\mathrm{i} \mathbf{k} (\mathbf{x}-\mathbf{x}') }  
\overline{\hat{f}}(t,\mathbf{k})\hat{g}(t', \mathbf{k}) d^3\mathbf{k}
d^3\mathbf{k}\, dt\,  dt'
\\
\bar{\bar{\omega}}^\delta_2(f,g)&= 
\frac{1}{(2\pi)^3}\iiint \frac{\overline{\eta^0_k(t)}}{(1-(\chi_\nu(t) W)^2)^{\frac{1}{4}}} \frac{\eta^0_{k}(t')}{{
(1-(\chi_\nu(t') W)^2)^{\frac{1}{4}}}}   e^{\mathrm{i} \mathbf{k} (\mathbf{x}-\mathbf{x}') }  
\overline{\hat{f}}(t,\mathbf{k})\hat{g}(t', \mathbf{k}) d^3\mathbf{k}
d^3\mathbf{k}\, dt\,  dt'
\end{align*}
and thus 
\[
\lim_{\nu\to\infty}\left(\bar{\bar{\omega}}_2(f,g) -\bar{\bar{\omega}}_2^\delta(f,g)\right) = 0 .
\]
Therefore, for $t>0$, in the limits $\mathsf{k}\to 0$ and  $\nu \to \infty$ and at first order, we have that the perturbed two-point function gets coordinates $x_1$ and $x_2$ rescaled respectively by $\sqrt{1+W}$ and $\sqrt{1-W}$. In other words, the two-point function is the vacuum two-point function with Synge function for the perturbation \Eq \eqref{eq:perturbation-z} $\sigma^{\mathrm{TT}}$ at the place of $\sigma$. Namely
\[
\bar{\bar{\omega}}_2 = \frac{1}{4\pi^2} \frac{m}{\sqrt{2\sigma^{\mathrm{TT}}}}
K_1\left(m \sqrt{2\sigma^{\mathrm{TT}}}\right).
\]
where $\sigma^{\mathrm{TT}}$ for $t>0$ and $\mathsf{k}\to 0$ is
\begin{align*}
	\sigma^{\mathrm{TT}}(x',x) &= \sigma(x',x) + \frac{W}{2}\at (x_1' - x_1)^2 - (x_2' - x_2)^2 \ct \\
	&= \frac{1}{2} \at - (t' - t)^2 + (1+W) (x_1' - x_1)^2 + (1-W) (x_2' - x_2)^2 + (x_3' - x_3)^2\ct.
\end{align*}
As a consequence, where $\chi_{\nu} = 1$ and for $\mathsf{k}\to 0$, the Hadamard singularity looks like the Hadamard singularity on the background with $\sigma^{\mathrm{TT}}$ instead of $\sigma$. Moreover, by direct inspection, we get that the van Vleck-Morette determinant does not depend on $W$ and $[\mathsf{v}_0]=m^2/{16\pi^2}$ as well as $[\mathsf{v}_1]= m^{4}/{64\pi^2}$ alike to the Minkowski background.	Hence, the Hadamard parametrix in the perturbed spacetime is 
	\[
	\mathcal{H} =  \frac{1}{8\pi^2 \sigma^{TT}} + \left(\frac{m^2}{16 \pi^2} + \frac{m^4}{64 \pi^2} \sigma^{TT} + O({\sigma^{TT}}^2)  \right)\log (\sigma^{TT} \mu^2),
	\]
giving
\[
	\expval{\Wick{T_{ab}}}_{\bar{\bar\omega}}=
\left( 
-\frac{1}{64 \pi^2} \left(-\frac{3}{2}+ 2\gamma + \log \left(\frac{m^2}{2\mu^2}\right)
\right)
+
 \alpha_1
 \right) m^4(\eta_{ab}+ h_{ab} + O(W^2)).
\]
Finally, since $\alpha_1$ is chosen according to Proposition \ref{prop:background}, we have that $\expval{\Wick{T_{ab}}}_{\bar{\bar\omega}}=0$ and, recalling \Eq \eqref{eq:TabHtt-caso}, it implies that ${\tilde{\alpha}}_1^{\mathrm{TT}}=0$.
	\end{proof}

\section{Well-posedness of the Cauchy problem and asymptotic analysis}\label{se:prototype-and-solution}

Through the analysis presented in Section~\ref{sec3}, we have derived the linearised semiclassical Einstein–Klein–Gordon equations and demonstrated that they can be written in terms of the decoupled pair of nonlocal partial differential equations
\begin{align}\label{eq:semiclassical-sources-2}
   & \begin{cases}
 	&-\frac{1}{2 \kappa} \square \bar{h}_{ab}^{\mathrm{S}}
	= - \frac{1}{4} \cS \cK_0 [\bar{h}_{ab}^{\mathrm{S}}] +{\tilde{\alpha}}^{\mathrm{S}}_1 m^4 \bar{h}_{ab}^{\mathrm{S}}  -\frac{1}{2}
{\tilde{\alpha}}^{\mathrm{S}}_2
 m^2 \square \bar{h}_{ab}^{\mathrm{S}}+
 {\tilde{\alpha}}^{\mathrm{S}}_3 \square\square \bar{h}_{ab}^{\mathrm{S}}
 + S_{ab}^{\text{S}}\\
 &\bar{h}_{ab}^{\mathrm{S}}\vert_{\Sigma_{-\varepsilon}}= 0\\
 &\partial_t\bar{h}_{ab}^{\mathrm{S}}\vert_{\Sigma_{-\varepsilon}}=0\\
\end{cases}
\\
&\begin{cases}
 &  -\frac{1}{2 \kappa} \square \bar{h}_{ab}^{\mathrm{TT}}
		=\frac{1}{2} 
		\cT \cK_0 [\bar{h}_{ab}^{\mathrm{TT}}]+{\tilde{\alpha}}^{\mathrm{TT}}_1 m^4 \bar{h}_{ab}^{\mathrm{TT}} -\frac{1}{2}
{\tilde{\alpha}}^{\mathrm{TT}}_2 m^2 \square \bar{h}_{ab}^{\mathrm{TT}}+{\tilde{\alpha}}^{\mathrm{TT}}_4 \square\square \bar{h}_{ab}^{\mathrm{TT}}+S_{ab}^{\mathrm{TT}}\\
&\bar{h}_{ab}^{\mathrm{TT}}\vert_{\Sigma_{-\varepsilon}}  =0\\
&\partial_t\bar{h}_{ab}^{\mathrm{TT}}\vert_{\Sigma_{-\varepsilon}}=0
\end{cases}
\end{align}
with given past-compact sources $S_{ab}^{\text{S}}$ and $S_{ab}^{\text{TT}}$, supported in the future of $\Sigma_{-\varepsilon}$, where we decomposed the past-compact and divergence-free metric perturbation $\bar{h}_{ab}$, following the discussion of Section~\ref{se:metric-perturbation-decomposition}, into its scalar and transverse-traceless components, i.e.
\begin{align*}
\bar{h}_{ab}=\bar{h}_{ab}^{\mathrm{S}}+\bar{h}_{ab}^{\mathrm{TT}}\, .
\end{align*}
The linear operators $\mathcal{K}_0$, $\mathcal{S}$ and $\mathcal{T}$ appearing in the above equations are defined in Proposition~\ref{prop:non-local} and ${\tilde{\alpha}}^{\mathrm{S}}_i$ and ${\tilde{\alpha}}^{\mathrm{TT}}_i$ are renormalisation constants. Moreover, according to Theorem~\ref{thm:fixing-alpha1}, we recall that the \emph{principle of general local covariance} and the \emph{principle of perturbative agreement} imply that some renormalisation constants are already fixed by their background value, i.e.~it holds that
\begin{align*}
    {\tilde{\alpha}}^{\mathrm{S}}_1=\frac{1}{64 \pi^2}\,, \qquad {\tilde{\alpha}}^{\mathrm{TT}}_1=0\,.
\end{align*}

This section is devoted to the analysis of the linearised semiclassical equations as written in Eq.~\eqref{eq:semiclassical-sources-2}. We begin by observing that the equations governing $\bar{h}_{ab}^{\mathrm{S}}$ and $\bar{h}_{ab}^{\mathrm{TT}}$ share a common structural form. By examining the structure of the nonlocal operator $\mathcal{K}_{0}$, we shall first demonstrate that both equations can be recast into the following prototypical form:
\begin{equation}\label{eq:Proto}\iota_{\mathcal{M}}\tilde{\mathsf{G}}_{\mathrm{ret}}( \square (\square-a_1)(\square-a_2)\phi\otimes \check\rho) +\sum_{j=0}^{2}b_j \square^j \phi=S \,,\end{equation}
where $S$ is some past-compact source, where $a_{i}$ and $b_{i}$ are real constants, where $\tilde{\mathsf{G}}_{\mathrm{ret}}$ is the retarded fundamental solution of the Klein-Gordon operator $\tilde{\square}-4m^{2}$ on \emph{$5$-dimensional Minkowski spacetime} $\mathcal{M}_{5}:=\mathcal{M}\times\mathbb{R}$ and where $\iota_{\mathcal{M}}$ realises the restriction to $\mathcal{M}\cong\mathcal{M}\times\{0\}\subset\mathcal{M}_{5}$. The functions $\rho$ and $\check{\rho}$ have been defined in Eq.~\eqref{eq:nonlocalK0} and Eq.~\eqref{eq:rhocheck}, respectively.

As written, the analysis of the prototypical equations~\eqref{eq:Proto} is challenging, since the highest-derivative term is governed by a \emph{nonlocal} operator, thereby precluding the direct application of standard techniques for local partial differential equations. The central idea of the following analysis, therefore, is to suitably \emph{invert} the nonlocal contribution, thereby recasting the equation into a form that is more amenable for an analysis of its Cauchy problem and asymptotic behaviour to be performed.

The analysis will be performed in two steps. First of all, by introducing suitable auxiliary functions $\mathfrak{h}_{j}$ for $j=0,1,2$, we will rewrite the prototypical \Eq~\eqref{eq:Proto} in the form 
\begin{align}\label{eq:Proto2}
\iota_{\mathcal{M}}\tilde{\mathsf{G}}_{\mathrm{ret}} ((\square^3 \phi+\beta_2\square^{2} \phi + \beta_1 \square \phi +\beta_0\phi) \otimes ( \check \rho+b_2\mathfrak{h}_2 )) =S
\end{align}
for suitable constants $\beta_{i}$. Subsequently, we show that, provided that $\check \rho+b_2\mathfrak{h}_2$ satisfies suitable regularity and positivity conditions, which ultimately will translate into restrictions on the parameters $a_i$ and $b_i$, the operator $\iota_{\mathcal{M}}\tilde{\mathsf{G}}_{\mathrm{ret}}$ becomes invertible, thereby leading to a reformulation of the equation into a fully local problem. As a second step, we show that a similar procedure can actually always been applied to the highest-order terms without any restrictions on the parameters $b_{i}$, hence leading to a 6th order hyperbolic equation with \emph{subleading} nonlocal contributions. More precisely, we will show that Eq.~\eqref{eq:Proto2} can equivalently be written in the form
\begin{align*}
\square^3 \phi+\beta_2\square^{2} \phi + \beta_1 \square \phi +\beta_0\phi
+
\mathcal{G}_{\varsigma}^{-1}( 
(\tilde{b}_1-b_1) \square \phi + 
(\tilde{b}_0-b_0) \phi)
=
\mathcal{G}_{\varsigma}^{-1}(S)\,,
\end{align*}
where $\mathcal{G}_{\varsigma}(\phi):=\iota_{\mathcal{M}}\tilde{\mathsf{G}}_{\mathrm{ret}}(\phi\otimes\check{\varsigma})$ with $\check{\varsigma}:=\check{\varrho}+b_{2}\mathfrak{h}_{2}$ denotes the nonlocal operator appearing in Eq.~\eqref{eq:Proto2} and where $\tilde{b}_{i}$ are suitable real constants.

Having recast the system in this form, we establish the well-posedness of the Cauchy problem by employing perturbation theory, treating the nonlocal potential as a perturbative term. Altogether, this leads to the first principal result of this section, Theorem~\ref{thm:convergence}, the well-posedness of the Cauchy problem for the linearised semiclassical Einstein-Klein-Gordon system in Minkowski spacetime.

In the last part of this section, we then perform an asymptotic analysis by studying the large-time behaviour of solutions to the prototypical semiclassical equations~\eqref{eq:Proto}. This is the second main result of this section, Theorem~\ref{thm:stability}, in which we show that for some ranges of parameters $a_{i}$ and $b_{i}$, solutions are bounded and decay at most as $t^{-\frac{3}{2}}$, while for some others we obtain solutions with exponential growth. The implications of this analysis of the toy model~\eqref{eq:Proto} when applied to the semiclassical Einstein equations are discussed in Section~\ref{sec:5}.

\subsection{The form of the linearised nonlocal semiclassical Einstein equations}

To examine the structure of the obtained equations~\eqref{eq:semiclassical-sources-2}, we recall from Proposition \ref{prop:non-local} that the operator $\cK_0$, acting on past-compact smooth functions, i.e.~$\cK_0:C^{\infty}_{\mathrm{pc}}(\mathcal{M})\to C^{\infty}_{\mathrm{pc}}(\mathcal{M})$, is given by
\begin{equation}\label{eq:nonlocalK0}
\cK_0 [\phi] = \square \int_{4m^2}^\infty \rho(M) \mathsf{G}_{\mathrm{ret}}^{M}(\phi) dM , \qquad 
\rho(M)= \frac{1}{16\pi^2} \sqrt{1-\frac{4m^2}{M}} \frac{1}{M}\,,
\end{equation}
where $\mathsf{G}_{\mathrm{ret}}^M$ is the retarded fundamental solution of the Klein-Gordon equation of mass square $M$, namely,
\[
\mathsf{G}_{\mathrm{ret}}^M:C^{\infty}_{\mathrm{pc}}(\mathcal{M})\to C^{\infty}_{\mathrm{pc}}(\mathcal{M})\,, \qquad (\square -M)\mathsf{G}_{\mathrm{ret}}^M(\phi) = \phi\,.
\]
We shall show now that the nonlocal operator present in $\cK_0$ can be understood as a restriction to four dimensions of a retarded operator in a five-dimensional spacetime. To this end, we consider the five-dimensional Minkowski space 
\[
\mathcal{M}_5 := \mathcal{M} \times \mathbb{R}
\]   
and the retarded propagator of the Klein-Gordon equation with mass $2m$ on past-compact functions $\tilde{\mathsf{G}}_{\mathrm{ret}}:C^\infty_{\mathrm{pc}}(\mathcal{M}_5)\to C^\infty_{\mathrm{pc}}(\mathcal{M}_5)$. Hence, denoting by $\tilde{\square}$ the d'Alembert operator on $\mathcal{M}_5$, it is such that
\[
(\tilde\square - 4m^2) \tilde{\mathsf{G}}_{\mathrm{ret}} \tilde\phi = \tilde\phi, \qquad \tilde{\phi}\in  C^\infty_{\mathrm{pc}}(\mathcal{M}_5).
\]
To start with, we discuss suitable mapping properties of the five-dimensional retarded operator $\tilde{\mathsf{G}}_{\mathrm{ret}}$.

\begin{proposition}\label{prop:extension}
The propagator $\tilde{\mathsf{G}}_{\mathrm{ret}}$ can be linearly extended to the domain $\mathcal{D} := C^\infty_{\mathrm{pc}}(\mathcal{M})\otimes L^2(\mathbb{R})$. Moreover, it holds that
\begin{align*}
    \tilde{\mathsf{G}}_{\mathrm{ret}} u\in C^{\infty}_{\mathrm{pc}}(\mathcal{M})\otimes C^{0}(\mathbb{R})\,,\qquad\forall u\in\mathcal{D}\, .
\end{align*}
In particular, $(\tilde{\mathsf{G}}_{\mathrm{ret}} u)(x,z)$ with $(x,z)\in\mathcal{M}\times\mathbb{R}$ is continuous in $z$ for any $u\in\mathcal{D}$.
\end{proposition}

\begin{proof}
To start with, we recall that if $u$ is both past- and spatially compact, i.e.~$u\in C^{\infty}_{\mathrm{pc}}(\mathcal{M}_5) \cap C^{\infty}_{\mathrm{sc}}(\mathcal{M}_5)$, the  operator $\tilde{\mathsf{G}}_{\mathrm{ret}}$, i.e.~the retarded operator of the Klein-Gordon equation of mass $2m$ on $\mathcal{M}_5$, is most conveniently expressed using the ordinary spatial Fourier transform:
\begin{equation}\label{eq:retarded-spectral}
\tilde{\mathsf{G}}_{\mathrm{ret}} u(t,x,z) = -\iiint \theta(t-t') \frac{\sin (\omega(\mathbf{p},p_z)(t-t'))}{\omega(\mathbf{p},p_z)}  \tilde{u}(t',\mathbf{p},p_z)        e^{-\mathrm{i}\mathbf{p} x } e^{-\mathrm{i}p_z z} \,dt'\,d^3\mathbf{p}\, dp_z\,,
\end{equation}
where $\tilde{u}(t,\mathbf{p},p_z) = \int e^{\mathrm{i}\mathbf{p} x } e^{\mathrm{i}p_z z} u(t,x,z)\,d^3xdz$ and where $\omega(\mathbf{p},p_z) = \sqrt{\mathbf{p}^2+p_z^2+4m ^2}$. 
This expression can be extended to the case of $u\in C^\infty_{\mathrm{pc}}(\mathcal{M}_5)$ exploiting the causal properties of the retarded propagator.

Indeed, for any $(t,x,z) \in \mathcal{M}_5$, we can find a spatially compact smooth function $\chi_{(t,x,z)}$ such that $\chi_{(t,x,z)}=1$ on $J^{-}(t,x,z)$. Then, by the causal properties of $\tilde{\mathsf{G}}_{\mathrm{ret}}$, we have for any $u\in C^{\infty}_{\mathrm{pc}}(\mathcal{M}_5)$ that
\[
\tilde{\mathsf{G}}_{\mathrm{ret}}(u)(t,x,z) = \tilde{\mathsf{G}}_{\mathrm{ret}}(u\chi_{(t,x,z)})(t,x,z)\,.
\]
Moreover, $u\chi_{(t,x,z)}$ is in $C^{\infty}_{\mathrm{pc}}(\mathcal{M}_5)\cap C^{\infty}_{\mathrm{sc}}(\mathcal{M}_5)$ and we can apply the expression of $\tilde{\mathsf{G}}_{\mathrm{ret}}$ given in Eq.~\eqref{eq:retarded-spectral}.

To prove that we can extend $\tilde{\mathsf{G}}_{\mathrm{ret}}$ to $\mathcal{D}$, we consider
$ f\in C^{\infty}_{\mathrm{pc}}(\mathcal{M})\cap C^{\infty}_{\mathrm{sc}}(\mathcal{M})$ and use it to construct the function
\[
K_f^{(t,x)}(z) := -
\iiint \theta(t-t') \frac{\sin (\omega(\mathbf{p},p_z)(t-t'))}{\omega(\mathbf{p},p_z)}  \tilde{f}(t',\mathbf{p})         e^{-\mathrm{i}\mathbf{p} x } e^{-\mathrm{i}p_z z}  \,dt' \,d^3\mathbf{p} \,dp_z\,, \qquad 
\] 
where, as before, $\omega(\mathbf{p},p_z) = \sqrt{\mathbf{p}^2+p_z^2+4m^2}$ and now
\[
\tilde{f}(t,\mathbf{p}) = \int f(t,x) e^{\mathrm{i}\mathbf{p}x}\, dx\,.
\]
For any $(t,x)\in\mathcal{M}$, the function $K_f^{(t,x)}$ is square-integrable in $z$, because $\tilde{f}(t,\mathbf{p})$ is square-integrable in $\mathbf{p}$ and $\omega(\mathbf{p},p_z)^{-1}$ is square-integrable in $p_z$ with respect to the Lebesgue measure, with a uniformly bounded $L^2$-norm for every value of $\mathbf{p}$.
Hence, by convolution, it
defines a map from $L^2(\mathbb{R},dz)$ to $C^{0}(\mathbb{R})$ given by
\[
\mathcal{K}^{(t,x)}_f(g) := K^{(t,x)}_f * g\,.
\]
To prove continuity in $z$, we observe that $\mathcal{K}^{(t,x)}_f(g)(z)= \langle  K^{(t,x)}_f(z-\cdot)  , g \rangle$. Hence, 
\[
|\mathcal{K}^{(t,x)}_f(g)(z)|\leq \|{K}^{(t,x)}_f(z-\cdot)\|_2 \|g\|_2 =\|{K}^{(t,x)}_f\|_2 \|g\|_2
\]
and, for any $g$, we conclude that the map $z\mapsto g(z-\cdot)$ is a continuous map from $\mathbb{R}$ to $L^2(\mathbb{R})$. 

Consider now a compactly supported smooth function $g\in C^{\infty}_{\mathrm{c}}(\mathbb{R})$ and $f \in  C^{\infty}_{\mathrm{pc}}(\mathcal{M})$
and use them to build 
\[
u(t,x,z) = f(t,x) g(z)\,,
\]
which is a past-compact smooth function defined on $\mathcal{M}_5$. Moreover, consider a point $(t,x)\in \mathcal{M}$. Then, as explained above, we can find a spatially compact function $\chi$ defined on $\mathcal{M}$, which is $1$ on $J^{-}(t,x)$. 
This leads, by direct inspection, to 
\[
\tilde{\mathsf{G}}_{\mathrm{ret}}(u )(t,x,z)=\tilde{\mathsf{G}}_{\mathrm{ret}}(f\otimes g )(t,x,z) =      \tilde{\mathsf{G}}_{\mathrm{ret}}(\chi f\otimes g )(t,x,z) = \mathcal{K}^{(t,x)}_{\chi f}(g)(z) 
\]
and the right-hand side is continuous in $z$ and smooth in $(t,x)$. Furthermore,
\[
|\tilde{\mathsf{G}}_{\mathrm{ret}}(f\otimes g)(t,x,z)|\leq  \| K^{(t,x)}_{\chi f}\|_2 \|g\|_2.
\]
We can thus extend $\tilde{\mathsf{G}}_{\mathrm{ret}}(f\otimes g)$ to $\mathcal{D}$ keeping continuity in $z$, by density of $C_{\mathrm{c}}^{\infty}(\mathbb{R})$ in $L^2(\mathbb{R})$.

To prove that $\tilde{\mathsf{G}}_{\mathrm{ret}}(f\otimes g)$ is smooth in $(t,x)$, we analyse its wavefront set for $f\otimes g\in\mathcal{D}$. 
Since $f$ is smooth, we have by Theorem 8.2.9 in \cite{Hormander} that
\begin{equation*}
    \mathrm{WF}(f\otimes g) \subset \{(t,x,z;0,0,p_z) \in T^{*}\mathcal{M}_5\mid p_z \in \mathbb{R} \backslash \{ 0 \} \}\,.
\end{equation*}
Furthermore, the singularity of the $5$-dimensional retarded propagator is captured by
\begin{equation*}
    \mathrm{WF}(\tilde{\mathsf{G}}_{\mathrm{ret}}) = \{(\tilde{x},\tilde{y}; \tilde{k},\tilde{p})\subset T^{*}(\mathcal{M}_5 \times \mathcal{M}_5) \setminus \{ 0 \} \mid (\tilde{x},\tilde{k})\sim (\tilde{y},-\tilde{p}) \}  \cup \{(\tilde{x},\tilde{x}; \tilde{k},-\tilde{k})\subset T^{*}(\mathcal{M}_5 \times \mathcal{M}_5)\mid \tilde{k}\neq 0 \} \,,
\end{equation*}
where $(\tilde{x},\tilde{k})\sim (\tilde{y},-\tilde{p})$ holds if it exists an affinely parametrised null geodesic $\gamma$ in $\mathcal{M}_5$ which connects $\tilde{x}$ and $\tilde{y}$ and such that its tangent vector $\dot\gamma$ equals to $\eta^{-1} \tilde{k}$ in $\tilde{x}$ and $-\eta^{-1} \tilde{p}$ in $\tilde{y}$. Then, by the support properties of the retarded propagator and up to the introduction of a suitable cutoff function similar to the one used above to make $f\otimes g$ of compact support, we can apply Theorem 8.2.13 of H\"ormander \cite{Hormander} as (using the notation introduced therein) $\mathrm{WF}'(\tilde{\mathsf{G}}_{\mathrm{ret}})_Y=\emptyset$. 
The very same theorem implies that 
\begin{align*}
\mathrm{WF}(\tilde{\mathsf{G}}_{\mathrm{ret}}(f\otimes g)) 
&\subset \mathrm{WF}'(\tilde{\mathsf{G}}_{\mathrm{ret}}) \circ \mathrm{WF}(f\otimes g) 
\\
&= \{(\tilde{x},\tilde{k})\in T^{*}\mathcal{M}_5 | (\tilde{x},\tilde{y},\tilde{k},-\tilde{p})\in \mathrm{WF}(\tilde{\mathsf{G}}_{\mathrm{ret}}) \text{ for some } (\tilde{y},\tilde{p})\in \mathrm{WF}(f\otimes g)\}.
\end{align*}
In view of the form of the wavefront sets of $\tilde{\mathsf{G}}_{\mathrm{ret}}$ and of $f\otimes g$, we have that  
\[
\mathrm{WF}(\tilde{\mathsf{G}}_{\mathrm{ret}}(f\otimes g)) \subset \{(t,x,z;0,0,p_z)\in T^{*}\mathcal{M}_5 \backslash \{0\}\}\,,
\]
thus concluding the proof.
\end{proof}
The following proposition permits now to rewrite the nonlocal operation present in the semiclassical Einstein equations in a suitable form: 

\begin{proposition}\label{prop:nonlocal-ret5d}
Consider the map $\rho$ given in Eq.~\eqref{eq:nonlocalK0}. Then, for every $\phi\in C^{\infty}_{\mathrm{pc}}(\mathcal{M})$, we have that
\begin{equation}\label{eq:operatorIG}
\int_{4m^2}^\infty   \rho(M) \mathsf{G}_{\mathrm{ret}}^M(\phi)   dM  = \iota_{\mathcal{M}}
\tilde{\mathsf{G}}_{\mathrm{ret}}
(\phi\otimes \check{\rho}) \,,
\end{equation}
where $\iota_{\mathcal{M}}:C^{\infty}_{\mathrm{pc}}(\mathcal{M}_5)\to C^{\infty}_{\mathrm{pc}}(\mathcal{M})$ realises the restriction to $\mathcal{M}$ seen as a suitable subspace of $\mathcal{M}_5$, i.e.
\[
(\iota_{\mathcal{M}} \tilde{\phi})(x) = \tilde{\phi}(x,0)\, , \qquad \forall\tilde{\phi}\in C^{\infty}_{\mathrm{pc}}(\mathcal{M}_5)\,,
\]
and where $\check{\rho}$ denotes the inverse Fourier-Plancharel transform of  ${\rho}({p^2+4m^2}) |p|$, i.e.
\begin{equation}\label{eq:rhocheck}
\check{\rho}(z) = \int_\mathbb{R}   {\rho}({p_z^2+4m^2}) |p_z| e^{\mathrm{i}p_zz }\,dp_z\, .
\end{equation}
\end{proposition}

\begin{remark}\label{Remark:nonlocal-ret5d}
    In \Eq~\eqref{eq:operatorIG}, we note that ${\rho}({p^2+4m^2}) |p|$ is square-integrable, which implies that also $\check{\rho}\in L^{2}(\mathbb{R})$. Moreover, notice that $f\otimes \check\rho$ is past-compact and sufficiently regular to be in the extended domain $\mathcal{D}$ of $\tilde{\mathsf{G}}_{\mathrm{ret}}$ introduced in Proposition~\ref{prop:extension}. Last but not least, $\tilde{\mathsf{G}}_{\mathrm{ret}}(\phi\otimes \check{\rho})$ is at least continuous and past-compact and hence, the action of the restriction operator $\iota_{\mathcal{M}}$ is well-defined.
\end{remark}

\begin{proof}[Proof of Proposition~\ref{prop:nonlocal-ret5d}.]
Consider a function $f\in C^{\infty}_{\mathrm{sc}}(\mathcal{M})\cap C^{\infty}_{\mathrm{pc}}(\mathcal{M})$. Then,
writing $\omega_M(\mathbf{p})=\sqrt{\mathbf{p}^2+M}$ and $\omega(\mathbf{p})=\sqrt{\mathbf{p}^2+p_z^2+4m^2}$, as usual, and denoting by $\tilde{f}$ the spatial Fourier transform of $f$, we have that 
\begin{align*}
    \int\limits_{4m^2}^{\infty} \int {\mathsf{G}}_{\mathrm{ret}}^M(x-y) &f(y) \rho(M)\,  d y \, d M 
    =     
    -\int\limits_{4m^2}^{\infty}  \rho(M) \int\limits_{\mathbb{R}^3}\int\limits_{-\infty}^{t}   \frac{\sin (\omega_M(\mathbf{p})(t-t'))}{\omega_M(\mathbf{p})}  \tilde{f}(t',\mathbf{p})         e^{-\mathrm{i}\mathbf{p} x } \, dt' \,d^3\mathbf{p} \,dM, 
    \\
    &=     
    -\int\limits_{-\infty}^{\infty}   \int\limits_{\mathbb{R}^3}\int\limits_{-\infty}^{t} \frac{\sin (\omega(\mathbf{p},p_z)(t-t'))}{\omega(\mathbf{p},p_z)}  \tilde{f}(t',\mathbf{p}) \rho(p_z^2+4m^2)|p_z|         e^{-\mathrm{i}\mathbf{p} x }\,  dt' \,d^3\mathbf{p} \,d p_z,    
    \\
    &= 
    \tilde{\mathsf{G}}_{\mathrm{ret}}
    (f \otimes \check{\rho})\big|_{z = 0}.
\end{align*}
As explained in Remark~\ref{Remark:nonlocal-ret5d}, $\check{\rho}$ is a square-integrable function, hence $f \otimes \check{\rho} \in \mathcal{D}$, where $\mathcal{D}$ is the domain introduced in Proposition \ref{prop:extension}. Moreover, by the very same proposition, $\Tilde{\mathsf{G}}_{\mathrm{ret}}$ can be applied on $f \otimes \check{\rho}$ and gives origin to a continuous function in $z$. In particular, the restriction to $z=0$ can be taken. 
\end{proof}

As an immediate consequence of Proposition \ref{prop:nonlocal-ret5d}, we obtain the following result, which allows us to express the nonlocal operator $\mathcal{K}_0$ introduced in Eq.~\eqref{eq:nonlocalK0}, which appears in the linearised semiclassical Einstein-Klein-Gordon equations (see Proposition~\ref{prop:non-local}), in terms of the 5-dimensional retarded propagator $\tilde{\mathsf{G}}_{\mathrm{ret}}$.

\begin{corollary}\label{cor: formaK0}
Consider the nonlocal operator $\mathcal{K}_0$ introduced in Eq.~\eqref{eq:nonlocalK0}. Then, for all $\phi\in C^\infty_{\mathrm{pc}}(\mathcal{M})$, it holds that
\[
\mathcal{K}_0(\phi) = \iota_{\mathcal{M}}
\tilde{\mathsf{G}}_{\mathrm{ret}}(\square\phi\otimes \check{\rho})\,,
\]
where $\rho$ and $\check{\rho}$ are defined in Eq.~\eqref{eq:nonlocalK0} and Eq.~\eqref{eq:rhocheck}, respectively.
\end{corollary}

Now, examining the structure of the linearised semiclassical Einstein–Klein–Gordon equations as presented in \Eq~\eqref{eq:semiclassical-sources}, we observe that both the scalar and transverse-traceless mode can be cast into a nonlocal partial differential equation of the prototypical form
\[
\cK_0 [(\square-a_1)(\square-a_2)\phi] +\sum_{j=0}^{2}b_j \square^j \phi=S\,,
\]
where $S$ is a smooth past-compact source in $\mathcal{M}$ and where $a_i,b_j$ are fixed real constants. Using Corollary~\ref{cor: formaK0}, this can equivalently be written in the form 
\begin{equation}\label{eq:semiclassical-prototype}
\iota_{\mathcal{M}}\tilde{\mathsf{G}}_{\mathrm{ret}}( \square (\square-a_1)(\square-a_2)\phi\otimes \check\rho) +\sum_{j=0}^{2}b_j \square^j \phi=S
\end{equation}
with $\rho$ given in Eq.~\eqref{eq:nonlocalK0} and $\check{\rho}$ introduced in Proposition \ref{prop:nonlocal-ret5d}.

The next proposition is one of the central technical results of this chapter and will be of crucial importance for the subsequent analysis. It shows that nonlocal operators of the type considered in Corollary~\ref{cor: formaK0}, i.e.~$\mathcal{G}_\varsigma:=\iota_{\mathcal{M}}\tilde{\mathsf{G}}_{\mathrm{ret}}(\bullet\otimes \check{\varsigma})$, corresponding to different choices of functions $\varsigma$, are in fact invertible on the space of past-compact functions on $\mathcal{M}$, provided that the function $\check{\varsigma}$ appearing on the right-hand side satisfies suitable regularity assumptions. As we will see later on, this property enables us to recast \Eq~\eqref{eq:semiclassical-prototype}, which, as discussed above, are the prototypical linearised semiclassical Einstein-Klein-Gordon equations~\eqref{eq:semiclassical-sources-2}, as a hyperbolic PDE with a \emph{subleading} nonlocal contribution, thereby placing it within a well-controllable analytical framework.

\begin{proposition}\label{prop:inverse-log}
    Let $\varsigma$ be a positive and square-integrable function on $[4m^2,\infty)$ that is in addition H\"older continuous\footnote{By definition, H\"older continuity means that $|\varsigma(M)-\varsigma(M')|\leq C |M-M'|^{\alpha}$. This condition makes integrals of the form $\int_{x-\epsilon}^{x+\epsilon} \frac{\varsigma(M)}{M-x}\, dM$ finite, since $\int_{x-\epsilon}^{x+\epsilon} \frac{\varsigma(M)}{M-x}\, dM= \int_{x-\epsilon}^{x+\epsilon} \frac{\varsigma(x)}{M-x} dM+\int_{x-\epsilon}^{x+\epsilon} \frac{\varsigma(M)-\varsigma(x)}{M-x}\,dM$, where the first integral vanishes by symmetry, while the integrand of the second is integrable by H\"older continuity.} of weight $0<\alpha\leq 1$. Moreover, consider the operator\footnote{This operator acts on past-compact smooth functions $\phi$, which we assume to be supported in $t>0$ to keep the following notation simple.} $\mathcal{G}_\varsigma:C^{\infty}_{c}(\mathbb{R}^{+}\times\mathbb{R}^{3})
 \to C^{\infty}_{\mathrm{pc}}(\mathcal{M})$ defined by
    \begin{align*}
        \Phi = \mathcal{G}_\varsigma(\phi) := 
\iota_{\mathcal{M}}\tilde{\mathsf{G}}_{\mathrm{ret}}( \phi \otimes \check\varsigma)\,,
    \end{align*}
    where $\check\varsigma$ denotes the Fourier-Plancharel transform of $\varsigma(p^{2}+4m^{2})\vert p\vert$, as in Eq.~\eqref{eq:rhocheck}, where $\tilde{\mathsf{G}}_{\mathrm{ret}}$ is the retarded fundamental solution of mass $2m$ in $\mathcal{M}_5$ and where $\iota_{\mathcal{M}}$ realises the restriction to $\mathcal{M}$, as usual. Then, the following holds true:
    \begin{itemize}
        \item[(i)]The operator $\mathcal{G}_\varsigma$ is invertible and its inverse is given by 
        \begin{equation}\label{eq:inverse-L}
            \hat{\phi}(t,\mathbf{p}) = \int_0^{t} \mathsf{K}(t-t',\mathbf{p}) (c+\partial_{t'}^2+\mathbf{p}^2)\hat{\Phi}(t',\mathbf{p}) \,dt'\,,
        \end{equation}
        where $\hat{\cdot}$ denotes the spatial Fourier transform, $1<c<4m^2$ is an arbitrary constant and where $\mathsf{K}$ is defined by means of its Laplace transform in time as 
        \begin{equation}\label{eq:Fw2}
            \mathscr{L}(\mathsf{K})(s,\mathbf{p}) =  - F(\mathbf{p}^2+s^2)^{-1} \quad \text{with} \quad
                F(w^2) := (w^2+c)g(-w^2) =\int_{4m^2}^{\infty} \varsigma(M)    \frac{w^2+c}{w^2+M}  \,dM\,.
        \end{equation}
        The function $F(w^2)$ is a strictly positive monotonically increasing function of $w^2\in[0,\infty)$. Hence, also $\mathscr{L}(\mathsf{K})$ is a positive, monotonically decreasing function of $w^2 = \mathbf{p}^2+s^2$ that is bounded. Moreover, the function $F(z)$ is analytic in $\mathbb{C}\setminus (-\infty,-4m^2]$.

        In the special case $\varsigma(M)=\rho(M)$ with $\rho$ as in Eq.~\eqref{eq:nonlocalK0}, it holds that
        \begin{equation*}
            \mathscr{L}(\mathsf{K})(s,\mathbf{p}) =  -\left(\frac{(w^2+c)}{16\pi^2}\left(\frac{2m}{w^3}\sqrt{4+\frac{w^2}{m^2}}\log \left({\frac{w}{2m}}+\sqrt{1+\frac{w^2}{4m^2}}\right)-\frac{2}{w^2}\right) \right)^{-1},\,\,\, w^2=\textbf{p}^{2}+s^2\,.
        \end{equation*}
     \item[(ii)]If $F(z)$ grows as $\log(z)$ for large $z$, as in the special case $\varsigma=\rho$, the inversion formula can be extended to past-compact smooth functions that are not necessarily spatially compact, in which case we have
\begin{equation*}
\phi :=  \mathcal{G}_\varsigma^{-1}(\Phi)= \int_\mathbb{R} \left(\varphi_{\mathrm{con}}(M)\theta(M-4m^2)+  \varphi_{\mathrm{disc}} \delta(M-c)\right)\mathsf{G}_{\mathrm{ret}}^M(
(\square- c)\Phi\big)\, dM\,,
\end{equation*}
where the continuous part $\varphi_{\mathrm{cont}}$ and the discrete part $\varphi_{\mathrm{disc}}$ are given by
\begin{align*}
\varphi_{\mathrm{con}}(M) 
=\frac{\varsigma(M)}{(c-M)(\Re g(M)^{2} +\pi^2\varsigma(M)^2)}, \qquad
\varphi_{\mathrm{disc}} = \lim_{\epsilon\to0} \frac{\epsilon}{F(-c+\epsilon )} =: \frac{1}{g(c)}\, .
\end{align*}
    \end{itemize}
\end{proposition}

\begin{remark}
For later applications, we note that the inversion formula above can be written in the following equivalent ways
\begin{equation}\label{eq:inverse-L-S}
\begin{aligned}
\phi 
&:=  \mathcal{G}_\varsigma^{-1}(\Phi)= 
\int_\mathbb{R} \left(\varphi_{\mathrm{con}}(M)\theta(M-4m^2)+  \varphi_{\mathrm{disc}} \delta(M-c)\right)\mathsf{G}_{\mathrm{ret}}^M((\square- c
)
\Phi) \,dM
\\
&= 
\varphi_{\mathrm{disc}} \mathsf{G}_{\mathrm{ret}}^{c}(( \square- c)\Phi) 
+
\int_\mathbb{R} \varphi_{\mathrm{con}}(M)\theta(M-4m^2)\mathsf{G}_{\mathrm{ret}}^M((\square- c)\Phi) \,dM
\\
&= 
\varphi_{\mathrm{disc}} \Phi 
+
\int_{4m^2}^\infty  \varphi_{\mathrm{con}}(M)
\mathsf{G}_{\mathrm{ret}}^M((
\square- c)\Phi)\, dM
\\
&= 
\varphi_{\mathrm{disc}} \Phi
+
((\square- c)\Phi)\int_{4m^2}^\infty \frac{\varphi_{\mathrm{con}}(M)}{M-c} dM
-
\int_{4m^2}^\infty \varphi_{\mathrm{con}}(M)\frac{1}{M-c}\mathsf{G}_{\mathrm{ret}}^M((\square-c)^2\Phi)\, dM
\end{aligned}
\end{equation}
and the continuous measure can also be obtained as the limit of the difference approaching the branch cut, i.e.
\begin{equation*}
    \varphi_{\mathrm{con}}(M) = \lim_{\epsilon \to 0^+} \frac{1}{2\pi \mathrm{i}} \left( \frac{1}{F(-M+\mathrm{i} \epsilon )} - \frac{1}{F(-M-\mathrm{i} \epsilon )} \right) \, .
\end{equation*}
\end{remark}

\begin{proof}[Proof of Proposition~\ref{prop:inverse-log}.]
Since both $\varsigma$, for $M\in[4m^2,\infty)$, and $\frac{w^2+c}{w^2+M}$, for every $w^2$, are square-integrable functions, we note that $F$ in Eq.~\eqref{eq:Fw2} is well-defined.
Furthermore, $F$ is a strictly positive function of $w^2$, because $\varsigma$ is positive and the integrand defining $F$ is positive as well. Moreover, it is also monotonically increasing on account of
\[
F'(w^2) = \int_{4m^2}^{\infty} \varsigma(M)    \frac{M-c}{(w^2+M)^2}  dM
\]
and the integrand is again $L^1$ on $[4m^2,\infty)$ and positive, since $1 < c < 4m^2$.

Let us now consider an arbitrary function $\phi \in C^{\infty}_{c}(\mathbb{R}_+\times \mathbb{R}^3)$ and the corresponding $\Phi := \mathcal{G}_\varsigma(\phi)$. Applying the Laplace transform in time and the Fourier transform in space and making use of the convolution theorem for the Laplace transform, we obtain
\begin{align*}
\mathscr{L}(\Phi)(s,\mathbf{p}) 
&=  
-
\mathscr{L}(\phi)(s,\mathbf{p})  \int_{4m^2}^{\infty} \varsigma(M)  \int_0^\infty e^{-st}  \frac{\sin (\omega_M(\mathbf{p})t)}{\omega_M(\mathbf{p})}\,
dt\, dM
\\
&=  -\mathscr{L}(\phi)(s,\mathbf{p})  \int_{4m^2}^{\infty} \varsigma(M)    \frac{1}{\mathbf{p}^2+s^2+M}\,  dM
\\
&=  -\mathscr{L}(\phi)(s,\mathbf{p})  \frac{1}{c+\mathbf{p}^2+s^2} F(\mathbf{p}^2+s^2)\,.
\end{align*}
The constant $c$ is positive and $F$ is invertible as a function of $\mathbf{p}^2+s^2$. Hence,
\begin{align*}
\mathscr{L}(\phi)(s,\mathbf{p}) &=   -(F(\mathbf{p}^2+s^2))^{-1}
(c+\mathbf{p}^2+s^2)
\mathscr{L}(\Phi)(s,\mathbf{p})\,,
\end{align*}
which is the Laplace transform of the sought inversion formula.
Moreover, we observe that the operator associated to $\mathsf{K}$ is bounded because $F^{-1}$ is a bounded function. In the particular case of $\varsigma=\rho$, the form of $\mathscr{L}(\mathsf{K})$ can be obtained, leading to the expression:
\[
\mathscr{L}(\mathsf{K})(s,\mathbf{p}) = - \left(\frac{(w^2+c)}{16\pi^2}
\left(\frac{2m}{w^3}\sqrt{4+\frac{w^2}{m^2}}\log \left(\frac{w}{2m}+\sqrt{1+\frac{w^2}{4m^2}}\right)-\frac{2}{w^2}\right) \right)^{-1}, \quad w^2 = s^2 + \mathbf{p}^2.
\]

The function $F(w^2)=g(-w^2)(w^2+c)$ above is such that $g(z)$ is the Stieltjes transform of the measure $\varsigma(M)dM$. As such it is analytic on the complex plane up to the points of the real line $[4m^2,\infty)$, where it has a cut. Similarly, $F$ is analytic on $\mathbb{C}\setminus (-\infty,-4m^2]$. 

To obtain the extension of the inverse of $\mathcal{G}_\varsigma$ on past-compact function we now proceed as follows. The inverse of the function $F(-z)$ is also analytic on $\mathbb{C}\setminus ([4m^2,\infty)\cup\{c\})$, where $c$ is there because $c-z$ vanishes on $c$ and the only zero of $F(-z)$ in its  analytic domain occurs for $z=c$.
Hence,  $F(-z)^{-1}$ is the Stieltjes transform of a measure formed by two parts. A continuous part supported on $[4m^2,\infty)$ and an atomic part supported on $c$. 
It can be written as
\[
F(w^2)^{-1} = \int_\mathbb{R} \left(\varphi_{\mathrm{con}}(M)\theta(M-4m^2)+  \varphi_{\mathrm{disc}} \delta(M-c)\right)\frac{1}{M+w^2}\,   dM\, .
\] 
The continuous part is computed by the Stieltjes-Perron inversion formula, yielding
\begin{align*}
\varphi_{\mathrm{con}}(M) 
&= \lim_{\epsilon\to 0^+}\frac{1}{\pi} \Im F(-M+\mathrm{i}\epsilon)^{-1} 
\\
&
= \lim_{\epsilon\to 0^+}
\frac{1}{\pi}\frac{\Im F(-M + i\epsilon)}{(\Re F(-M))^{2} +(\Im F(-M))^2}
\\
&= 
\frac{(c-M)\varsigma(M)}{(\Re F(-M))^{2} +\pi^2 (c-M)^2 \varsigma(M)^2}\,,
\end{align*}
where we used the fact that $\lim_{\epsilon\to 0^+} \frac{1}{\pi} \Im F(-M+\mathrm{i}\epsilon) = (c-M)\varsigma(M)$ and the fact that the branch cat of $F(z)$ is in the imaginary part, while the real part is continuous.
The real part of $F(-M)$ is well-defined thanks to the H\"older continuity of $\varsigma$.

Notice that $\varphi_{\mathrm{con}}(M)$ decays as $1/\log(M)^2$ for large $M$. To prove this claim, we observe that the asymptotic growth of $F(-M)$ is $\log(M)$, by hypothesis.
Hence, since $\varsigma$ is continuous and $L^2$, we have that for $M\in [4m^2,\infty)$,
\[
|\varphi_{\mathrm{con}}(M)| \leq \frac{C}{(\log(M)-\log(2m^2))^2} 
\]
for some constant $C$.
This decay for large $M$ makes the function $\varphi_{\mathrm{con}}(M)M^{-1}$ integrable, hence the formula displayed in \eqref{eq:inverse-L-S} is well-defined also on past-compact smooth functions.
The discrete part of the measure is obtained with the residue theorem, i.e.
\[
\varphi_{\mathrm{disc}} = \lim_{\epsilon\to0} \frac{\epsilon}{F(-c+\epsilon )} = \frac{1}{g(c)}.
\]
With this inversion of the Stieltjes transform at disposal and with the explicit form of $\varphi_{\mathrm{disc}}$ and $\varphi_{\mathrm{con}}$ the inversion formula \eqref{eq:inverse-L} takes the form given in \eqref{eq:inverse-L-S}.
In view of the support properties of the retarded propagator and in view of the form of $\varphi_{\mathrm{con}}$, the inversion formula given in \eqref{eq:inverse-L-S} can be extended to past-compact smooth functions that are not necessarily spatially compact.
\end{proof}

\begin{remark}
    The operator $\mathscr{L}(\mathsf{K})$ is similar to the retarded operator $K$ studied in Proposition 5.3 of \cite{MedaPinSiemcosmo}, in which it is employed to prove the existence and uniqueness of solutions of the semiclassical Einstein equations in the cosmological case. Furthermore, if we denote by $\mathsf{K}_{\mathrm{ret}}$ the operator with integral kernel $\mathsf{K}$ given in Eq.~\eqref{eq:inverse-L}, then it holds that  $\mathsf{K}_{\mathrm{ret}}\circ\mathcal{K}_0=\square$.
\end{remark}

\subsubsection{The semiclassical Einstein equations as a mixed-dimensional local problem}
As an interesting observation, we note that the prototypical form of the semiclassical Einstein equations, Eq.~\eqref{eq:semiclassical-prototype}, can be reformulated as a \emph{local} mixed-dimensional problem for $\psi\oplus\phi$ defined over 
$(C^\infty_{\mathrm{pc}}(\mathcal{M}) \otimes C^0(\mathbb{R}) )\oplus C^\infty_{\mathrm{pc}}(\mathcal{M})$, explicitly given by
\begin{equation}\label{eq:local-problem}
\begin{cases}
(\tilde\square -4m^2) \psi = \phi\otimes \check\rho
\\
\sum_{j=0}^2b_j \square^j \phi= S-\iota_{\mathcal{M}}\square(\square-a_1)(\square-a_2) \psi\,,
\end{cases}
\end{equation}
where $\check\rho$ is the $L^2(\mathbb{R},dz)$-function introduced in Proposition \ref{prop:nonlocal-ret5d}; for $\psi$ defined on $\mathcal{M}_5$, where $\iota_{\mathcal{M}} \psi(\cdot) = \psi(\cdot,0)$ realises the restriction to  $\mathcal{M}\times \{0\} \subset \mathcal{M}_5$, as usual. Indeed, the equivalence of the Cauchy problems for Eq.~\eqref{eq:semiclassical-prototype} and Eq.~\eqref{eq:local-problem} can readily be established, as the following proposition demonstrates.

\begin{proposition}
Let $S$ be an assigned smooth past-compact source on $\mathcal{M}$ and let $\{a_i\}_{i\in\{1,2\}}$ and $\{b_j\}_{j\in\{0,1,2\}}$ be fixed real numbers. Moreover, consider $\psi\oplus \phi$ with smooth and past-compact $\phi$ and with $\psi \in 
C^\infty_{\mathrm{pc}}(\mathcal{M})\otimes C^0(\mathbb{R})$. Then, $\psi\oplus \phi$ is a solution of the local problem \eqref{eq:local-problem} if and only if $\phi$ is a solution of the nonlocal problem \eqref{eq:semiclassical-prototype}.
\end{proposition}
\begin{proof}
Let $\psi\oplus \phi$ be a solution of \eqref{eq:local-problem}. From the first equation in $\eqref{eq:local-problem}$ we have that 
\[
\psi = \tilde{\mathsf{G}}_{\mathrm{ret}}(\phi \otimes \check{\rho})\,,
\]
because the kernel of $\tilde{\mathsf{G}}_{\mathrm{ret}}$ is trivial on past-compact functions. Furthermore, as proven in Proposition \ref{prop:extension}, since $\check{\rho}\in L^2(\mathbb{R})$ and $\phi\in C^{\infty}_{\mathrm{pc}}(\mathcal{M})$, we conclude that $\tilde{\mathsf{G}}_{\mathrm{ret}}(\phi \otimes \check{\rho}) \in C^\infty_{\mathrm{pc}}(\mathcal{M})\otimes C^0(\mathbb{R})$.
Substituting this $\psi$ into the second equation, we obtain that $\phi$ is a solution of \eqref{eq:semiclassical-prototype}.

Consider now $\phi$, a past-compact solution of \eqref{eq:semiclassical-prototype}, and construct
\[
\psi = \tilde{\mathsf{G}}_{\mathrm{ret}}(\phi \otimes \check{\rho})\,,
\]
which is in $C^\infty_{\mathrm{pc}}(\mathcal{M})\otimes C^0(\mathbb{R})$ by Proposition \ref{prop:extension}. Thus, $(\tilde{\square}-4m^2)\psi = \phi \otimes \check{\rho}$, which is the first equation in \eqref{eq:local-problem}. Furthermore, the nonlocal equation \eqref{eq:semiclassical-prototype} is the second equation of 
\eqref{eq:local-problem} for $\psi$ and $\phi$ just defined, which concludes the proof
\end{proof}

\subsection{Reduction of the equation to a simpler one}

In order to rewrite equations of the form \eqref{eq:semiclassical-prototype} in a more manageable way, we operate as follows.
We introduce auxiliary functions\footnote{We recall that $\mathfrak{h}_{j}\in C^{1}(\mathbb{R})$, by the Sobolev embedding theorems.}  $\mathfrak{h}_{j}\in W^{2,2}(\mathbb{R})$ (the $L^{2}$-Sobolev space of order 2)
with $j\in \{0,1,2\}$ with the requirement that $\mathfrak{h}_{j}(0)=1$.
With these functions at disposal, the semiclassical equation 
\[
\iota_{\mathcal{M}}\tilde{\mathsf{G}}_{\mathrm{ret}} \at \square (\square-a_1)(\square-a_2)\phi\otimes \check\rho \ct +\sum_{j=0}^{2}b_j \square^j \phi=S
\] 
for a given smooth and past-compact source $S$,
can be written in the form
\[
\iota_{\mathcal{M}}\tilde{\mathsf{G}}_{\mathrm{ret}} \at \square (\square-a_1)(\square-a_2)\phi\otimes \check\rho +\sum_{j=0}^{2}(\tilde\square-4m^2)(b_j \square^j \phi\otimes \mathfrak{h}_j) \ct
=S\,,
\]
where $\tilde\square$ is the d'Alembert operator on $\mathcal{M}_5$. Using the Leibniz rule and expanding $\tilde\square$, we have
\[
\iota_{\mathcal{M}}\tilde{\mathsf{G}}_{\mathrm{ret}} \at \square (\square-a_1)(\square-a_2)\phi\otimes \check\rho +\sum_{j=0}^{2} ((\square-4m^2)b_j \square^j \phi\otimes \mathfrak{h}_j)+ 
 \sum_{j=0}^{2}b_j \square^j \phi\otimes \partial_z^2\mathfrak{h}_j \ct
 =S.
\] 
Now, let us further assume that $\mathfrak{h}_j$ satisfy the following set of ordinary differential equations for some constants $\beta_i$ and with the constraint $\mathfrak{h}_j(0)=1$:
\begin{equation}\label{eq:equationsforh}
\begin{cases}
b_2\partial_z^2{\mathfrak{h}_2} - 4m^2 b_2 {\mathfrak{h}_2}  + b_1 {\mathfrak{h}_1}  -(a_1+a_2) \check \rho  =  \beta_2(\check \rho+b_2\mathfrak{h}_2)
\\
b_1\partial_z^2{\mathfrak{h}_1} - 4m^2 b_1 {\mathfrak{h}_1}  + b_0 {\mathfrak{h}_0}  +(a_1a_2) \check \rho  =  \beta_1(\check \rho+b_2\mathfrak{h}_2)
\\
b_0\partial_z^2{\mathfrak{h}_0} - 4m^2 b_0 {\mathfrak{h}_0}    =  \beta_0(\check \rho+b_2\mathfrak{h}_2).
\end{cases}
\end{equation}
Then, once a set of such $\mathfrak{h}_i$ solving the previous equations with constraints for certain $\beta_i$ is obtained, the semiclassical equations become
\[
\iota_{\mathcal{M}}\tilde{\mathsf{G}}_{\mathrm{ret}} ((\square^3 \phi+\beta_2\square^{2} \phi + \beta_1 \square \phi +\beta_0\phi) \otimes ( \check \rho+b_2\mathfrak{h}_2 )) =S\,.
\] 
Now, according to Proposition \ref{prop:inverse-log}, if $\check \rho+b_2\mathfrak{h}_2 $ is positive and sufficiently regular, as we shall discuss at the end of the section, then the map $f \mapsto \iota_{\mathcal{M}}\tilde{\mathsf{G}}_{\mathrm{ret}}(f \otimes ( \check \rho+b_2\mathfrak{h}_2 ))$ has trivial integral kernel. In this case, the semiclassical equation is satisfied provided it holds that 
\begin{equation}\label{eq:semi-local}
\square^3 \phi+\beta_2\square^{2} \phi + \beta_1 \square \phi +\beta_0\phi ={S}_N\,,
\end{equation} 
where $S_{N}$ is a new source implicitly defined by $\iota_{\mathcal{M}}\tilde{\mathsf{G}}_{\mathrm{ret}}({S}_N\otimes  ( \check \rho+b_2\mathfrak{h}_2 )) = S$. In other words, if the required regularity and positivity conditions are matched, then the prototypical semiclassical equation~\eqref{eq:semiclassical-prototype} can equivalently be written into the form of the fully \emph{local} hyperbolic equation~\eqref{eq:semi-local}.

Now, for the rest of this section, the goal is to show that, under suitable restrictions on the parameters $a_{i}$ and $b_{i}$ appearing in the semiclassical equations~\eqref{eq:semiclassical-prototype}, the needed regularity and positivity conditions on $\check \rho+b_2\mathfrak{h}_2 $ are satisfied so that the previous discussion indeed applies. Based on this result, we shall then see in the next section, Section~\ref{sec:Subleading}, that a similar inversion of the nonlocality can actually always be achieved \emph{without} restrictions on the parameters $b_{i}$ when performed just for the highest-order terms of the equation, hence allowing us to rewrite the prototypical semiclassical equation~\eqref{eq:semiclassical-prototype} into a hyperbolic equation with \emph{subleading} nonlocal terms.

To start with, we discuss the solvability of the system in Eq.~\eqref{eq:equationsforh}.

\begin{lemma}\label{le:solution-mathfrakh}
Consider $\{a_j\}_{j\in\{1,2\}}\subset \mathbb{R}$, $\{b_j\}_{j\in\{0,1,2\}}\subset \mathbb{R}\setminus \{0\}$ and $\{\beta_j\}_{j\in\{0,1,2\}}\subset \mathbb{R}$ be chosen in such a way that the polynomial $M^3 +\beta_2 M^2 + \beta_1 M + \beta_0$ does not have zeros in $M\in[4m^2,\infty)$.

The solution of the system of equations given in \eqref{eq:equationsforh} with these parameters is best displayed in the Fourier domain in $z$-direction. Writing $M=p_z^2 +4m^2$, the solution has the form
\begin{align*}
b_2 \hat{\mathfrak{h}}_2(p_z) &=  \left(\frac{
M(M-a_1)(M-a_2) 
}{
 M^3 +\beta_2 M^2 + \beta_1 M + \beta_0  
 }
-1
\right)\hat{\check{\rho}}(p_z)
\\
b_1 \hat{\mathfrak{h}}_1(p_z) 
&=  
-\frac{\left(M\beta_1+\beta_0\right)}{M^2}(\hat{\check{\rho}}+b_2 \hat{\mathfrak{h}}_2)(p_z)
+ \frac{a_1a_2}{M} \hat{\check{\rho}}(p_z) 
\\
b_0 \hat{\mathfrak{h}}_0(p_z) &=  
-\frac{\beta_0}{M}(\hat{\check{\rho}}+b_2 \hat{\mathfrak{h}}_2)(p_z).
\end{align*}
Moreover, the corresponding functions $\mathfrak{h}_{j}$ satisfy $\{\mathfrak{h}_j\}_{j} \subset W^{2,2}(\mathbb{R})$  . 
\end{lemma}
\begin{proof}
 In Fourier domain in $z$-direction, the system of equations for $\mathfrak{h}_j$ in \eqref{eq:equationsforh} can be written as
 \[
 \begin{pmatrix}
 p_z^2 + 4m^2 + \beta_2 & -1 & 0  \\
 \beta_1 &  p_z^2+ 4m^2 & -1 \\
\beta_0 & 0&  p_z^2+ 4m^2 
 \end{pmatrix} 
 \begin{pmatrix}
b_2 \hat{\mathfrak{h}}_2
\\
b_1 \hat{\mathfrak{h}}_1
\\
b_0 \hat{\mathfrak{h}}_0
 \end{pmatrix} 
 = 
 \hat{\check{\rho}}
 \begin{pmatrix} 
 -\beta_2 - (a_1+a_2)
 \\
 -\beta_1 + a_1a_2
 \\
 -\beta_0
 \end{pmatrix}.
 \]
 Writing  $M=p_z^2+4m^2$ and inverting the previous relation, we obtain the claimed expressions for $b_{i}\hat{h}_{i}$.

Taking into account that $\hat{\check{\rho}}$ is a square-integrable function for every $j$, and assuming the hypothesis imposed on the constants $\beta_i$, we find
\[
|b_j\hat{\mathfrak{h}}_j|\leq |\hat{\check{\rho}}| \frac{C_j}{M(p_z)}
\]
for some positive constants $C_j$. Since $b_i\neq 0$, $\hat{\check{\rho}} \in L^2(\mathbb{R})$ and $M(p_z)=p_z^2+4m^2$, this bound implies that $\mathfrak{h}_j \in W^{2,2}(\mathbb{R})$.
\end{proof}
Once a set of $\mathfrak{h}_j$ solving \eqref{eq:equationsforh} is obtained, we can fix the  
parameter $\beta_0$, $\beta_1$ and $\beta_2$ with the condition $\mathfrak{h}_{j}(0)=1$, which in Fourier space translates to
\begin{equation}\label{eq:constraint}
\int_{4m^2}^\infty   \hat{\mathfrak{h}}_j  \frac{dM}{\sqrt{M-4m^2}} =1. 
\end{equation}

We discuss how to find those parameters in the next proposition.

\begin{proposition}\label{prop:gamma}
Consider $a_{1},a_2 \in \mathbb{R}$ and $b_0,b_1,b_2\in \mathbb{R}\setminus\{0\}$. 
Moreover, consider the Stieltjes transform of the density $\rho(M)$ given in \eqref{eq:nonlocalK0}, described by the complex-valued function 
\begin{equation}\label{eq:J}
J(z):= 
\int_{4m^2}^\infty \frac{\rho(M)}{M-z}dM  = 
\frac{1}{8\pi^2} \left(\frac{1}{z}
-
\frac{   (4 m^2 - z) 
     m \arccsc {\left(\frac{2m }{\sqrt{z}} \right)}}{m \sqrt{
 4 m^2 - z} z^{\frac{3}{2}}} 
\right) , \qquad z\in D\,,
\end{equation}
which has domain 
\begin{equation}\label{eq:domainD}
D=\{z\in \mathbb{C} \mid  \Im{z}\neq 0 \quad \text{or} \quad  \Im(z)=0,\;  \Re(z)< 4m^2\}.
\end{equation}
If the complex-valued equation 
\begin{equation}\label{eq:equation-beta}
\gamma_i(a_1-\gamma_i)(a_2-\gamma_i) J(\gamma_i)=b_0 +b_1\gamma_i+b_2\gamma_i^2, \qquad \gamma_i\in D
\end{equation}
admits three and only three independent solutions $\gamma_0\neq \gamma_1\neq\gamma_2$ in $D$, then the three parameters
\begin{equation}\label{eq:beta-gamma}\begin{aligned}
\beta_0 &=-\gamma_0\gamma_1\gamma_2,
\\
\beta_1 &= \gamma_0\gamma_1+\gamma_0\gamma_2+\gamma_1\gamma_2, 
\\
\beta_2 &= -\gamma_0-\gamma_1-\gamma_2
\end{aligned}
\end{equation}
are real-valued, the polynomial
$P(M)=M^3 +\beta_2 M^2 + \beta_1 M + \beta_0$ does not have zeros in $M\in[4m^2,\infty)$ and with these $\beta_{i}$, the constraint given in Eq.~\eqref{eq:constraint} is satisfied.
\end{proposition}
\begin{proof}
By the symmetry property of the real and imaginary part of \Eq \eqref{eq:equation-beta}, if $z$ is a solution, then also $\bar{z}$ solves the equation. If there are three and only three solutions $\gamma_i$ in $D$, then either all of them are real-valued, or one is real-valued with the other two being related by complex conjugation. In all these cases, $\beta_{i}$ are real-valued. 

Now, notice that 
\[
P(M)=M^3 +\beta_2 M^2 + \beta_1 M + \beta_0 = (M-\gamma_0)(M-\gamma_1)(M-\gamma_2).
\]
Hence, $\gamma_i$ coincide with the zeros of the polynomial $P(M)$ and lead to the above form of $\beta_i$, which are real by the previous remark. If there are three independent $\gamma_i$ contained in $D$ solving \eqref{eq:equation-beta}, then the polynomial $P(M)$ cannot have zeros in $[4m^2,\infty)$, since $D\cap [4m^2,\infty)=\emptyset$.

It remains to prove that with this choice of $\beta_i$, the constraints given in \eqref{eq:constraint} are satisfied.
We start rewriting the solution obtained in Lemma \ref{le:solution-mathfrakh} with the $\gamma_i$ at the place of $\beta_j$ to obtain
\begin{align*}
b_2 \hat{\mathfrak{h}}_2 &=  \left(\frac{
M(M-a_1)(M-a_2) 
}{
  (M-\gamma_0)(M-\gamma_1)(M-\gamma_2) 
 }
-1
\right)\hat{\check{\rho}}\,,
\\
b_1 \hat{\mathfrak{h}}_1 
&=  
\left(\left(-(\gamma_1\gamma_2 + \gamma_0\gamma_1+\gamma_1\gamma_2)+\frac{\gamma_0 \gamma_1\gamma_2}{M}\right)
\left(\frac{
(M-a_1)(M-a_2) 
}{
 (M-\gamma_0)(M-\gamma_1)(M-\gamma_2)  
 }
\right)
+ \frac{a_1a_2}{M} \right) \hat{\check{\rho}}\,, 
\\
b_0 \hat{\mathfrak{h}}_0 &=  
{\gamma_0 \gamma_1\gamma_2}
\left(\frac{
(M-a_1)(M-a_2) 
}{
 (M-\gamma_0)(M-\gamma_1)(M-\gamma_2)  
 }
\right)\hat{\check{\rho}}\, .
\end{align*}
We then use the following partial fraction expansion to rewrite the functions $\mathfrak{h}_i$:
\begin{equation}\label{eqRM}
\begin{aligned}
R(M)&=\frac{(M-a_1)(M-a_2)}{ (M-\gamma_0)(M-\gamma_1)(M-\gamma_2) } = \sum_{i=0}^2 \frac{A_i}{M-\gamma_i}, \qquad A_{i} = (R(M)(M-\gamma_i))|_{M=\gamma_i}
\\
M R(M)&=\frac{M(M-a_1)(M-a_2)}{ (M-\gamma_0)(M-\gamma_1)(M-\gamma_2) } = 1+\sum_{i=0}^2 \frac{\gamma_i A_i}{M-\gamma_i}, \qquad \sum_{i=0}^2A_{i} = 1
\\
\frac{R(M)}{M}&=\frac{(M-a_1)(M-a_2)}{ M(M-\gamma_0)(M-\gamma_1)(M-\gamma_2) } = -\frac{a_1 a_2}{\gamma_0\gamma_1\gamma_2 M}+\sum_{i=0}^2 \frac{A_i}{\gamma_i}\frac{1}{M-\gamma_i}, 
\end{aligned}
\end{equation}
Now, in terms of these expansions, we obtain
\begin{align*}
b_2 \hat{\mathfrak{h}}_2 &=  
\left(
M R(M)-1 
\right)\hat{\check{\rho}} 
=
\sum_{i=0}^2 \gamma_i A_i \frac{\hat{\check{\rho}}}{M-\gamma_i} 
\\
b_1 \hat{\mathfrak{h}}_1 
&=  
\left(-(\gamma_1\gamma_2 + \gamma_0\gamma_1+\gamma_1\gamma_2) R(M)+\gamma_0 \gamma_1\gamma_2\frac{R(M)}{M}
+ \frac{a_1a_2}{M}  \right)\hat{\check{\rho}} 
\\
&=   -(\gamma_1\gamma_2 + \gamma_0\gamma_1+\gamma_1\gamma_2)\left(\sum_{i=0}^2 A_i \frac{\hat{\check{\rho}}}{M-\gamma_i}\right) +       \gamma_0 \gamma_1\gamma_2\left(\sum_{i=0}^2  \frac{A_i}{\gamma_i} \frac{\hat{\check{\rho}}}{M-\gamma_i}  \right)
\\
b_0 \hat{\mathfrak{h}}_0 &=  
{\gamma_0 \gamma_1\gamma_2} R(M)
\hat{\check{\rho}}
=
{\gamma_0 \gamma_1\gamma_2}\sum_{i=0}^2 A_i \frac{\hat{\check{\rho}}}{M-\gamma_i}.
\end{align*}
We can now impose the constraints 
\[
\int_{4m^2}^\infty   \hat{\mathfrak{h}}_j  \frac{dM}{\sqrt{M-4m^2}} =1
\]
to get a condition that is satisfied by $\gamma_i$. These constraints, together with the partial fraction expansion obtained above and the definition of the function $J$, implies that 
\begin{align*}
b_2 &=  
\sum_{i=0}^2 \gamma_i A_i J(\gamma_i) 
\\
b_1  
&=   -(\gamma_1\gamma_2 + \gamma_0\gamma_1+\gamma_1\gamma_2)\left(\sum_{i=0}^2 A_i J(\gamma_i)\right) +       \gamma_0 \gamma_1\gamma_2\left(\sum_{i=0}^2  \frac{A_i}{\gamma_i} J(\gamma_i)  \right)
\\
b_0  &=  
{\gamma_0 \gamma_1\gamma_2}\sum_{i=0}^2 A_i 
J(\gamma_i)
\end{align*}
where as before $A_{i} = (R(M)(M-\gamma_i))|_{M=\gamma_i}$
and for every $i$, $A_i$ depends nonlinearly on $\{\gamma_j\}$ and on $a_j$.  
Assuming for the moment that for every $i$, $\gamma_i \neq 0$, $\gamma_i\neq a_1$ and $\gamma_i\neq a_2$, we rewrite this system of equations as 
\begin{align*}
\sum_{i=0}^2 \gamma_i A_i J(\gamma_i)  &=  b_2
\\
\sum_{i=0}^2  \frac{A_i}{\gamma_i} J(\gamma_i)  &= \frac{b_1}{\gamma_0 \gamma_1\gamma_2} + 
\frac{b_0}{(\gamma_0 \gamma_1\gamma_2)^2}(\gamma_1\gamma_2 + \gamma_0\gamma_1+\gamma_1\gamma_2)
\\
\sum_{i=0}^2 A_i 
J(\gamma_i)
  &=  
\frac{b_0}{\gamma_0 \gamma_1\gamma_2}.
\end{align*}
In order to solve for $J(\gamma_i)$, we rewrite the system as 
\begin{align*}
\begin{pmatrix}
\gamma_2 & \gamma_1 & \gamma_0
\\
\frac{1}{\gamma_2} & \frac{1}{\gamma_1} & \frac{1}{\gamma_0}
\\
1 & 1 & 1
\end{pmatrix}
\begin{pmatrix}
A_2 J(\gamma_2)
\\
A_1 J(\gamma_1)\\
A_0 J(\gamma_0)
\end{pmatrix}
=
\begin{pmatrix}
1 & 0 & 0
\\
0 & \frac{1}{\gamma_0\gamma_1\gamma_2}& \frac{1}{(\gamma_0\gamma_1\gamma_2) }(\frac{1}{\gamma_0}+\frac{1}{\gamma_1}+\frac{1}{\gamma_2})
\\
0 & 0 & \frac{1}{\gamma_0\gamma_1\gamma_2}
\end{pmatrix}
\begin{pmatrix}
b_2
\\
b_1
\\
b_0
\end{pmatrix}.
\end{align*}
Using the form of $A_i$ and inverting the appropriate matrices, we get that the system decouples and for every $i$ we obtain
\[
J(\gamma_i)=\frac{b_0 +b_1\gamma_i+b_2\gamma_i^2}{\gamma_i(a_1-\gamma_i)(a_2-\gamma_i)}\,,
\]
which can also be written as \eqref{eq:equation-beta} given above which is meaningful also when the 
working hypothesis $\gamma_i \neq 0$, $\gamma_i\neq a_1$ and $\gamma_i\neq a_2$ are not satisfied.
This concludes the proof.
\end{proof}

It remains to prove that ${\check{\rho}}+b_2 {\mathfrak{h}}_2$ is sufficiently regular so that 
\[
\Phi = \iota_\mathcal{M} \tilde{\mathsf{G}}_{\mathrm{ret}}(\phi\otimes (\check{\rho}+b_2 {\mathfrak{h}}_2))
\]
can be inverted by means of Proposition \ref{prop:inverse-log} for past-compact $\phi$. This is an essential step to extract the local equation with a source from the nonlocal one.

We have the following proposition, which gives a restriction on the parameters $a_i$.
\begin{proposition}\label{prop:hypo-inveresion}
Let $a_1,a_2 \leq 4m^2$ and $b_i$ be chosen as in Lemma \ref{le:solution-mathfrakh}. 
Suppose that we have found three independent $\gamma_i\in D$ solving Eq.~\eqref{eq:equation-beta} of Proposition \ref{prop:gamma} with $D$ given in \eqref{eq:domainD}. Then, the function $\varsigma$ defined by
\[
\check{\varsigma}(z)={\check{\rho}}(z)+b_2 {\mathfrak{h}}_2(z) , \qquad z\in \mathbb{R}, 
\]
with the Fourier-transformed functions $\check{\varsigma}$ and $\check{\rho}$ defined as in Eq.~\eqref{eq:rhocheck}, with $\rho$ defined in Eq.~\eqref{eq:nonlocalK0} and with $\mathfrak{h}_2$ as in the thesis of Lemma \ref{le:solution-mathfrakh}, satisfies the hypothesis of Proposition~\ref{prop:inverse-log}.
\end{proposition}
\begin{proof}
To prove the claim, we notice that the Fourier-Plancherel transform of the function 
$\check{\rho}+b_2 \mathfrak{h}_2$ is
\[ 
\frac{\hat{\check{\rho}}(p_z)+b_2 \hat{\mathfrak{h}}_2(p_z)}{|p_z|} = M R(M) \rho(M) \,,
\]
where $M=p_z^2+4m^2$ and $R(M)$ is given in \eqref{eqRM}.

By hypothesis, we have $a_i\leq 4m^2$. Moreover, by hypothesis, $\gamma_1\neq \gamma_2\neq \gamma_3$ are zeros of \eqref{eq:equation-beta}. Hence, by Proposition \ref{prop:gamma}, $\gamma_i\not\in[4m^2,\infty)$. By direct inspection, we note that the function $MR(M)\geq 0$ for every $M \in [4m^2,\infty)$ and $M R(M)$ is bounded. This implies that $M R(M) {\rho(M)} $ is positive and in $L^2([4m^2,\infty),dM)$. Thus, $\varsigma$ satisfies the hypothesis of Proposition \ref{prop:inverse-log}. In particular, the logarithmic growth of the corresponding $F(w^2)$ (the Stieltjes transform of $MR(M)\rho(M)$ multiplied by $c+w^2$) follows from the fact that $(MR(M)-1)\rho(M)$ is absolutely integrable and square-integrable and the fact that it can be discarded in the analysis of the asymptotic behaviour of $F(w^2)$, which is dominated by the contribution to the Stieltjes transform due to $\rho(M)$.
\end{proof}

\subsection{Local equation plus subleading nonlocal contributions}\label{sec:Subleading}

According to Proposition~\ref{prop:gamma}
and Lemma~\ref{le:solution-mathfrakh}, if we can find three independent solutions of 
\Eq~\eqref{eq:equation-beta}, then we can rewrite 
\[
\iota_{\mathcal{M}}\tilde{\mathsf{G}}_{\mathrm{ret}}( \square (\square-a_1)(\square-a_2)\phi\otimes \check\rho) +\sum_{j=0}^{2}b_j \square^j \phi=S
\] 
as
\[
\iota_{\mathcal{M}}\tilde{\mathsf{G}}_{\mathrm{ret}} ((\square^3 \phi+\beta_2\square^{2} \phi + \beta_1 \square \phi +\beta_0\phi) \otimes ( \check \rho+b_2\mathfrak{h}_2 )) =S.
\] 
In other words, all contributions can be incorporated into the argument of the nonlocal operator $\iota_{\mathcal{M}}\tilde{\mathsf{G}}_{\mathrm{ret}}$. We prove here that this is always the case for the term containing the highest order derivatives without restricting the hypothesis on the coefficients $b_i$. Hence, after inverting the nonlocal operator, we can write the semiclassical equation in normal form up to a term that is nonlocal but contains only lower-order derivatives of $\phi$. We have the following proposition.

\begin{proposition}\label{prop:normal-form}
Consider $a_1,a_2 \in (-\infty,4m^2)$. 
For every $\tilde{b}_0,\tilde{b}_1,\tilde{b}_2 \in \mathbb{R}$, there exist ${b}_0$ and ${b}_1$ in $\mathbb{R}$ and $b_2=\tilde{b}_2$, such that \Eq~\eqref{eq:equation-beta}  admits three real independent solutions $\gamma_i\in (-\infty,4m^2)$, and thus, 
according to Proposition \ref{prop:gamma}, we conclude that
\[
\iota_{\mathcal{M}}\tilde{\mathsf{G}}_{\mathrm{ret}}( \square (\square-a_1)(\square-a_2)\phi\otimes \check\rho) +\sum_{j=0}^{2}\tilde{b}_j \square^j \phi=S
\] 
can be rewritten into the form
\[
\iota_{\mathcal{M}}\tilde{\mathsf{G}}_{\mathrm{ret}} ((\square^3 \phi+\beta_2\square^{2} \phi + \beta_1 \square \phi +\beta_0\phi) \otimes ( \check \rho+b_2\mathfrak{h}_2 ))
+(\tilde{b}_1-b_1) \square \phi + 
(\tilde{b}_0-b_0) \phi
=S\,,
\]
where $\beta_i$ are as in 
Proposition~\ref{prop:gamma}.
Denoting $\check\varsigma=\check\rho + b_2\mathfrak{h}_2 $ and considering the corresponding $\varsigma$ obtained from $\check\varsigma$,
inverting the transformation given in \eqref{eq:rhocheck}, we conclude that this equation can be written as 
\[
\square^3 \phi+\beta_2\square^{2} \phi + \beta_1 \square \phi +\beta_0\phi
+
\mathcal{G}_{\varsigma}^{-1}( 
(\tilde{b}_1-b_1) \square \phi + 
(\tilde{b}_0-b_0) \phi)
=
\mathcal{G}_{\varsigma}^{-1}(S)\,,
\]
where $\mathcal{G}_{\varsigma}^{-1}$ is given in \eqref{eq:inverse-L-S} of Proposition~\ref{prop:inverse-log} and $\beta_i$ are obtained from $\gamma_i$ as in \eqref{eq:beta-gamma}.
\end{proposition}

To prove that there exists a choice of $b_i$ with $b_2=\tilde{b}_2$ such that \Eq 
\eqref{eq:equation-beta}
has three real independent solutions, we first rewrite this equation as \begin{equation}\label{eq:left-right}
(\gamma-a_1)(\gamma-a_2) J(\gamma)=\frac{b_0}{\gamma} +b_1+b_2\gamma, \qquad \gamma\in I=(-\infty,4m^2)\,,
\end{equation}
where $J$ is given in \eqref{eq:J} and it is restricted here to $I=(-\infty,4m^2)$. Then, we have the following technical result: 

\begin{lemma}\label{le:convexity}
Consider the function $A:I\to\mathbb{R}$ with $I=(-\infty,4m^2)$ defined in the left-hand side of \eqref{eq:left-right} as
\[
A(\gamma):=(a_1-\gamma)(a_2-\gamma) J(\gamma), \qquad \gamma\in I.
\]
The function $A$ 
is a convex function.
Furthermore, $A'(\gamma)$ is monotonically increasing function with image $A'(I)=\mathbb{R}$.
\end{lemma}

\begin{proof}
To establish convexity of $A$, we study $A''(\gamma)$. From 
the definition of $J$ given in \eqref{eq:J}, we have that 
\[
A''(\gamma) = 
\int_{4m^2}^\infty \rho(M)\frac{2 (M-a_1)(M-a_2)}{(M-\gamma)^3}\,dM\,.
\]
We observe that $a_1,a_2$ and $\gamma$ are contained in $(-\infty,4m^2)$. Hence, they are outside of the domain of integration of $M$. Furthermore, $\rho(M)$ is positive and thus the integration function is positive and integrable. This implies that the function $A$ is convex.

We study now the derivative $A'(\gamma)$.
Since $A$ is convex, $A'$ is monotonically increasing. To determine the image of $I$ under $A'$, we first note that 
\[
\lim_{\gamma\to -\infty} A'(\gamma)=-\infty,
\qquad 
\lim_{\gamma\to {4m^2}^-} A'(\gamma)=\infty\,.
\]
These two limits follow by direct inspection, since $\rho(M)=\frac{1}{16\pi^2} \frac{1}{M}\sqrt{1-\frac{4m^2}{M}}$ and thus, it follows that the second integral
in 
\begin{align*}
A'(\gamma) 
&= 
(\gamma-a_1)(\gamma-a_2) 
\int_{4m^2}^\infty \rho(M) \frac{1}{(M-\gamma)^2} dM
+
(2\gamma-a_1-a_2) 
\int_{4m^2}^\infty \rho(M) \frac{1}{(M-\gamma)} dM
\end{align*}
is finite for $\gamma$ near $4m^2$, while the first integral diverges for $\gamma$ near $4m^2$. 
For $\gamma\to -\infty$, the contribution from the second integral dominates that of the first, and it diverges to $-\infty$.
\end{proof}

\begin{proof}[Proof of Proposition~\ref{prop:normal-form}.]
By the previous lemma, it holds that $A'(I) = \mathbb{R}$. Hence, for any choice of $b_2 \in \mathbb{R}$, one can always find a point $\tilde\gamma \in I$ and a subordinate unique choice of $q\in \mathbb{R}$ (uniqueness is a consequence of monotonicity), such that $b_2 = A'(\tilde\gamma)$ and such that the graph of
\[
\gamma \mapsto q+b_2 \gamma
\]
is tangent to the graph of $A(\gamma)$ with intersection at $\tilde\gamma$. 

Consider now the function at the right-hand side, i.e.
\[
B(\gamma) := \frac{b_0}{\gamma} + b_1 + b_2 \gamma, \qquad \gamma \in I.
\]
This function has a vertical asymptote in $\gamma=0$ and an oblique asymptote at infinity with slope $b_2$. Therefore, keeping $b_2$ fixed, if we choose $b_0=\epsilon_1$ and $b_1 = q+\epsilon_2$ for sufficiently small and positive $\epsilon_1$ and $\epsilon_2$, the graph of the function $B(\gamma)$ is very close to the union of the vertical and the oblique asymptote of $B(\gamma)$. Hence, the equation
\[
A(\gamma)=B(\gamma), \qquad \gamma \in I
\]
has three independent real solutions in $I$: one near $0$, where $B(\gamma)$ has a vertical asymptote and $A(\gamma)$ is finite, and the other two near $\tilde\gamma$, that exist by convexity of $A$ (see Lemma~\ref{le:convexity}). More precisely, these two solutions are close to the two points in which $q+\gamma b_2$ crosses $A(\gamma)$ that exist by convexity of $A$.

The thesis follows from an application of Proposition \ref{prop:gamma} and Proposition \ref{prop:inverse-log} to the equation with coefficients $b_0,b_1,b_2$, which now satisfies the hypotheses thanks to the result of Proposition \ref{prop:hypo-inveresion}.
\end{proof}

For the remainder of this section, we study the Cauchy problem of the 6th order hyperbolic equation with subleading nonlocal interaction obtained in Proposition~\ref{prop:normal-form}, i.e.
\begin{equation}\label{eq:ret-mixed}
\square^3 \phi+\beta_2\square^{2} \phi + \beta_1 \square \phi +\beta_0\phi
+
\mathcal{G}_{\varsigma}^{-1}( 
(\tilde{b}_1-b_1) \square \phi + 
(\tilde{b}_0-b_0) \phi)
=
\mathcal{G}_{\varsigma}^{-1}(S)\,,
\end{equation}
which by Proposition \ref{prop:normal-form} is equivalent to the prototypical equation~\eqref{eq:semiclassical-prototype}. Since the nonlocal term contains only lower-order derivatives and is retarded in nature, its solutions can be studied with methods similar in spirit than the one used in \cite{FMS}, i.e.~by employing perturbation theory, albeit in a more specialised and slightly different context.

To start with, let us consider the leading-order local part of Eq.~\eqref{eq:ret-mixed}. The corresponding differential operator can be factorised as
\begin{equation*}
    L=(\square^3 +\beta_2\square^{2}  + \beta_1 \square  +\beta_0) = (\square -\gamma_0)(\square -\gamma_1)(\square -\gamma_2),
\end{equation*}
where $\gamma_i\in (-\infty,4m^2)$.
The operator $L$ admits a retarded fundamental solution $\mathsf{G}_{L,\mathrm{ret}}$, which is obtained as a suitable linear combination of the retarded fundamental solutions of $\square -\gamma_i$. 
Notice that, even if some of the $\gamma_i$ could be negative, the retarded fundamental solution is still well-defined. In particular, we have that
\[
\mathsf{G}_{L,\mathrm{ret}} (\phi) = 
\sum_{i=0}^2
c_i \mathsf{G}^{\gamma_i}_{\mathrm{ret}} (\phi)\,,
\]
where $c_i$ are the real coefficients of the partial fraction expansion
\[
f(M):=\frac{1}{(M-\gamma_0)(M-\gamma_1)(M-\gamma_2)} = \sum_{i=0}^2 c_{i} \frac{1}{M-\gamma_i}, \qquad c_{i} = \lim_{M\to\gamma_i} (M-\gamma_i)f(M)\,.
\]
Having defined the retarded fundamental solution of the local operator $L$, we can rewrite \Eq~\eqref{eq:ret-mixed} in the form
\begin{align}\label{eq.PertEq}
    \phi+W_{\mathrm{ret}}(\phi)=\mathsf{S}\,,
\end{align}
where the nonlocal operator $W_{\mathrm{ret}}$ and the past-compact source $\mathsf{S}$ are defined as
\begin{align*}
    W_{\mathrm{ret}}(\phi) :=
\mathsf{G}_{L,\mathrm{ret}} (\mathcal{G}_{\varsigma}^{-1}( 
((\tilde{b}_1-b_1) \square  + 
(\tilde{b}_0-b_0) )\phi))\qquad\text{and}\qquad \mathsf{S} := \mathsf{G}_{L,\mathrm{ret}} (\mathcal{G}_{\varsigma}^{-1}(S))\,,
\end{align*}
respectively. The idea is now to use perturbation theory to obtain the unique past-compact solution to this equation. More precisely, we define a sequence $(\phi_{n})_{n\in\mathbb{N}}\in C^{\infty}_{\mathrm{pc}}(\mathcal{M})^{\mathbb{N}}$ inductively by
\begin{align*}
    \phi_0 := \mathsf{S},\qquad \phi_n:= W_{\mathrm{ret}}(\phi_{n-1})\quad \text{for} \quad n>0\,.
\end{align*}
Then, the idea is to show that the series defined by
\begin{equation}\label{eq:dyson-series}
\phi = \sum_{n\geq 0} (-1)^{n}\phi_{n}=\sum_{n\geq 1}(-1)^{n}W_{\mathrm{ret}}^{n}(\mathsf{S})
\end{equation}
converges absolutely in\footnote{The space $C^{\infty}_{\mathrm{pc}}(\mathcal{M})$ is naturally equipped with its LF-space topology induced by the family of seminorms $\Vert\varphi\Vert_{\alpha,K}:=\sup_{(t,\textbf{x})\in K}\vert\partial^{\alpha}\varphi(t,\textbf{x})\vert$ labelled by $\alpha\in\mathbb{N}^{4}$ and past-compact $K\subset\mathcal{M}$. In other words, it is the \emph{strict direct limit} of the Fréchet spaces $C_{K}^{\infty}(\mathcal{M}):=\{\phi\in C^{\infty}(\mathcal{M})\mid\mathrm{supp}(\phi)\subset K\}$ equipped with the subspace topologies inherited from $C^{\infty}(\mathcal{M})$. As a such, it is a \emph{complete locally convex topological vector space} that is, however, not metrisable. Nevertheless, due to completeness, absolute convergence of a series implies convergence.} $C^{\infty}_{\mathrm{pc}}(\mathcal{M})$. By construction, it is then clear that $\phi$ is indeed a past-compact and smooth solution to the \Eq~\eqref{eq.PertEq}. 

To prove that the series in Eq.~\eqref{eq:dyson-series} indeed converges, we use the fact that $W_{\mathrm{ret}}$ is a retarded operator and its continuity properties together with Gr\"onwall lemma in integrated form. To this end, we have the following theorem.

\begin{theorem}\label{thm:convergence}
For $a_i \in (-\infty,4m^2)$ and for any choice of past-compact smooth source $\mathsf{S}$ supported in the future of $\Sigma_{-\epsilon}$, the equation
$$
\phi+W_{\mathrm{ret}}(\phi)=\mathsf{S}\,,
$$
which is equivalent to the prototypical semiclassical equation~\eqref{eq.PertEq},
admits a unique smooth and past-compact solution $\phi$, supported in the future of $\Sigma_{-\varepsilon}$.
\end{theorem}

\begin{proof}
It remains to prove that the series defining $\phi$ in \eqref{eq:dyson-series} converges in the LF-space $C^{\infty}_{\mathrm{pc}}(\mathcal{M})$.
We start by analysing the convergence in the uniform topology. Consider a generic point $(T,\mathbf{x})\in\mathcal{M}$. The operator $W_{\mathrm{ret}}$ is retarded, hence, due to its support properties we can without lose of generality assume that $\mathsf{S}$ is also spatially compact. If this not the case, we can multiply $\mathsf{S}$ with a smooth characteristic function $\chi$ of $J^{-}(T,\mathbf{x})$, making it spatially compact without altering the result. Therefore, without loss of generality, we may replace $\mathsf{S}$ with $\chi \mathsf{S}$ in this proof.

We proceed by analysing the regularity of the operator $W_{\mathrm{ret}}$.
After using the Fourier-Laplace transform, we have that
\begin{align*}
\mathscr{L}\mathscr{F}(W_{\mathrm{ret}}\psi)(s,\mathbf{p}) 
&= -\frac{(w^2+c)((\tilde{b}_1-b_1)w^2 + \tilde{b}_0-b_0)}{(w^2+\gamma_0)(w^2+\gamma_1)(w^2+\gamma_2)}
\mathscr{L}(\mathsf{K})(s,\mathbf{p})
 \hat{\psi}(s,\mathbf{p}), \qquad w^2  = s^2 + \mathbf{p}^2,
\end{align*}
where we used the operator 
$\mathscr{L}(\mathsf{K})(s,\mathbf{p})=F(s^2+\mathbf{p}^2)^{-1}$ 
given in Eq.~\eqref{eq:Fw2} in terms of $\varsigma=\rho+b_2\mathfrak{h}_2$.

The operator with integral kernel $\mathsf{K}$ is bounded and we shall control later that it is locally integrable in time uniformly in $\mathbf{p}$. 
The remaining contribution to $W_{\mathrm{ret}}$ admits the partial fraction expansion 
\[
L(w^2)=-\frac{(w^2+c)((\tilde{b}_1-b_1)w^2 + \tilde{b}_0-b_0)}{(w^2+\gamma_0)(w^2+\gamma_1)(w^2+\gamma_2)} = -\sum_{i=0}^2 \frac{d_i}{w^2+\gamma_i}, \qquad d_i = \lim_{w^2\to -\gamma_i} -L(w^2)(w^2+\gamma_i).
\]
We thus have that 
\begin{equation}\label{eq:decomposeWret}
W_{\mathrm{ret}}(\phi) = \sum_{i=0}^2 d_i\mathsf{G}^{\gamma_i}_{\mathrm{ret}}\circ \mathsf{K} \phi .
\end{equation}

To prove that the integral kernel of $\mathsf{K}(t,\mathbf{p})$ is locally integrable in $t$  (and the result is uniformly bounded in $\mathbf{p}\in \mathbb{R}^3$ ), we use a result similar to the one obtained in Lemma 5.5 and Lemma 5.6 in \cite{MedaPinSiemcosmo}, after isolating the leading contribution in the integral appearing in $\mathsf{K}(t,\mathbf{p})$. 

Recalling the 
expression for $\mathsf{K}$ obtained in \eqref{eq:inverse-L-S}, we have that
\[ 
\mathsf{K}(t,\mathbf{p})=\varphi_{\mathrm{disc}} \frac{\sin{(\omega_ct)}}{\omega_c}
-
\int_{4m^2}^\infty  \varphi_{\mathrm{con}}(M)
\frac{\sin(\omega_M t)}{\omega_M} 
 \,dM,
\]
where $\omega_c=\sqrt{\mathbf{p}^2+c}$, where $\omega_M=\sqrt{\mathbf{p}^2+M}$, where $\varphi_{\mathrm{disc}}$ is a suitable constant and where
\begin{align}\label{eq:varphile}
\varphi_{\mathrm{con}}(M) 
=
\frac{(c-M)\varsigma(M)}{(\Re F(-M))^{2} +\pi^2 (c-M)^2 \varsigma(M)^2}
=\varphi_{\mathrm{lead}}(M)+\mathcal{R}(M)\,,
\end{align}
with the leading-order contribution $\varphi_{\mathrm{lead}}$ given by
\[
\varphi_{\mathrm{lead}} = -\frac{16\pi^2}{\log\left(\frac{M}{4m^2} \delta\right)^2+\pi^2}
\]
and with $\mathcal{R}(M)$ being a remainder term. 
The leading-order contribution in $M$ is obtained from the asymptotic form of $\Re F (-M)$, which is proportional to the Hilbert transform of $\varsigma$ evaluated in $-M$ and multiplied by $(c-M)$. The asymptotic form of the Hilbert transform of $\theta(M-4m^2)/M$ is proportional to $M^{-1}\log(\frac{M}{4m^2})$ and can be obtained by direct computation. Instead, the subleading contribution is given by the asymptotic form of the Hilbert transform of an $L^1$-function supported in $[4m^2,\infty)$, captured by 
\[
\frac{\log \delta}{M}  = \frac{1}{M}\int_{-\infty}^\infty\left(16\pi^2 \varsigma(y)-\frac{1}{y} \right) \theta(y - 4m^2)\, dy
\]
for some positive $\delta$. Furthermore, for large values of $M$, the asymptotic form of the numerator is
\[
(c-M)\varsigma(M) \sim- \frac{1}{16\pi^2}.
\]
This implies that the reminder $\mathcal{R}(M)=\varphi_{\mathrm{con}}(M)-\varphi_{\mathrm{lead}}(M)$ is a continuous square-integrable function on $[4m^2,\infty)$ and it is such that $\mathcal{R}(M)M^{-\frac{1}{2}}$ is also absolutely integrable on $[4m^2,\infty)$.

The leading contribution in $\mathsf{K}(t,\mathbf{p})$ is given by
\[
C(t,\mathbf{p}) = -\int_{4m^2}^\infty  \varphi_{\mathrm{lead}}(M)
\frac{\sin(\omega_M t)}{\omega_M}\, 
 dM\, .
\]
By Lemma \ref{lem:estimateLead}, which we defer the proof of until below to keep the focus here on the main estimates, we conclude that this contribution is locally integrable near $t=0$ and uniformly bounded away from $0$.
The remainder 
\[
\mathsf{K}(t,\mathbf{p})-C(t,\mathbf{p})  
=\varphi_{\mathrm{disc}} \frac{\sin{(\omega_ct)}}{\omega_c}
-
\int_{4m^2}^\infty  \mathcal{R}(M)
\frac{\sin(\omega_M t)}{\omega_M} 
 dM
\]
is also bounded in $t$, uniformly in $\mathbf{p}$, due to the properties of $\mathcal{R}(M)$.  

This observation implies that 
there exists a suitable constant $B$, such that 
\[
|\mathsf{K}\hat{\phi}(t,\mathbf{p})| \leq t B
|\hat{\phi}(t,\mathbf{p})|.
\]
With this estimate at disposal, recalling the expression of $W_{\mathrm{ret}}$ given in 
\eqref{eq:decomposeWret} and 
using the properties of the retarded operators, we obtain the following estimate for the series constructing $\phi$ out of $\mathsf{S}$ given in \eqref{eq:dyson-series}: 
\begin{equation}\label{eq: stimaGron}
    \begin{aligned}
    |\mathscr{F}\phi(T,\mathbf{p})| &= |\sum_{n\geq 1} (-1)^{n}\mathscr{F}W_{\mathrm{ret}}^n (\chi \mathsf{S})(T,\mathbf{p})| \\
    &\leq \left(\exp\left(\sum_i |d_i||TB|\int_0^T \frac{|\sin((T-u) \sqrt{\mathbf{p}^2+\gamma_i})|}{|\sqrt{\mathbf{p}^2+\gamma_i}|} du\right)
-1\right)\int_0^T |\mathscr{F}(\chi\mathsf{S})(t,\mathbf{p})| dt
\end{aligned}
\end{equation}
The function $\chi\mathsf{S}(t,\mathbf{p})$ is square-integrable in $\mathbf{p}$ for all $t\leq T$, because it is the Fourier transform of a smooth and compactly supported function in space. The previous bound proves that the series defining $\phi$ converges absolutely at all times. 
In particular, notice that, if $\gamma_i$ is positive, we may estimate 
\[
\frac{|\sin((T-u) \sqrt{\mathbf{p}^2+\gamma_i})|}{|\sqrt{\mathbf{p}^2+\gamma_i}|}\leq \frac{1}{\sqrt{\gamma_i}},
\]
while, if it is negative, we use the estimate
\[
\frac{|\sin((T-u) \sqrt{\mathbf{p}^2+\gamma_i})|}{|\sqrt{\mathbf{p}^2+\gamma_i}|}\leq \frac{e^{T |\gamma_i|}}{\sqrt{|\gamma_i|}}
\]
and thus
\begin{equation*}
    |\mathscr{F}\phi(T,\mathbf{p})|
\leq f(T)\int_0^T |\mathscr{F}(\chi\mathsf{S})(t,\mathbf{p})| dt\,,
\end{equation*}
where 
\[
f(T)=\exp\left( T^2|B|\sum_i |d_i|\frac{
e^{T\sqrt{|\gamma_i|}}
}{\sqrt{|\gamma_i|}}\right)-1
\]
is a continuous function on $[0,\infty)$. The very same bound then holds also in the uniform topology, i.e.~the series \eqref{eq:dyson-series} converges absolutely in the supremum norm $\| \cdot \|_\infty$ for any fixed time $T$. 
We begin by noticing that the uniform norm of the $n$-th term is bounded by the $L^1$ norm of its Fourier transform, namely
$$
\|\phi_n(T, \cdot)\|_\infty = \sup_{x \in \mathbb{R}^3} \left| \frac{1}{(2\pi)^3} \int_{\mathbb{R}^3} e^{i p \cdot x} \mathcal{F}\phi_n(T, p) \, d^3p \right| \le \frac{1}{(2\pi)^3} \int_{\mathbb{R}^3} |\mathcal{F}\phi_n(T, p)| \, d^3p
$$
To show absolute uniform convergence, we sum these norms over all $n$ and we get
\begin{align*}
    \sum_{n \ge 1} \|\phi_n(T, \cdot)\|_\infty & \le \frac{1}{(2\pi)^3}  \sum_{n \ge 1}\int_{\mathbb{R}^3}   |\mathcal{F}\phi_n(T, p)| \, d^3p  \le \frac{1}{(2\pi)^3}  \int_{\mathbb{R}^3}  \sum_{n \ge 1} |\mathcal{F}\phi_n(T, p)| \, d^3p \\
    & \le \frac{1}{(2\pi)^3} \int_{\mathbb{R}^3} f(T) \int_0^T |\mathcal{F}(\chi S)(t, p)| dt\, d^3p  \le \frac{f(T)}{(2\pi)^3} \int_0^T \left( \int_{\mathbb{R}^3} |\mathcal{F}(\chi S)(t, p)| \, d^3p \right) dt\,,
\end{align*}
where we used Fubini-Tonelli's theorem and the bounds similar the one obtained in Eq.~\eqref{eq: stimaGron}.
Since $\chi S(t,x)$ is a smooth and compactly supported function, $\mathcal{F}(\chi S)(t, p)$ decays faster than any polynomial in $p$, hence 
$$
\int_{\mathbb{R}^3} |\mathcal{F}(\chi S)(t, p)| \, d^3p \leq \infty
$$
and continuous for all $t \in [0, T]$. Using that also $f(T)$ is continuous, the series converges absolutely in the uniform topology. 

To conclude the proof, we observe that this result can be used to prove convergence in all other seminorms of the LF-space $C^{\infty}_{\mathrm{pc}}(\mathcal{M})$. To this end, notice that the spacetime derivatives of all order commute with  $\mathsf{G}^{\gamma_i}_{\mathrm{ret}}$, for every $i$, as well as with $\mathsf{K}$. Therefore, they commute with $W_\mathrm{ret}$. Hence, $W_\mathrm{ret}$ and $\sum_n (-1)^{n}W_\mathrm{ret}^n$ are continuous for all the seminorms of the LF-space $C^\infty_{\mathrm{pc}}(\mathcal{M})$ with the same controlling constants. This allows us to prove convergence in the uniform norm that implies that in all the other seminorms.

Uniqueness is proved as follows. Consider two possibly distinguished solutions of \Eq \eqref{eq.PertEq}: $\phi_1, \phi_2 \in C^{\infty}_{\mathrm{pc}}(\mathcal{M})$. Then, their difference $\phi = \phi_1 - \phi_2 \in C^{\infty}_{\mathrm{pc}}(\mathcal{M})$, is a solution of the homogeneous equation, namely
\begin{equation*}
    \phi + W_{\mathrm{ret}}(\phi) = 0.
\end{equation*}
Let us introduce the temporal action of $W_{\mathrm{ret}}$ as that map $W_{\mathrm{ret}}(t,\tau): C^{\infty}_{\mathrm{c}}(\Sigma_{\tau}) \to C^{\infty}_{\mathrm{c}}(\Sigma_{t})$ between test functions on the respective Cauchy hypersurfaces at time $\tau$ respectively $t$ (cf.~\cite[Sec.~3.2]{FMS}), so that for all $\phi \in C^{\infty}_{\mathrm{pc}}(\mathcal{M}) \cap C^{\infty}_{\mathrm{sc}}(\mathcal{M})$, denoting by $\phi_{\tau}$ its restriction to the Cauchy surface $\Sigma_{\tau}$,
\begin{equation*}
    (W_{\mathrm{ret}} \phi)(t,\mathbf{x}) = \int_{-\infty}^t (W_{\mathrm{ret}}(t,\tau)\phi_{\tau} )(\mathbf{x}) d\tau  .
\end{equation*}
In particular, its existence is obtained following the argument in \cite[Section $3$]{FMS} and using the properties of $W_{\mathrm{ret}}$ discussed above. Furthermore, let us denote by $ \| \cdot \|_t$ the $L^2$-norm on compactly supported smooth functions over the Cauchy surface at fixed time $t$. Then, combining H\"older inequality with Fubini theorem first and an argument similar to the one used to obtain the estimate in \Eq \eqref{eq: stimaGron}
\begin{align*}
    \Vert\phi_t\Vert_{t}\leq\int_{-\infty}^{t}\Vert W_{\mathrm{ret}}(t,\tau)\phi_{\tau}\Vert_{t}\,d\tau\leq C(e^{c t}+1)\int_{-\infty}^{t}\Vert \phi_{\tau}\Vert_{\tau}\,d\tau.
\end{align*}

Now, consider $F(t):= \int_{-\infty}^{t}\Vert\phi_{\tau}\Vert_{\tau}d\tau$. Then, the above equation translates to $\dot{F}(t)\leq C(e^{t c}+1)F(t)$ or in other words
\begin{align*}
    \frac{d}{dt}(e^{-C f(t) }F(t))\leq 0, \qquad f(t) = \int^t_{\inf t(\textrm{supp} \phi)} (e^{c\tau}+1)d\tau.
 \end{align*}
The function $G(t):=e^{- Cf(t)}F(t)$ is non-negative, i.e. $G(t)\geq 0$, and, by construction, $G(t)$ is monotonically decreasing since $\dot{G}(t)\leq 0$. Moreover, by past-compactness, $G(t)=0$ for all $t<t_0$ for some fixed $t_0$. Hence, we conclude $G(t)=0$ and so $\phi=0$. 
\end{proof}

We conclude this section with a technical Lemma that was needed in the proof of Theorem~\ref{thm:convergence}.

\begin{lemma}\label{lem:estimateLead}
Consider the function
\[
C(t,\mathbf{p}) = -\int_{4m^2}^\infty  \varphi_{\mathrm{lead}}(M)
\frac{\sin(\omega_M t)}{\omega_M} 
 dM\,,
\]
where $\varphi_{\mathrm{lead}}$ is given in \eqref{eq:varphile}. Then, for $t\neq 0$, it holds that
\[
|C(t,\mathbf{p})| \leq \frac{C}{t}\,,
\]
where the bound is uniform in $\mathbf{p}$. Furthermore, $C(t,\mathbf{p})$ is locally integrable near $t=0$ and the integral near $0$ is uniformly bounded in $p$.
\end{lemma}
\begin{proof}
We proceed as in Lemma 5.6 of \cite{MedaPinSiemcosmo}.
Integrating by parts, we have
\[
C(t,\mathbf{p}) = -2\int_{4m^2}^\infty   \varphi_{\mathrm{lead}}'(M)
\frac{\cos(\omega_M t)-\cos(\omega_{4m^2} t)}{t} 
 dM\,,
\]
where the second term comes from the boundary term. Now, the derivative of $\varphi_{\mathrm{lead}}(M)$ is
\begin{equation}\label{eq:philedderivative}
\varphi_{\mathrm{lead}}'(M) = \frac{16\pi^2}{M} \frac{\log (\frac{M}{4m^2}\delta)}{(\log(\frac{M}{4m^2}\delta)^2 +\pi^2)^2}\,,
\end{equation} 
which is an integrable function. We thus get the first uniform bound of $C(t,\mathbf{p})$, namely $C(t,\mathbf{p})$ is bounded for large values of $t$ uniformly in $\mathbf{p}$ and decays as $1/t$.

To prove that $C(t,\mathbf{p})$ is locally integrable near $t=0$, let $t$ be a small positive constant such that $1/t > 4m^2$ and split the region of integration as
\[
\int_{4m^2}^{1/t} dM + \int_{1/t}^\infty dM 
\]
The first contribution is such that 
\begin{align*}
    T_1 &= -2\int_{4m^2}^{1/t}  \varphi_{\mathrm{lead}}'(M) \frac{\cos(\omega_M t)-\cos(\omega_{4m^2} t)}{t} dM \\
        &= \frac{4}{t} \int_{4m^2}^{1/t}  \varphi_{\mathrm{lead}}'(M) \sin\left( \frac{(\omega_M-\omega_{4m^2}) t}{2} \right) \sin\left( \frac{(\omega_M + \omega_{4m^2}) t}{2} \right)\,,
\end{align*}
which, when approximating the sine function with its argument, gives
\begin{align*}
    |T_1| &\leq t \int_{4m^2}^{1/t}  \left|\varphi_{\mathrm{lead}}'(M) \right| ( M -4m^2) dM \\
    &\leq (1 - 4m^2 t) \int_{4m^2}^{+\infty}  \left|\varphi_{\mathrm{lead}}'(M) \right| dM \leq C (1 - 4m^2 t)\,.
\end{align*}
The latter estimate implies that $T_1$ is integrable near $t$ when $t$ is very close to zero.

The second integral can be written as
\[
T_2=-2\int_{1/t}^{\infty}  \varphi_{\mathrm{lead}}'(M)
\frac{\cos(\omega_M t)-\cos(\omega_{4m^2} t)}{t}
 dM\, .
\]
Using the explicit expression of $\varphi'_{\mathrm{lead}}(M)$ as given in \eqref{eq:philedderivative}, we conclude that the integral $T_{2}$ is bounded. In particular, it holds that 
\begin{align*}
|T_2| &\leq \frac{16\pi^2}{t}\int_{\delta/(4m^2t)}^\infty
\frac{1}{M}\frac{\log (M)}{(\log(M)^2 +\pi^2)^2}
 dM
  \\
 &\leq \frac{8\pi^2}{t}
 \frac{1}{\log(\frac{4m^2 t}{\delta})^2 +\pi^2}\,,
 \end{align*}
 where the latter bound is integrable for $t$ near $0$.
\end{proof}

\subsection{Asymptotic analysis for large time}
In Theorem \ref{thm:convergence}, we have established that solutions of~\eqref{eq:ret-mixed} always exist under mild assumptions on the parameters appearing in \Eq~\eqref{eq:ret-mixed} and for a generic past-compact smooth source. In Proposition \ref{prop:normal-form}, we have seen that if $a_1,a_2<4m^2$, the semiclassical equation given in \eqref{eq:semiclassical-prototype} can be put in the form \eqref{eq:ret-mixed}. Thus, combining the outcome of that Proposition with that of Theorem \ref{thm:convergence}, we can find solutions of the semiclassical Einstein equation. 
According to \eqref{eq:dyson-series},
the solution at the point $(t,x)$ only depends on the form of the source $S$ in the causal past of $(t,x)$. 
Hence, the sum given in \eqref{eq:dyson-series} and its convergence established in Theorem \ref{thm:convergence}, implicitly defines a retarded operator that maps smooth past-compact sources to smooth past-compact solutions. In this section, we study the large time behaviour of this retarded operator and of the solutions corresponding to decaying or compactly supported sources. 

To this end, we start by rewriting the retarded operator of \eqref{eq:semiclassical-prototype}
using the Fourier-Laplace transform.
In particular, we see that 
\[
Q(w^2)\mathscr{L}\mathscr{F}\phi(s,\mathbf{p}) = 
\mathscr{L}\mathscr{F}S(s,\mathbf{p}), \qquad w^2  = s^2 + \mathbf{p}^2\,,
\]
where
\begin{equation}\label{eq:Q}
Q(w^2) = 
w^2(w^2+a_1)(w^2+a_2)J(-w^2)+b_0-b_1w^2+b_2w^4
\end{equation}
with
\[
J(-w^2) = g_\rho(-w^2) =
\int_{4m^2}^{\infty} \rho(M)\frac{1}{w^2+M}  dM \, .
\]
By construction, $Q(z)$ is an analytic function on $\mathbb{C} \setminus (-\infty,-4m^2]$.

We hence conclude that the solution of the semiclassical equation can be written as
\begin{align*}
\mathscr{L}\mathscr{F}(\phi)(s,\mathbf{p}) 
&= \frac{1}{Q(w^2)}
 \hat{S}(s,\mathbf{p}), \qquad w^2  = s^2 + \mathbf{p}^2.
\end{align*}
To analyse the large time behaviour of that solution, we compute the form the operator associated to this inverse formula.
With the help of 
Stieltjes-Perron inversion formula,  used in a way similar to the one presented in the proof of Proposition \ref{prop:inverse-log},
we obtain
\begin{equation}\label{eq:retarded-propagator-full}
\phi = \sum_{\gamma \in \mathcal{Z}}
\frac{1}{Q'(-\gamma)} \mathsf{G}_{\mathrm{ret}}^{\gamma}(S) 
+ 
\int_{4m^2}^\infty
\vartheta(M)\mathsf{G}_{\mathrm{ret}}^{M}(S)\, dM\, ,
\end{equation}
where 
$\mathcal{Z}$ is the set of zeros of $Q$ in $\mathbb{C}\setminus (-\infty,-4m^2)$.
Furthermore, 
\[
\vartheta(M) = \lim_{\epsilon \to 0^+}\frac{1}{2\pi \mathrm{i}}\left(\frac{1}{Q(-M+\mathrm{i}\epsilon)}
-\frac{1}{Q(-M-\mathrm{i}\epsilon)} \right)= \frac{1}{2\pi} \frac{2 \Im Q }{(\Re Q)^2 +|\Im Q|^2 }.
\]
Notice that the Fourier-Laplace transform of the retarded propagator of mass square equal to $M$, with our sign conventions, is proportional to $1/(-w^2-M)$.

In \Eq~\eqref{eq:retarded-propagator-full}, the first contribution, due to the zeros of $Q$ in $\mathbb{C} \setminus (-\infty, -4m^2)$, corresponds to the pole part of the retarded operator. The second contribution comes from the branch cut, which lies on the interval $(-\infty, -4m^2) \subset \mathbb{C}$ as shown in the domain of $Q$ on the first Riemann sheet. Notice that the sum over the pole part corresponds to a sum over all zeros of $Q$. This sum is usually over three zeros and if the zeros are contained on the real line, it gives origin to a well-defined retarded operator. 

There are, however, cases for the choices of the involved parameters and of the mass and the coupling constant to the scalar curvature, where some of the zeros of $Q$ fall into the cut and so only two zeros survive. 

With this formula at disposal, 
the decay of $\phi$ for large times can be analysed as in \cite{MedaPinStab}. In particular, we obtain the following result:
\begin{theorem}\label{thm:stability}
    Consider $\mathcal{Z}$, the set of zeros of $Q(w^2)$ given in \eqref{eq:Q}. If $\mathcal{Z} \subset(-4m^2,0)$ and if $S\in C^\infty_{\mathrm{c}}(\mathcal{M})$, the solution \eqref{eq:retarded-propagator-full} of \eqref{eq:semiclassical-prototype} are bounded in time. Furthermore, they decay at least as $t^{-3/2}$  for large time. 
    If $\mathcal{Z} \subset(-4m^2,\infty)$, with some, $\gamma\in\mathcal{Z}$, $\gamma>0$, the solution is exponentially growing and the fastest growth of the components of $\phi$ is governed by the largest element of $L\in \mathcal{Z}$.  
\end{theorem}
\begin{proof}
The proof is a consequence of the decay properties of the retarded operator for the Klein-Gordon equation of arbitrary square mass on Minkowski spacetime. 

We recall that, from \eqref{eq:retarded-propagator-full}, if $S$ has compact support, for $t$ larger than any time of the points in the support of $S$, then
\begin{align*}
\phi(t,\mathbf{x}) = &\sum_{\gamma \in \mathcal{Z}}
\frac{1}{Q'(-\gamma)} 
 \int 
 \left(
 e^{\mathrm{i} t w_\gamma} e^{\mathrm{i} \mathbf{x} \mathbf{p}}
\frac{\hat{S}(w_\gamma,p)}{2 w_\gamma \mathrm{i}}
-
e^{-\mathrm{i} t  w_\gamma} e^{\mathrm{i} \mathbf{x} \mathbf{p}}
\frac{\hat{S}(-w_\gamma,p)}{2 w_\gamma \mathrm{i}}
\right)\, d^3\mathbf{p}
\\
&+ 
\int_{4m^2}^\infty
\vartheta(M)
 \int \left(e^{\mathrm{i} t w_\gamma} e^{\mathrm{i} \mathbf{x} \mathbf{p}}
\frac{\hat{S}(w_M,p)}{2 w_M \mathrm{i}}
-
e^{-\mathrm{i} t  w_M} e^{\mathrm{i} \mathbf{x} \mathbf{p}}
\frac{\hat{S}(-w_M,p)}{2 w_M \mathrm{i}}\right)\,
d^3\mathbf{p} \, dM
,\quad w_M=\sqrt{\mathbf{p}^2+M}\, ,
\end{align*}
where $\hat{S}$ denotes the spatial and temporal Fourier transform of $S$. Hence, if $\mathcal{Z} \subset(-4m^2,\infty)$, considering the integral over the momenta and using standard stationary phase methods to estimate the decay in time (see e.g. Lemma A.1 in \cite{MedaPinStab} for a proof of that decay), we see that for large values of $t$,
\[
|\phi(t,\mathbf{x})| \leq \frac{C}{t^{\frac{3}{2}}}\,,
\]
where the constant $C$ depends on the particular form of $Q$ via
\[
C = c \left(
\sum_{\gamma\in\mathcal{Z}} \frac{\sqrt{|\gamma|}}{|Q'(-\gamma)|}
\left(
 |\hat{S}(\sqrt{\gamma},0)|
 +
  |\hat{S}(-\sqrt{\gamma},0)|
\right)
+
\int_{4m^2}^\infty 
|\vartheta(M)| \sqrt{M} 
\left(
 |\hat{S}(\sqrt{M},0)|
 +
  |\hat{S}(-\sqrt{M},0)|
\right)
dM 
\right)
\]
for some $c>0$ independent on $Q$. The integral over $M$ can be taken because $\hat{S}(w,0)$ decays rapidly on $w$ and $\vartheta(M)$ is an integrable function.

If now $\mathcal{Z} \subset(-4m^2,\infty)$ 
and if there is $\gamma\in\mathcal{Z}$ with $\gamma>0$, the continuous part contribution to $\phi(t,\mathbf{x})$ still decays as $1/t^{3/2}$ for large $t$. However, the pole part contains contributions that grow exponentially in time. The largest exponential growth (bounded by $c e^{\sqrt{L}t}$) corresponds to the zero momenta contribution and to the largest zero $L>0$ in $\mathcal{Z}$. 
\end{proof}

\section{Linear instability of the semiclassical Einstein equations}\label{sec:5}

In this section, we discuss the existence of solutions of \eqref{se:prototype-and-solution}, applying the analysis of the 
prototypical equation \eqref{eq:semiclassical-prototype} and the study of the stability of its solutions 
presented in the previous section. In particular, we shall analyse the 
equation governing the S and TT modes given in \Eq~\eqref{eq:semiclassical-sources}.

As discussed in Theorem \ref{thm:convergence} and more precisely in 
Proposition \ref{prop:hypo-inveresion}
and \ref{prop:normal-form},
the  particular form of the solutions depends on the parameters $a_i$ and $b_j$ in \Eq~\eqref{eq:semiclassical-prototype}.
Furthermore, according to Theorem \ref{thm:stability}, 
the stability of the solutions is related to the set of zeros of $Q$ given in  \eqref{eq:Q}.
These zeros are also affected by the parameters $a_i$ and $b_j$ in \eqref{eq:semiclassical-prototype}.
In the case of semiclassical Einstein equations, the parameters $b_j$ are also modified by the choice of some renormalisation constants and some of them are fixed according to the discussion presented in Section \ref{se:renormalisation-constants}.

We shall thus analyse the form of the constants $a_i$ and $b_j$ for the case of the $S$ and $TT$ modes comparing the set of equations in \eqref{eq:semiclassical-sources} with the prototype equation studied in the previous section in \Eq~\eqref{eq:semiclassical-prototype}.
We shall furthermore analyse the asymptotic form of the solutions on the basis of physical parameters appearing in the semiclassical equation.

\subsection{The S modes}

In view of the form of the operators 
$\mathcal{S}$ given in \eqref{eq:S-T-operator},
we have that $a_1=a_2 = a =\frac{2m^2}{6\xi-1}$. 
In particular, let us recall here the form of the equation in \eqref{eq:semiclassical-sources} for the scalar modes:
\begin{equation}\label{eq: 5DperS}
    \frac{6(\frac{1}{6} - \xi)^2}{4}\iota_{\mathcal{M}} \tilde{\mathsf{G}}_{\mathrm{ret}}\left( \square \left( \square - \frac{2m^2}{6\xi - 1} \right)^2 \overline{h}^{\mathrm{S}}_{ab}  \otimes \check\rho \right)  - {\tilde{\alpha}}^{\mathrm{S}}_1 m^4 \overline{h}^{\mathrm{S}}_{ab} - \frac{1}{2}\left( \frac{1}{\kappa} - {\tilde{\alpha}}^{\mathrm{S}}_2 m^2 \right) \square \overline{h}^{\mathrm{S}}_{ab} - {\tilde{\alpha}}^{\mathrm{S}}_3 \square \square \overline{h}^{\mathrm{S}}_{ab}  = S_{ab}^{\mathrm{S}}\,,
\end{equation}
where we used the result in Corollary \ref{cor: formaK0} to write the term $\mathcal{K}_0(\overline{h}^{\mathrm{S}}_{ab})$ in $\mathcal{S}$. As a first remark, we notice that, in order to fulfil the hypotheses of Proposition \ref{prop:inverse-log}, we need to choose $\xi$ in such a way that 
\[
\frac{2}{6\xi-1} < 4\,.
\]
Finally, assuming that ${\tilde{\alpha}}_2^{\mathrm{S}}=0$ so that no renormalisation of the Newton constant is taken into account and considering the result of Proposition \ref{prop:background} and 
Theorem \ref{thm:fixing-alpha1} to fix 
${\tilde{\alpha}}^{\mathrm{S}}_1 = (64 \pi)^{-1}$,
we have by comparing term by term \eqref{eq: 5DperS} with \eqref{eq:semiclassical-prototype} that the parameters in the $S$ case become:
\begin{equation}
    \label{eq:S-coeff}
        a_1=a_2=a = \frac{2m^2}{6\xi-1}, \quad b_0=- {\tilde{\alpha}}^{\mathrm{S}}_1 \frac{ 4m^4 }{6\left(\frac{1}{6}-\xi\right)^2}, \quad  b_1=-\frac{2}{\kappa} \frac{1}{6\left(\frac{1}{6}-\xi\right)^2}
\end{equation}
and $b_2$ is a free parameter which depends on the renormalisation constant ${\tilde{\alpha}}^{\mathrm{S}}_3$. Recall that $\kappa = 8 \pi G$.
The study of existence of solutions and their  asymptotic behaviour at large times 
depends on the form of zeros of 
\begin{equation}\label{eq:S-equation}
F^{\mathrm{S}}(\gamma)=\gamma(a-\gamma)^2 J(\gamma)-\left({b_0 +b_1\gamma+b_2\gamma^2}\right).
\end{equation}
Notice that this function equals $Q(-w)$ given in \eqref{eq:Q} for $\xi=-w^2$ and for the set of parameters studied in this section. 
Depending on the mass $m$ and on the coupling to the scalar curvature $\xi$,
there are two negative numbers $q_2<q_1<0$ such that, 
if $b_2<0$ the function $F^{\mathrm{S}}$ in \eqref{eq:S-equation} admits three or two zeros in $D =\mathbb{C}\setminus [4m^2,\infty)$.
If $b_2 <q_1$, these zeros are all real valued. 
If $b_2 \in (q_2,q_1)$ there are three zeros, they are real valued and 
in this case the solutions of the semiclassical equation for the S modes
satisfy a local equation of the form given in  \eqref{eq:semi-local}. The corresponding parameters $\beta_i$ can be found 
following Proposition \ref{prop:gamma}, in particular, the square masses of the component of the solutions of that equation coincide with the zeros $\gamma_i\in D$ of $F^{\mathrm{S}}$.

If $b_2\leq q_2$, \Eq \eqref{eq:S-equation} admits only two solutions in $D$. In this case, the existence of the solution is obtained applying Theorem \ref{thm:convergence}. 

Furthermore, we observe that in all these cases, there is a zero $\gamma_0$ of $F^{\mathrm{S}}$ that it is real valued and that is slightly smaller than $0$.
Its magnitude is governed by the value of $b_0$ and thus by the form of  
${\tilde{\alpha}}^{\mathrm{S}}_1$, which is fixed with Theorem \ref{thm:fixing-alpha1}.
For sufficiently large $-b_2$, this is the only negative zero that can be found in $D$.

According to Theorem \ref{thm:stability}, the presence of a negative zero implies that the solution of the equation for the S modes are unstable and the exponential growth of these unstable solutions is governed by the magnitude of the largest negative zero in the sense that 
\[
e^{-\sqrt{|\gamma_0|}t}h_{ab}^{\mathrm{S}}(t)
\]
is bounded for large times. 

To make the previous discussion clearer, we refer to Figure \ref{fig:S-modes}, which displays the qualitative form of the zero set of the real and imaginary parts of $F^{\mathrm{S}}$ for various choices of $b_2$.
We see, in particular, that there is a real valued zero $\gamma_0$ that is slightly smaller than $0$ and we see that for smaller values of $b_2$, there are two complex-valued and complex conjugate solutions with positive real part. As soon as $b_2$ increases, these two zeros are closer and closer to the real line and for larger values of $b_2$, two positive real-valued solutions appear. 
For even larger $b_2$, a solution falls into the cut, namely outside $D$.

\begin{figure} 
\includegraphics[width=0.45\textwidth]{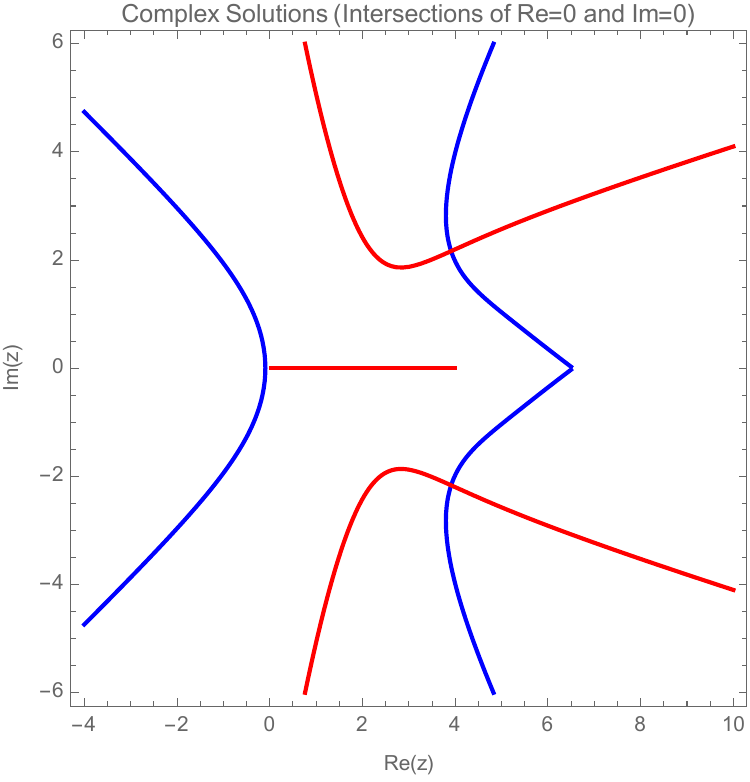} 
\includegraphics[width=0.45\textwidth]{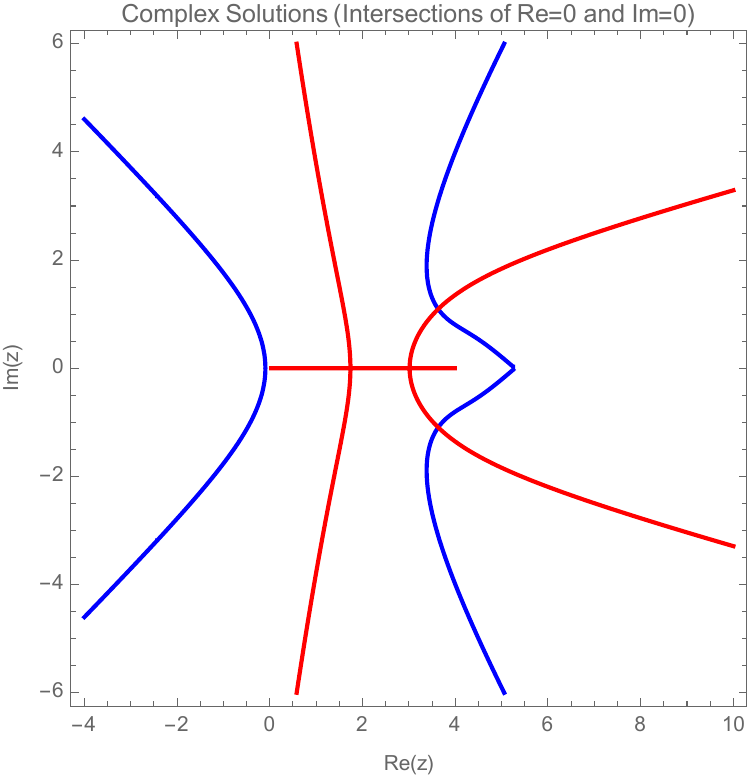}

\includegraphics[width=0.45\textwidth]{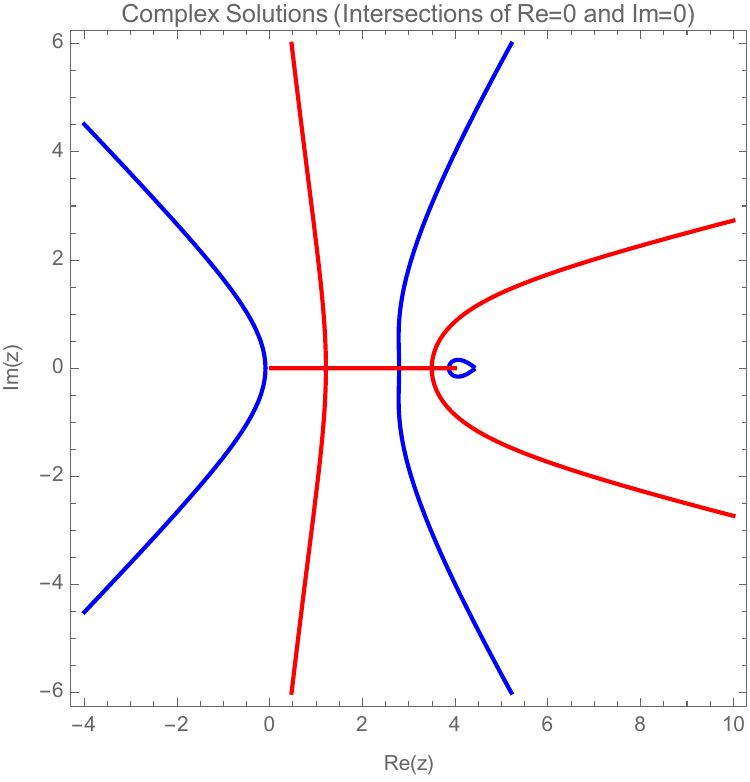}
\includegraphics[width=0.45\textwidth]{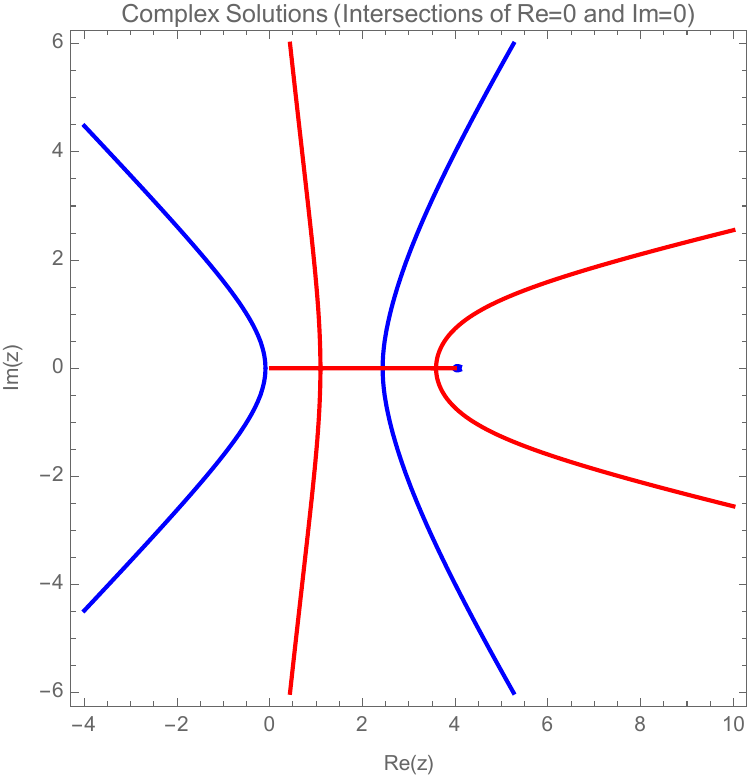}
\captionof{figure}{The blue line denotes the zero set of the real part of $F(z)$, while the red line denotes the zero set the imaginary part of $F$ for $a=-1$, $b_0=-1$, $b_2=-10$ and $b_2$ respectively in $\{3,4,5,5.4\}$}
\label{fig:S-modes}
\end{figure}

\subsection{The TT modes}
We now repeat the analysis for the TT modes. Starting from the form of the operator $\mathcal{T}$ given in \eqref{eq:S-T-operator}, we also have in this case that $a_1=a_2 = a =4m^2$. Recalling \Eq \eqref{eq:semiclassical-sources}, we write it in a form that is directly comparable with \eqref{eq:semiclassical-prototype}:
\begin{equation*}
    - \frac{1}{120} \iota_{\mathcal{M}} \tilde{\mathsf{G}}_{\mathrm{ret}}\left( \square \left( \square - 4m^2 \right)^2 \overline{h}^{\mathrm{TT}}_{ab}  \otimes \check\rho \right)  - {\tilde{\alpha}}^{\mathrm{TT}}_1 m^4 \overline{h}^{\mathrm{TT}}_{ab} -\frac{1}{2} \left( \frac{1}{\kappa} - {\tilde{\alpha}}^{\mathrm{TT}}_2 m^2 \right)\square \overline{h}^{\mathrm{TT}}_{ab} -{\tilde{\alpha}}^{\mathrm{TT}}_4 \square \square \overline{h}^{\mathrm{TT}}_{ab} = S^{\mathrm{TT}}_{ab}\,,
\end{equation*}
where we used the result in Corollary \ref{cor: formaK0} to write the term $\mathcal{K}_0(\overline{h}^{\mathrm{TT}}_{ab})$ in $\mathcal{T}$.
Assuming again that ${\tilde{\alpha}}_2^{\mathrm{TT}}=0$ so that there is no renormalisation of the Newton constant, and setting ${\tilde{\alpha}}^{\mathrm{TT}}_1=0$ following Theorem \ref{thm:fixing-alpha1},
we have that 
\begin{equation*}
    a=4m^2 , \qquad b_0=0, \qquad b_1=\frac{60}{\kappa}\,,
\end{equation*}
while $b_2$ is a free parameter proportional to ${\tilde{\alpha}}^{\mathrm{TT}}_4$. In this specific case, since $b_0=0$, there is no need to find $\mathfrak{h}_0$ in obtaining a local form of the semiclassical Einstein equation. Therefore, the problem of determining
the form of the solutions of the semiclassical equation for the TT modes and to analyse their stability gets slightly simplified.
Similarly to the analogous function for the S modes given in \eqref{eq:S-equation},
we thus study the zeros of the function
\begin{equation}\label{eq:TT-equation}
F^{\mathrm{TT}}(\gamma)=
\gamma (a-\gamma)^2 J(\gamma)-\gamma({b_1+b_2\gamma}).
\end{equation}
We recall that according to Proposition \ref{prop:gamma}, if we find 
three independent zeros. The TT modes solve a local equation whose solutions is formed by a superposition of solutions of the ordinary Klein-Gordon equations whose square masses coincides with the zeros. 
Zero is always one of these zeros of $F^{\mathrm{TT}}$ in \eqref{eq:TT-equation} and the corresponding TT modes correspond to the classical gravitational waves. 

We see from a qualitative analysis that if $m$ is sufficiently large, and for some choices of large and negative $b_2$, independent positive semidefinite real-valued zeros exists. 
This is e.g. the case for the choice of the parameters $m=1,a=4,b_0=0,b_1=2.5,b_2=-0,7$, where one zero is $0$ and there are other two real valued positive zeros.
However, in this regime of the parameters the mass of the field is chosen to be of the order of Planck mass. 

For more physically reasonable choices of the mass, and if $b_2$ is chosen to be positive and large, the zero $\gamma=0$ corresponding to classical gravitational waves is always there. 
Furthermore, one of the zeros of $F$ falls into the cat, hence in this case the existence of solutions is provided by Theorem \ref{thm:convergence}. 
One of the zeros, say $\gamma_0$, is, however, always negative. According to Theorem \ref{thm:stability}, this corresponds to an unstable solution. However, we observe that the magnitude of $\gamma_0$ can be made as small as we want by choosing larger and larger positive $b_2$.

\subsubsection{Linear instability}

For physical choices of the masses, there are instabilities both in the TT and S modes, which can be seen as exponential growths for large times.
All these instabilities can be removed by a conformal transformation with a conformal factor
$\Omega = e^{-Ht}$ with a parameter $H$, which is related to the largest negative zeros of \eqref{eq:S-equation} and of \eqref{eq:TT-equation}.

We see that for physically reasonable choices of the mass $m$, if ${\tilde{\alpha}}^{\mathrm{S}}_3$ is sufficiently large and negative and if 
${\tilde{\alpha}}^{\mathrm{TT}}_4$
is sufficiently large and positive, the parameter $H$ is universal. It actually depends on ${\tilde{\alpha}}^{\mathrm{S}}_1$, namely on the renormalisation of the cosmological constant fixed in Theorem \ref{thm:fixing-alpha1}.

These universal exponential growing (depending on the renormalisation constant and not on the particularly chosen solution)
suggests that the instabilities seen on a Minkowski background might be reabsorbed in a redefinition of the background solution.
The suggested background is the cosmological patch of the de Sitter spacetime with Hubble parameter $H$.

This behaviour differs in a significant way from the analysis presented in Anderson et al. in \cite{Anderson}. This different behaviour is originated by the choice of the renormalisation parameters done in Theorem \ref{thm:fixing-alpha1} and fixed considering the background to be a solution of the semiclassical Einstein equations and applying the principle of general local covariance \cite{BFV}. 

\subsection{Comparison with cosmological observations}

    To estimate the value of $H$ associated with large $b_2$, let us consider the equation of the masses of the perturbation of the S modes given in \Eq \eqref{eq:S-equation}. We are interested in the solution corresponding to a negative mass square that is typically close to $0$, and we want to quantitatively estimate this closeness.

    To obtain that particular solution, we analyse the function $F^{\mathrm{S}}$ given in \eqref{eq:S-equation}. We observe that to estimate the negative zeros of that function, the contribution due to the complex valued function $J(z)$ in $F^{\mathrm{S}}$ is negligible.
    Hence, to find the negative zeros we study the equation
\[
    b_0 + b_1 \gamma + b_2 \gamma^2 = 0,
\]
where the coefficients $b_i$ were defined in \Eq \eqref{eq:S-coeff}. In Planck units, $\kappa^{-1}$ is the squared value of the reduced Planck mass $M_P$, and thus $b_1$ is much larger than the remaining renormalisation constants. Hence, the mass parameter associated to the unstable solution may be approximated by
\[
    \tilde{\gamma} \simeq - \frac{b_0}{b_1} = - 16 \pi G {\tilde{\alpha}}^{\mathrm{S}}_1 m^4
\]
and it does not depend on the choice of ${\tilde{\alpha}}^{\mathrm{S}}_3$ and $\xi$. Thus, the effective Hubble parameter of the predicted de Sitter spacetime reads $H = \sqrt{-\tilde \gamma}$, and hence we obtain the following value of the cosmological constant in Planck units according to the (semiclassical) Friedmann equations
\[
    \Lambda = 3 \Omega_\Lambda H^2 = 6 \Omega_\Lambda{\tilde{\alpha}}^{\mathrm{S}}_1 \at \frac{m}{M_P} \ct^4 M_P^2,
\]
where $\Omega_\Lambda$ is the density parameter associated with the cosmological constant in the $\Lambda$-CDM model. Comparing this theoretical result with the experimental values of $\Omega_\Lambda$ and $\Lambda$ inferred from the Planck measurements of the Cosmic Microwave Background \cite{Planck}, i.e., $\Omega_\Lambda \simeq 0.685$ and $\Lambda \simeq 7.15 \times 10^{-121} M_P^2$, we obtain an approximated value of the mass of the field, that is $m \simeq 7.8 \times 10^{-3}$ eV.

Indeed, this numerical value may align, e.g., with the current results about the masses of light neutrinos, and the predicted mass of the axion derived from cosmological observations within the so-called ALPs Dark Matter models, where an axion-like particle is considered a primary candidate to explain the Dark Matter component of our universe (see, e.g., \cite{AxionDM,AxionCosmo,PDG} and references therein). In this viewpoint, and according to our analysis, the hypothetical Dark Matter particle gives rise to the component of Dark Energy appearing in the semiclassical Einstein equations in the guise of a cosmological constant; thus, it drives the expansion of the universe at cosmological scales through its backreaction upon the background geometry.

\appendix

\section{Stress-energy tensor renormalisation}\label{app:reno-freedom}
Consider a generic globally hyperbolic spacetime $(\cM_n,g_{ab})$, where $n > 2$ is the spacetime dimension. Consider on it a Klein-Gordon field $\phi$ of mass $m$, coupling to the scalar curvature $\xi$ and endowed with a quantum Hadamard state $\omega$. According to the description given in \cite{HW05}, the stress-energy tensor is 
\[
    \Wick{T_{ab}} = 
    \frac{1}{2} \nabla_a\nabla_b \Wick{\Phi} -\Wick{\bar{\Phi}_{ab}} - \frac{1}{2}g_{ab} \left(\frac{1}{2}\square_g \Wick{\Phi}+m^2\Wick{\Phi} \right) + \xi\left(G_{ab}-\nabla_a\nabla_b+g_{ab}\square_g\right)\Wick{\Phi}
\] 
where $\square_g:= g^{ab} \nabla_a \nabla_b$ denotes the \emph{d'Alembertian} on $(\cM_n,g_{ab})$ here, while $\Wick{\Phi}$ and $\Wick{\bar{\Phi}_{ab}}$ respectively are the Wick ordering of $\Phi=\phi^2$, $\Phi_{ab} = \phi\nabla_a\nabla_b\phi$, and $\bar{\Phi}_{ab}$ is the trace reversal of $\Phi$. Their expectation values are such that 
\[
    \omega(\Wick{\Phi}) = [\mathsf{w}]+C, \qquad \omega(\Wick{\bar{\Phi}_{ab}})= [\nabla_a\nabla_b \mathsf{w}] - \frac{1}{2} g_{ab} [\square_g \mathsf{w}]+\bar{C}_{ab}
\]
with $\mathsf{w} = \omega_2-\mathcal{H}$, for $\omega_2$ the two-point function of the state, $\mathcal{H}$ the Hadamard parametrix \eqref{eq:Hadamard-parametrix} and $[\cdot]$ denoting the coinciding point limit. The scalar $C$ and the tensor $\bar{C}_{ab}$ take into account the renormalisation freedoms and the anomalous contribution. They are obtained by local tensors constructed out of the metric and the mass, and must have the correct behaviour under rescaling of the spacetime metric and of the mass of the field. Finally, they are chosen in a way that the resulting stress-energy tensor is covariantly conserved. According to \cite{HW05}, these requirements do not fix completely $C$ and $\bar{C}_{ab}$ but leave a residual freedom. Their precise form is
\begin{align*}
    C &= d_1 m^2 + d_2 R, \\
    \bar{C}_{ab} &=  \frac{d_2}{2} R G_{ab} + d_2 \left( \frac{\square_g R}{4}  +\frac{1}{6}(1-2\xi)R^2 - \frac{m^2 R}{2}\right) g_{ab} -\frac{n}{2(n+2)} A g_{ab} \\
    &\quad + c_1 m^4 g_{ab} + c_2 m^2 G_{ab} + c_3 I_{ab} + c_4 J_{ab},
\end{align*}
where $d_1,d_2 \in \mathbb{R}$ and $c_1,c_2,c_3,c_4 \in \mathbb{R}$ are renormalisation constants, whereas $\frac{n}{2(n+2)}Ag_{ab}$ corresponds to the anomalous contribution which was given in \Eq \eqref{eq:Anomaly} for $n=4$. In particular, the form of $\bar{C}_{ab}$ guarantees that the resulting stress energy tensor is covariantly conserved.\\

Also, as it is of interest for the content of this paper, we recall in this appendix the explicit form of the tensor $I$ and $J$ appearing in \Eq \eqref{eq:T_ab-omega}. The validity of \Eq \eqref{eq:line-hbar} implies that their linear order contributions, in terms of the trace reversal metric perturbation $\bar{h}$ over the Minkowski background are
\begin{equation} 
\label{eq:def-I-J}
    \begin{aligned}
        J_{ab} &= 2\nabla_a \nabla_b R -2 g_{ab} \square_g R + \frac{1}{2} g_{ab} R^2 -2 R R_{ab} \\
        &= 2\nabla_a \nabla_b R -2 g_{ab} \square_g R + O(2) \\
        &= \nabla_a \nabla_b \square{\bar{h}_c}^c - \eta_{ab} \square\square {\bar{h}_c}^c + O(2) \\
        &= - \tau_{ab} \square\square {\bar{h}_c}^c + O(2) \\
        I_{ab} &= -2\square_g R_{ab} + \frac{2}{3} \nabla_a \nabla_b R + \frac{1}{3} \square_g R g_{ab} - \frac{1}{3} R^2 g_{ab} + \frac{4}{3} R R_{ab} + g_{ab}R_{cd}R^{cd} -4 R_{abcd}R^{cd} \\
        &= -2\square_g G_{ab}   - \square_g R g_{ab} + \frac{2}{3} \nabla_a \nabla_b R + \frac{1}{3} \square_g R g_{ab} +O(2) \\
        &= \square\square \bar{h}_{ab}  - \frac{1}{3} \tau_{ab} \square \square {\bar{h}_c}^c + O(2).
    \end{aligned}
\end{equation}
Moreover, regarding the linearisation of the anomalous contribution to $\langle\Wick{T_{ab}}\rangle_\omega$ given in \Eq \eqref{eq:T_ab-omega} encoded in $\frac{A}{3}g_{ab}$, we observe that most of the terms appearing in $A$ are at least quadratic in the perturbation, and thus do not contribute at the linear order. On the other hand, the remaining linear contribution, which reads 
\[
    \frac{A}{3}g_{ab} = \left(\frac{(6\xi - 1)m^2}{192 \pi^2}  \square {\bar{h}_c}^c + \frac{(5\xi -1)\square \square {\bar{h}_c}^c}{960\pi^2} \right) \eta_{ab} + O(2)
\]
has been already taken into account in $N_{ab}$, see \Eq \eqref{eq:linearisedT1}. Indeed, the linearised expectation value of the stress-energy tensor \eqref{eq:linearisedT1} is covariantly conserved, and hence it includes the linearised anomalous contribution by construction.

\section{Past-compact perturbations, cutoffs, and sources}
\label{se:source-perturbation}

\subsection{First order two-point function and its description with a source} 
\label{se:source-omega1}

Consider $\omega_2^{(1)}(x,y)$ a smooth symmetric bi-solution of the free equation of motion on the Minkowski background, which describes the modification at first order of the two-point function of the considered state. We turn it into a past-compact perturbation of the state, by multiplying it with a suitable cutoff function $\chi: M\to \RR$. The latter is a positive, smooth function and depends on the spacetime points through the time function $t$. It is vanishing where $t<-\epsilon < 0$, namely in the past of $\Sigma_{-\epsilon}$, and it is equal to $1$ in the future of $\Sigma_0$. The corresponding past compact bi-distribution on $\mathcal{M} \times \mathcal{M}$ is
\[
\omega^{(1)}_{\chi,2}= \omega^{(1)}_2 \circ (\chi\otimes 1+1 \otimes \chi). 
\]
Notice that $\omega^{(1)}_{\chi,2}=\omega^{(1)}_2$ on $J^{+}(\Sigma_0)\times J^{+}(\Sigma_0)$ and 
$\omega^{(1)}_{\chi,2}$ vanishes on $J^{-}(\Sigma_{-\epsilon})\otimes J^{-}(\Sigma_{-\epsilon})$.
For this reason $\omega^{(1)}_{\chi,2}(T^{(0)}_{ab})=[D_{ab}\omega^{(1)}_{\chi,2}]$, the coinciding point limit of $D_{ab}\omega_2^{(1)}$ where $D_{ab}$ is given as in \eqref{eq:T_ab-omega}. Observe that 
\[
(\square -m^2)_y \omega^{(1)}_{\chi,2}(x,y) = -\omega^{(1)}_2\circ 1 \otimes \chi''
- 2\partial_{0,y}\omega^{(1)}_2\circ 1 \otimes \chi'
\]
and similarly for $(\square -m^2)_x \omega^{(1)}_{\chi,2}(x,y)$.
Consider now 
\[
S'_{ab}:=[D_{ab}\omega^{(1)}_{\chi,2}]  
\]
and observe that 
\[
\nabla^aS'_{ab}= [\partial_{b,x} (\square -m^2)_y\omega^{(1)}_{\chi,2}]  
=-[\partial_{b,x} \omega^{(1)}_2]  \chi''
- 2[\partial_{b,x}\partial_{0,y}\omega^{(1)}_2]\chi'
\]
where we used the equation of motion satisfied by $\omega^{(1)}_2$ in the second equality.
We thus have that $\nabla^aS'_{ab}$ is supported in the region of variation of $\chi$, corresponding with $J^{+}(\Sigma_{-\epsilon})\cap J^{-}(\Sigma_{0})$, and in general $S'_{ab}$ is not conserved. Hence, $S'_{ab}$ cannot be used as source for the semiclassical Einstein equation. However, by Proposition \ref{prop:decomposition} and Proposition \ref{prop:gauge-trasformation}, we can find a vector field $X_a$ which is past-compact and such that 
\[
S_{ab} := S'_{ab}+\partial_a X_b+\partial_b X_a
\]
is covariantly conserved. Therefore, it is this term which incorporates the first order variation $\omega^{(1)}_{\chi,2}$ and is covariantly conserved. Thus, following the analysis presented in the introduction, we shall use this as source in the semiclassical Einstein equations given in \eqref{eq:FPmetric}. We observe that, according to Proposition \ref{prop:decomposition}, the vector part of $S_{ab}$ vanishes. Hence, the transformation 
used to obtain $S_{ab}$ from $S'_{ab}$, corresponds to the vanishing of the vector part of $S'_{ab}$. Moreover, as $S_{ab}$ is computed without any information about the metric perturbation, can be regarded as a genuine source for the semiclassical Einstein equations. Finally, considering $\chi_\nu(t) \coloneqq \chi(t/\nu)$ in place of $\chi(t)$, in the limit $\nu\to \infty$, $X_a$ converges uniformly to $0$.

\subsection{Stress-energy tensor and independence from the time cutoff}
\label{se:source-cutoff}

To evaluate the relation in the expectation value of the stress-energy tensor on the background theory with respect to the foreground, we use a deformation argument. 
We are interested in the semiclassical Einstein equations in the region $J^+(\Sigma_0)$, the future of a fixed Cauchy surface $\Sigma_0 \subset \mathcal{M}$ constructed as the hypersurface at time $t = 0$ with respect to a chosen time function. We introduce a time cutoff $\chi: \mathcal{M} \to \RR$ which is positive, smooth and depends on the spacetime points through the time function $t$. Moreover, we demand it to be vanishing for $t<-\epsilon < 0$ and $1$ in the future of $\Sigma_0$. With this cutoff function at disposal, we identify the perturbed metric $g$ with the Minkowski background metric as follows:
\[
g_\chi=\eta+h \chi.
\]  
So, in the region where $\chi = 0$ it holds $g_{\chi}=\eta$ and, on $J^{+}(\Sigma_0)$ there is no dependence on $\chi$ as $g_\chi=g$. The goal of this appendix is to provide more details on the way in which the Bogoliubov map is used to obtain the linearised stress energy tensor for a past compact metric perturbation, and later discuss its independence from the choice of the cutoff allowing to restriction to past compact solutions of the semiclassical equations.\\

We treat changes in $\langle \Wick{T_{ab}}\rangle_\omega$ by means of the perturbative agreement and by the Bogoliubov map \cite{HW05,BDF,DHP,BDFR,Rejzner}
\[
\langle \Wick{T_{ab}}\rangle_\omega = (\omega^{(0)}+\omega^{(1)}) (\mathsf{R}_{V_\chi}(T_{ab}^{0}+T_{ab}^{1})) 
\]
where, as explained in Section \ref{sec3}
\[
\mathsf{R}_{V_\chi}(F)= -\mathrm{i} \left.\frac{d}{d\lambda} \mathsf{S}_V(\lambda F)\right|_{\lambda=0} = 
-\mathrm{i}\left.\frac{d}{d\lambda}\mathsf{S}( V)^{-1}\mathsf{S}(V+\lambda F)\right|_{\lambda=0},
\]
$\mathsf{S}(F)$ is the time ordered exponential of $\mathrm{i}F$
 and
\[
V_\chi = \mathcal{L}[g] - \mathcal{L}[\eta]  = \int \frac{1}{2}\phi(\mathcal{P}_{g_\chi})\phi d\mu_{g_\chi} -\int \frac{1}{2}\phi(\mathcal{P}_\eta)\phi  d\mu_\eta
\]
with $P_g = \square_g -m^2 -\xi R_g$. Notice that the cutoff is in the causal past of the region where the semiclassical Einstein equations are studied. The, expanding the Bogoliubov map perturbatively at first order, for every local field observable $F$, we have 
\[
\mathsf{R}_{V_\chi}(F) = F-\mathrm{i} \mathscr{T}\left(V_{\chi} F\right) + \mathrm{i} V_{\chi}  F +O(2)
\]
and so
\[
\omega (\mathsf{R}_{V_\chi}(F)) =\omega_0(F)- \mathrm{i}\omega^{(0)}(
\mathscr{T}(V_\chi   F) - V_\chi   F )  + \omega^{(1)}(F) + O(2).
\]

Let us now discuss how this expression depends on the choice of $\chi$. By the causal factorisation property, satisfied by $\mathsf{S}$ matrices and by the Bogoliubov map, changes in the cutoff $\chi$ which occur in the past of the region where we test our observables correspond to the adjoint action of a unitary operator. In particular, for every local observable $F$ with $\text{supp} F \subset J^+(\Sigma_0)$, if we consider another cutoff function $\chi'$ such that $\text{supp} (\chi-\chi')\subset J^{-}(\Sigma_{0})$,  
\begin{align*}
\mathsf{R}_{V_\chi}(F) 
&= \mathsf{S}_{V_{\chi'}}(V_{\chi-\chi'}) \mathsf{R}_{V_{\chi'}}(F) \mathsf{S}_{V_{\chi'}}(V_{\chi-\chi'})^{-1} 
\\
&= U_{\chi,\chi'} \mathsf{R}_{V_{\chi'}}(F) U_{\chi,\chi'}^{*}
\end{align*}
where we introduced the formally unitary operator
\[
U_{\chi,\chi'} = \mathsf{S}_{V_{\chi'}}(V_{\chi-\chi'}). 
\]
Notice that the action of this unitary operator can be reabsorbed in a redefinition of the state. Actually we might write
\[
\omega (\mathsf{R}_{V_\chi}(F)) = \omega (U_{\chi,\chi'} \mathsf{R}_{V_{\chi'}}(F) U_{\chi,\chi'}^{*}) = \tilde{\omega}_{\chi,\chi'}(\mathsf{R}_{V_{\chi'}}(F)).
\]
However, if we we now expand up to linear order in the perturbation the formal unitary operator
\[
U_{\chi,\chi'}=\mathsf{S}_{V_{\chi'}}(V_{\chi-\chi'}) =\mathsf{S}({V_{\chi'}})^{-1}\mathsf{S}(V_{\chi}) = 1+\mathrm{i} V_{\chi}-\mathrm{i}V_{\chi'} + O(2)
= 1+\mathrm{i} V_{\chi-\chi'} + O(2) 
\]
and therefore the change in the choice of the time cutoff for the state is expressed at linear order as
\[
\tilde{\omega}_{\chi,\chi'}(F) = {\omega}(U_{\chi,\chi'}F U_{\chi,\chi'}^*)  
= {\omega}(F) + \mathrm{i} \omega( [V_{\chi-\chi'}, F] ).
\]
This implies in particular that the state $\tilde{\omega}$ coincides with the state $\omega$ at zero order and their difference appears only at first order
\[
\tilde{\omega}_{\chi,\chi'}^{(0)}(F) =\omega_0(F) ,\qquad   \tilde{\omega}_{\chi,\chi'}^{(1)}(F) = \omega^{(1)}(F) + \mathrm{i}\omega^{(0)}([V_{\chi-\chi'}, F]).
\]
Finally, as $V_{\chi-\chi'}$ is supported in the past of $\Sigma_0$ and 
vanishes in the region where $\chi-\chi' =0$, this implies that $\tilde{\omega}_{\chi,\chi'}^{(1)} - \omega^{(1)}$ is also past-compact. As such and combining it with the argument in Appendix \ref{se:source-omega1}, considering a generic past compact source gives independence from the choice of the cutoff.

 
\renewcommand{\thefootnote}{\arabic{footnote}}
\setcounter{footnote}{0}

\end{document}